\documentclass[%
 reprint,
unsortedaddress,
nofootinbib,
 amsmath,amssymb,
 superscriptaddress,
 nobibnotes,
 aps,
 pra,
 twocolumn
]{revtex4-2}

\newif\ifcomments

\commentsfalse

\usepackage{iftex}
\ifPDFTeX
  \usepackage[utf8]{inputenc}
  \usepackage[T1]{fontenc}
\fi
\ifLuaTeX
  \usepackage{luatex85}
  \usepackage[noTeX]{mmap}
\fi

\usepackage{comment}
\usepackage{amsfonts}
\usepackage{amsmath}
\usepackage{mathtools,amsthm,amssymb} %
\usepackage{xcolor,xstring,xparse}
\usepackage{graphicx}
\usepackage{comment}
\definecolor{linkblue}{HTML}{001487}
\usepackage[colorlinks=true,allcolors=linkblue]{hyperref}
\usepackage[nameinlink,capitalize,noabbrev]{cleveref}
\usepackage{braket}
\usepackage{tikz}
\usepackage{qcircuit}
\usepackage{enumitem} %

\usepackage{adjustbox}
\usepackage{siunitx}
\usepackage{float}
\usepackage{amsthm}
\usepackage{tikz-cd}
\usepackage{tikz}
\usetikzlibrary{backgrounds}
\usepackage{hyperref}
\usepackage{thmtools}

\hypersetup{
    colorlinks = true,
    linkcolor = violet,
    linkbordercolor = {white}
}

\usepackage{booktabs}    
\usepackage{tabularx}    

\usetikzlibrary{fadings}
\usetikzlibrary{patterns}
\usetikzlibrary{shadows.blur}
\usetikzlibrary{shapes}

\ifcomments
  \newcommand{\jon}[1]{{\color{blue}\bf (Jonathan: #1)}}
  \newcommand{\era}[1]{\textit{\textcolor{cyan}{(#1 --ERA)}}}
\else
  \newcommand{\jon}[1]{}
  \newcommand{\era}[1]{\textit{\textcolor{cyan}{}}}
\fi

\theoremstyle{plain}
\newtheorem{theorem}{Theorem}[section]
\newtheorem*{theorem*}{Theorem}
\newtheorem{lemma}[theorem]{Lemma}

\Crefname{claim}{Claim}{Claims}

\newtheorem*{lemma*}{Lemma}

\newtheorem{corollary}[theorem]{Corollary}
\newtheorem*{corollary*}{Corollary}
\newtheorem{proposition}[theorem]{Proposition}

\newtheorem{question}[theorem]{Question}

\theoremstyle{definition}
\newtheorem{definition}[theorem]{Definition}
\newtheorem{assumption}[theorem]{Assumption}

\theoremstyle{remark}
\newtheorem{remark}[theorem]{Remark}
\numberwithin{equation}{section}

\renewcommand{\bm}[1]{\boldsymbol{#1}}

\newcommand{\bounds}[2]{\bigg\rvert_{#1}^{#2}}
\newcommand{\boudns}[2]{\bounds} 

\newcommand{\e}{\epsilon}

\newcommand{\checkwt}{\text{wt}_{\text{c}}}
\newcommand{\bitwt}{\text{wt}_{\text{b}}}

\newcommand{\ce}{\mathrm{e}}
\newcommand{\ci}{\mathrm{i}}
\newcommand{\cpi}{\mathrm{\pi}}
\DeclareMathOperator{\Tr}{Tr}

\begin{abstract}
Quantum algorithms are believed to offer advantages in solving certain hard discrete optimization problems, yet identifying when such advantages persist in explicit distributions of problem instances remains a foundational challenge.
Recently, a new quantum algorithm known as Decoded Quantum Interferometry (DQI) has been proposed to solve optimization problems by  decoding a corresponding LDPC error-correcting code.
Although DQI exhibits quantum advantage on certain structured problem instances, the possibility for advantage on random, unstructured problem instances is less well-understood.
Here we prove that, assuming decoding threshold upper bounds satisfied by state-of-the-art decoders, DQI is asymptotically obstructed by a spin glass phase transition in random local combinatorial optimization problems. This phase transition is heralded by the onset of the overlap gap property (OGP), a topological fragmentation of the near-optimal solution space widely conjectured to \emph{exactly} characterize the asymptotic performance of optimal efficient classical algorithms.
Our results therefore indicate that DQI, applied on the best known efficient decoders, is unlikely to exhibit quantum advantage on unstructured problem instances.
We support this result by proving that approximate message passing, a classical optimization algorithm, outperforms DQI on certain problem distributions.
\end{abstract}

\begin{document}

\title{
Spin Glass Transitions Obstruct Decoded Quantum Interferometry
}

\author{Eric R.\ Anschuetz}
\affiliation{Institute for Quantum Information and Matter, Caltech, 1200 E. California Blvd., Pasadena, CA 91125, USA}
\affiliation{Walter Burke Institute for Theoretical Physics, Caltech, 1200 E. California Blvd., Pasadena, CA 91125, USA}
\altaffiliation{Authors listed alphabetically}
\email{eans@caltech.edu}
\author{David Gamarnik}
\affiliation{Sloan School of Management, MIT, 100 Main St., Cambridge, MA 02142, USA}
\author{Jonathan Z.\ Lu}
\affiliation{Department of Mathematics, MIT, 77 Massachusetts Ave., Cambridge, MA 02139, USA}
\email{lujz@mit.edu}

\maketitle

\section{Introduction} \label{sec:intro}

\begin{figure*}[t!]
    \centering
    \includegraphics[width=\linewidth]{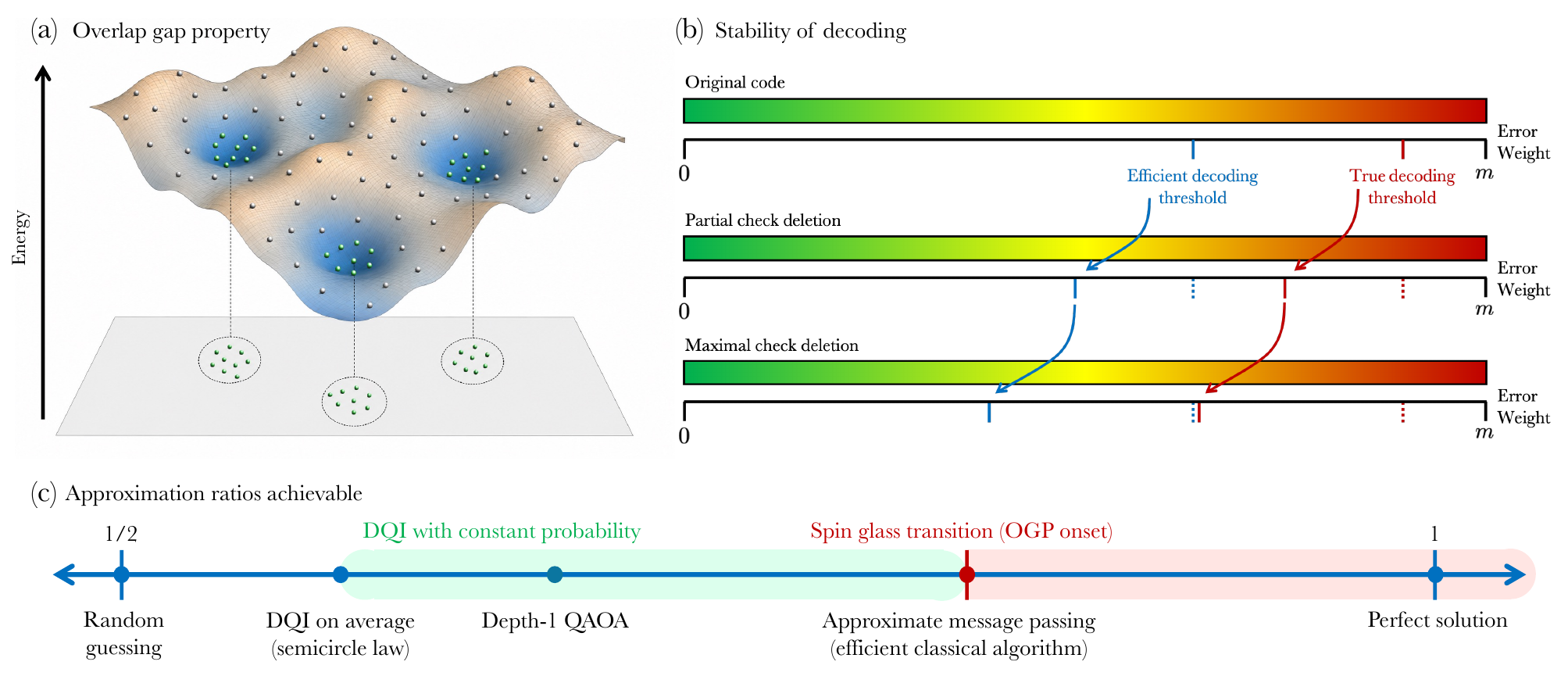}
    \caption{
    \textbf{(a)} The \emph{overlap gap property} characterizes spin glass phase transitions in combinatorial optimization problems, and can be viewed as the discrete analogue of many well-separated local optima in continuous optimization landscapes~\cite{gamarnik2023shatteringisingpurepspin,jones_et_al}. 
    We give the first proof that \textsc{MAX-$k$-XOR-SAT} exhibits an overlap gap property at constant $k$. 
    \textbf{(b)} We show that Decoded Quantum Interferometry (DQI) is unable to sample from low-energy configurations of \text{MAX-$k$-XOR-SAT} by connecting the performance of DQI to a property of LDPC codes that we call \emph{local restrictability}: good LDPC code ensembles which remain good when removing large fractions of checks. 
    If the best computationally efficient decoders have a threshold substantially below the true threshold, then we may remove a large fraction of checks without decreasing the threshold below that of computationally efficient decoders in the original code.
    This induces an effective locality which implies that DQI with an efficient decoder is unable to sample solutions in the spin glass phase with non-negligible probability, as the algorithm gets ``stuck'' in local optima. 
    \textbf{(c)} Comparison of the best approximation rates achievable on \textsc{MAX-$k$-XOR-SAT}; we show that DQI with efficient decoders lies in the green bubble, while a classical algorithm known as approximate message passing is known to exactly optimize to the spin glass phase transition on certain ensembles of problem instances.
    }
    \label{fig:ogp}
\end{figure*}

Quantum algorithms which are known or believed to give significant speedups, such as factoring integers or solving certain systems of equations, generally have significant and rigid algebraic structure~\cite{shor1999polynomial,gyurik2024quantum,berry2024analyzing,harrow2009quantum}.
At the same time, a tantalizing---but, as of yet not well understood---potential application for quantum algorithms is in \emph{discrete optimization}. 
In this setting, one is interested in computing
\begin{equation} \label{eq:approximate_optimization}
    \bm{z^\ast} = \underset{\bm{z}\in\left\{0,1\right\}^n}{\operatorname{argmax}} \; g\left(\bm{z}\right)
\end{equation}
given some objective function $g : \mathbb{F}_2^n \to \mathbb{R}$.

While complexity theory suggests that---at least in the worst case---optimizing a generic objective function $g$ to a sufficiently large constant approximation ratio is \textsc{NP}-hard, there is nevertheless hope of a \emph{quantum-classical approximation gap}, i.e., a setting in which the best approximation ratio $g_{\text{classical}}^*$ achievable by efficient classical algorithms is \emph{strictly smaller} than the best approximation ratio $g_{\text{quantum}}^*$ achievable by efficient quantum algorithms.
This possibility is particularly attractive in the ``average case,'' wherein the function $g$ is drawn from certain natural ensembles of problem instances. 
A priori, the existence of such a gap is not immediately ruled out by known results in complexity theory. It is a natural question, then, whether or not there are natural settings in which quantum algorithms yield average-case advantages in optimization over classical algorithms.

In recent years, the classical average-case hardness of this task has been studied using tools from a surprising place: statistical physics. This connection can be made more apparent by treating the function to be optimized, $g\left(\bm{z}\right)$, drawn from an ensemble of problem instances as a disordered classical spin system. The average-case classical hardness of optimizing $g\left(\bm{z}\right)$ beyond a value $g_c^\ast$ over this ensemble is widely conjectured to exactly coincide with the presence of a spin phase transition in the associated disordered system~\cite{doi:10.1073/pnas.2108492118,chen2019,gamarnik2022algorithmsbarrierssymmetricbinary,10.1214/23-AAP1953,10.1002/cpa.22222,cheairi2024algorithmicuniversalitylowdegreepolynomials}. Formally, this spin glass phase transition is defined by the onset of what is known as the \emph{overlap gap property} at a function value $g_{\mathrm{OGP}}^\ast$: a shattering of the space of near-optimal solutions into many well-separated clusters, an illustration of which is given in Fig.~\ref{fig:ogp}(a). The onset of this topological feature is rigorously known to inhibit classical algorithms which are Lipschitz functions of their inputs from preparing near-optimal solutions, and is also rigorously known to be a tight characterization of the value achieved by such algorithms in many settings~\cite{10.1002/cpa.22222,cheairi2024algorithmicuniversalitylowdegreepolynomials}. The onset of the OGP also coincides with the approximation ratios achieved by the best-known polynomial-time algorithms for random $k$-SAT instances, maximum independent set, and in the optimization of mean-field models, and for this reason is generally conjectured to be a tight characterization of the function value achieved by computationally efficient classical algorithms. We refer the interested reader to~\cite{doi:10.1073/pnas.2108492118} for more examples of problems for which the OGP is known to be tight with polynomial-time classical algorithms.

Conditioned on this well-established conjecture, then, the question of average-case quantum advantage for discrete optimization becomes:
\begin{question}
    For a given combinatorial optimization problem, is there a quantum algorithm that overcomes the spin glass (OGP) barrier?
\end{question}

Due to the ubiquity and widespread impact of combinatorial optimization problems across a plethora of fields and applications, much effort has been devoted to the development of quantum algorithms which can efficiently find solutions $\bm{z^\ast}$ whose achieved function value can exceed $g_{\mathrm{OGP}}^\ast$~\cite{farhi2014quantumapproximateoptimizationalgorithm,cerezo2021variational}.
Unfortunately, it is now known that many of these algorithms are unable to overcome the spin glass barrier~\cite{farhi2020quantumapproximateoptimizationalgorithm,anschuetz2022critical,anschuetzkiani2022,9996946,Anshu2023concentrationbounds,chen2023localalgorithmsfailurelogdepth,goh2025overlapgappropertylimits,anschuetz2025unified,anschuetz2025efficientlearningimpliesquantum}; for this reason, it seems unlikely that these will exhibit any quantum advantage for typical, unstructured optimization problems.

However, a new quantum algorithm for optimization has recently been developed that is fundamentally different from the others: \emph{Decoded Quantum Interferometry} (DQI), a quantum algorithm for \textsc{MAX-$k$-XOR-SAT}~\cite{jordan2025optimizationdecodedquantuminterferometry}. In \textsc{MAX-$k$-XOR-SAT}, one is given $m$ linear equations in $\mathbb{F}_2$ over $n<m$ variables, and is asked to find a variable assignment $\bm{z}\in\mathbb{F}_2^n$ satisfying as many clauses as possible. More concretely, a problem instance is encoded as an $m\times n$, $k$-row sparse constraint matrix $\bm{B}\in\mathbb{F}_2^{m\times n}$ as well as a vector of parities $\bm{v}\in\mathbb{F}_2^m$, and one hopes to maximize over $\bm{z}\in\mathbb{F}_2^n$:
\begin{equation}
    g\left(\bm{z}\right):=m-\left\lVert\bm{B}\bm{z}\oplus\bm{v}\right\rVert_1.
\end{equation}

DQI acts as a quantum reduction from instances of the \textsc{MAX-$k$-XOR-SAT} optimization problem to the bounded distance decoding of a classical linear code with parity check matrix $\bm{B}^\intercal\in\mathbb{F}_2^{n\times m}$. Unlike previous quantum heuristics, the performance of DQI can actually be \emph{guaranteed} given a classical decoder known to efficiently decode up to $\ell$ errors in an $m$-bit code. The expected fraction of satisfied clauses is given by a semicircle law~\cite{jordan2025optimizationdecodedquantuminterferometry}:
\begin{equation}\label{eq:dqi_semicircle_law}
    \frac{\left\langle g\right\rangle_{\mathrm{DQI}}}{m}=\left(\sqrt{\frac{\ell}{2m}}+\sqrt{\frac{1}{2}\left(1-\frac{\ell}{m}\right)}\right)^2.
\end{equation}
Armed with this performance guarantee, one can search for structured distributions over \textsc{MAX-$k$-XOR-SAT} for which the corresponding decoding problem is easy in order to find problems for which $\mu_{\mathrm{DQI}}$ is large. For instance, one can consider a special case of \textsc{MAX-$k$-XOR-SAT} called optimal polynomial intersection (OPI), where the decoding problem reduces to decoding a Reed--Solomon code~\cite{reed1960polynomial}. It is conjectured that DQI exhibits a quantum advantage on OPI~\cite{jordan2025optimizationdecodedquantuminterferometry}. If true, this conjecture remains consistent with the status quo intuition that quantum advantage necessitates algebraic structure.

However, if the \textsc{MAX-$k$-XOR-SAT} instance is instead chosen randomly from an unstructured ensemble, we cannot rely on these same techniques for decoding.
A crucial question then is whether one should expect a quantum advantage in this average case. In particular,
\begin{question}\label{q:dqi_ogp}
    Can DQI overcome the spin glass (OGP) threshold for unstructured \textsc{MAX-$k$-XOR-SAT} instances?
\end{question}

\subsection{Contributions}

Here, we answer Question~\ref{q:dqi_ogp} in the negative by showing that DQI \emph{cannot} surpass the OGP barrier. Specifically, we consider a definition of Lipschitz quantum algorithms that is due to \cite{anschuetz2025efficientlearningimpliesquantum}. It is defined with respect to the \emph{quantum Wasserstein distance} $d_{W_2}$~\cite{9420734}; informally, the quantum Wasserstein distance is an ``earth-mover's'' metric in that states which differ by a $p$-qubit quantum channel differ by at most $\operatorname{O}\left(p\right)$ in quantum Wasserstein distance.
\begin{definition}[Stable quantum algorithms, informal]\label{def:stab_qas_inf}
    Let $\bm{\mathcal{A}}$ be a quantum algorithm, i.e., a map from problem instances $\bm{X}=\left(\bm{B},\bm{v}\right)$ to quantum states $\bm{\mathcal{A}}\left(\bm{X}\right)$. We say $\bm{\mathcal{A}}$ is $\left(f,L,p_{\mathrm{st}}\right)$-\emph{stable} if there exist $L$ and a subset $S_{\bm{B}}$ of cardinality $f$ such that, with probability $1-p_{\mathrm{st}}$ over $\bm{X}$ and $\bm{X'}=\left(\bm{B'},\bm{v'}\right)$,
    \begin{equation}
        d_{W_2}\left(\Tr_{S_{\bm{B}}}\left(\bm{\mathcal{A}}\left(\bm{X}\right)\right),\Tr_{S_{\bm{B}}}\left(\bm{\mathcal{A}}\left(\bm{X'}\right)\right)\right)\leq L\left\lVert\bm{v}-\bm{v'}\right\rVert_1.
    \end{equation}
\end{definition}
We show that the OGP inhibits stable quantum algorithms for \textsc{MAX-XOR-SAT}. A similar statement was shown in \cite{anschuetz2025efficientlearningimpliesquantum} for quantum spin glass problems exhibiting a certain quantum modification of the OGP; we show that stable quantum algorithms are also obstructed by the classical OGP in classical combinatorial optimization problems.
\begin{theorem}[OGP inhibits stable quantum algorithms, informal]\label{thm:ogp_inhibits_stab_qas}
    Let $\mu_{\mathrm{OGP}}$ be the satisfied fraction at which the OGP occurs for \textsc{MAX-XOR-SAT}. Then, stable quantum algorithms can achieve no better than
    $g^\ast\leq\mu_{\mathrm{OGP}}m$ satisfied clauses with any constant probability.
\end{theorem}
We note that Theorem~\ref{thm:ogp_inhibits_stab_qas} also applies to quantum algorithms beyond DQI which were previously shown to be stable, including standard quantum algorithms such as log-depth QAOA and phase estimation~\cite{anschuetz2025efficientlearningimpliesquantum}.

We then focus on the setting where the LDPC parity check matrix $\bm{B}^\intercal$ is drawn from the standard Gallager ensemble of LDPC codes on $m$ bits with $k$-local checks and rate $r=1-n/m$. We consider the following assumption on the decoder DQI utilizes, which holds for all best-known efficient decoders for LDPC codes of which we are aware.
\begin{assumption}[DQI decoder conditions]\label{ass:dqi_dec_conds}
    The decoder satisfies one of the following two conditions:
    \begin{enumerate}
        \item The decoder is arbitrary, and is capable of correcting errors of relative weight at most:
        \begin{equation}\label{eq:rel_error_bound}
            \frac{\ell^\ast}{m}=\frac{1}{k\lambda},
        \end{equation}
        This is the best-known provably achievable decodable error weight achieved by efficient decoders on random LDPC codes~\cite{zyablov1975estimation,richardson2008modern,feldman2003using,koetter2005characterizations,ghazi2017lp}.
        We show explicitly that this bound holds for state-of-the-art exact decoders based on linear programming relaxations.
        \item The decoder satisfies a certain Lipschitz condition (Definition~\ref{def:l_inv_lip_dec}).
        We show that this class includes belief propagation (i.e., message passing) decoders in the standard density evolution model.
    \end{enumerate}
\end{assumption}
We prove that, under the Gallager ensemble of parity check matrices---a standard model for generic LDPC codes---and either of these two assumptions, DQI is stable according to Definition~\ref{def:stab_qas_inf}. In the case of the second assumption, our argument follows directly from the Lipschitzness of the utilized decoder and bypasses the choice of $\ell^*$ above. In the case of the first assumption, our argument for the stability of DQI follows from the observation that Gallager codes are what we call \emph{locally restrictable}: these are ensembles of codes which are still good codes when ignoring all but an $\epsilon$-fraction of syndromes. The stability of the DQI algorithm then follows as for any computationally efficient, nonlocal decoder utilized by the algorithm, there exists a (potentially computationally inefficient) which acts identically yet only utilizes a small fraction of the checks; see Fig.~\ref{fig:ogp}(b) for a summary. As stability depends only on the state prepared by the quantum algorithm this property is independent of the decoder actually used in the implementation of DQI, and in particular depends only on whether Assumption~\ref{ass:dqi_dec_conds} is satisfied.
\begin{theorem}[Stability of DQI, informal version of Theorems~\ref{thm:dqi_is_stable} and~\ref{thm:dqi_is_lipschitz}]
    Let $\bm{B}^\intercal$ be a parity check matrix drawn from an $\epsilon$-locally restrictable code family such as the Gallager ensemble. Assume the DQI decoder satisfies Assumption~\ref{ass:dqi_dec_conds}. Then, DQI is stable for this ensemble of \textsc{MAX-$k$-XOR-SAT} instances.
\end{theorem}

Next, we compute statistical properties of \textsc{MAX-$k$-XOR-SAT} at fixed $k$ and clause density $\lambda=m/n$, including the maximum achievable satisfied fraction as well as the overlap gap threshold. These are the first times these two quantities have been computed for \textsc{MAX-$k$-XOR-SAT} when $k$ is held fixed. Combining our results, we are able to show that DQI is topologically obstructed.
\begin{theorem}[DQI is topologically obstructed for \textsc{MAX-$k$-XOR-SAT}, informal version of Theorems~\ref{thm:dqi_fails_gallager} and~\ref{thm:bp_dqi_fails_ogp}]\label{thm:dqi_top_obs_intro}
    Assume the DQI decoder satisfies Assumption~\ref{ass:dqi_dec_conds}. DQI does not sample a bitstring $\bm{z}\in\mathbb{F}_2^n$ achieving a satisfied fraction more than $\mu_{\mathrm{DQI}}$, given by
    \begin{equation}
        \frac{1}{2}+\frac{1}{2\sqrt{\lambda}}\max\left(\left(1+\operatorname{o}_k\left(1\right)\right)\sqrt{\frac{2\ln\left(k\right)}{\ln\left(2\right)k}},\operatorname{\normalfont\textsc{P}}_k^{\mathrm{ALG}}\right)
    \end{equation}
    with any constant probability over both the randomness of the sampling and the distribution of problem instances. Here, $\operatorname{\normalfont\textsc{P}}_k^{\mathrm{ALG}}$ is the $k$th algorithmic Parisi constant~\cite{el2021optimization}, which we numerically fit:
    \begin{equation}
        \operatorname{\normalfont\textsc{P}}_k^{\mathrm{ALG}}\approx\sqrt{\frac{3.530\ln\left(k\right)}{k}}.
    \end{equation}
\end{theorem}
This topological obstruction limits the best performance DQI can achieve \emph{with substantial probability} and is thus a stronger upper bound than performance \emph{in expectation} given by the semicircle law.
We can compare the above bound with that of the semicircle law under the threshold from Eq.~\eqref{eq:rel_error_bound}, in the large $k$ limit:
\begin{align}
    \frac{\left\langle g\right\rangle_{\mathrm{DQI}}}{m} & = \left(\sqrt{\frac{\ell}{2m}}+\sqrt{\frac{1}{2}\left(1-\frac{\ell}{m}\right)}\right)^2 \\
    & \leq\frac{1}{2}+\left(1+\operatorname{o}_k\left(1\right)\right)\sqrt{\frac{1}{k\lambda}}.
\end{align}
Since the semicircle bound is a bound on the expected performance rather than the performance with non-negligible probability, it is smaller than $\mu_{\text{DQI}}$.
The asymptotic gap in $k$, however, suggests that DQI may not perform optimally over unstructured ensembles even among stable algorithms. Motivated by this observation, we show that a depth-1 quantum algorithm known as QAOA~\cite{farhi2014quantumapproximateoptimizationalgorithm} achieves the same $k$ scaling as $\mu_{\mathrm{DQI}}$ for \textsc{MAX-$k$-XOR-SAT}, beyond what DQI can achieve in expectation:
\begin{equation}
    \frac{\left\langle g\right\rangle_{\mathrm{QAOA}}}{m}=\frac{1}{2}+\left(1+\operatorname{o}_k\left(1\right)\right)\sqrt{\frac{\ln\left(k\right)}{4\ce k\lambda}}.
\end{equation}

Finally, we consider what is optimally achieved by classical algorithms. It is widely believed that the classical algorithm known as approximate message passing (AMP) achieves the OGP threshold for combinatorial optimization problems~\cite{doi:10.1073/pnas.2108492118,chen2019,gamarnik2022algorithmsbarrierssymmetricbinary,10.1214/23-AAP1953,10.1002/cpa.22222,cheairi2024algorithmicuniversalitylowdegreepolynomials}. Much like DQI, AMP has the nice property that its optimal performance can be computed without running the algorithm~\cite{alaoui2020algorithmicthresholdsmeanfield}. When applied to \textsc{MAX-$k$-XOR-SAT}, the SAT fraction achieved by AMP is~\cite{el2021optimization,marwaha2022boundsapproximating,cheairi2024algorithmicuniversalitylowdegreepolynomials}:
\begin{equation}
    \mu_{\mathrm{AMP}}=\frac{1}{2}+\frac{1}{2\sqrt{\lambda}}\operatorname{\normalfont\textsc{P}}_k^{\mathrm{ALG}}.
\end{equation}
Our Theorem~\ref{thm:dqi_top_obs_intro} then suggests that $\mu_{\mathrm{DQI}}\leq\mu_{\mathrm{AMP}}$,
and in particular DQI is unable to outperform AMP on typical \textsc{MAX-$k$-XOR-SAT} instances under Assumption~\ref{ass:dqi_dec_conds}. These results are summarized as Fig.~\ref{fig:ogp}(c).

\subsection{Prior Work}

The performance of DQI on combinatorial optimization problems has been studied extensively since the development of the algorithm. Simultaneously to our work, \cite{marwaha2025complexitydecodedquantuminterferometry} gave evidence that sampling from the measurement distribution of the DQI circuit is difficult classically. However, they leave open whether DQI is able to outperform classical algorithms on combinatorial optimization problems, which we address here. Soon after our results were released,~\cite{parekh2025quantumadvantagedecodedquantum} showed that DQI cannot outperform classical algorithms in optimizing \textsc{MAX-CUT}, using \textsc{MAX-CUT}-specific techniques which differ conceptually from the techniques we utilize here. Finally, \cite{anschuetz2025efficientlearningimpliesquantum} studied the average-case performance of quantum algorithms for the quantum ground state problem using topological arguments, all made with respect to a quantum earth mover's distance. The arguments made there do not apply to our setting, as we show \textsc{MAX-$k$-XOR-SAT} satisfies an OGP with respect to a new semimetric we construct, which is not directly comparable to the quantum earth mover's distance used in \cite{anschuetz2025efficientlearningimpliesquantum}. Furthermore, the techniques developed in \cite{anschuetz2025efficientlearningimpliesquantum} were not sufficient to demonstrate the stability of the DQI algorithm, which we demonstrate using new techniques and concepts---such as that of \emph{locally restrictable codes}---here.

\section{Preliminaries}\label{sec:prelims}

\subsection{\texorpdfstring{\textsc{MAX-$k$-XOR-SAT}}{MAX-k-XOR-SAT}}\label{sec:max_xor_sat_background}

We are interested in the hardness of \emph{\textsc{MAX-$k$-XOR-SAT}}. Here, the task is to satisfy as many of $m$ parity constraints over $\mathbb{F}_2$ as possible, where each parity constraint involves $k$ of $n$ variables. More concretely, a \textsc{MAX-$k$-XOR-SAT} problem is defined by a \emph{constraint matrix} $\bm{B}\in\mathbb{F}_2^{m\times n}$---assumed to be $k$-row sparse---and a list of \emph{parities} $\bm{v}\in\mathbb{F}_2^m$, and the goal is to find a bit string $\bm{z}\in\mathbb{F}_2^n$ which minimizes the number of violated clauses:
\begin{equation}
    \#\text{ violated clauses}=\left\lVert\bm{B}\bm{z}\oplus\bm{v}\right\rVert_1.
\end{equation}
We can equivalently write this as a maximization problem, where one is interested in maximizing the function:
\begin{equation}
    g\left(\bm{z}\right):=g_{\left(\bm{B},\bm{v}\right)}\left(\bm{z}\right):=m-\left\lVert\bm{B}\bm{z}\oplus\bm{v}\right\rVert_1.
\end{equation}
In what follows, we will use $\lambda:=m/n$ to denote the \emph{clause density}. Obviously, the problem is only interesting when the linear system of equations defining the problem is overdetermined, i.e., when $\lambda>1$. We also assume that $k>2$ and is even to simplify some of our later calculations.

Throughout the paper we will primarily be interested in the $n\to\infty$ limit. We will hold $\lambda$ and $k$ fixed as we take this limit. For this reason $\operatorname{O}\left(\cdot\right)$ will be used as big-O notation with respect to the variable $n$. In some situations, we wish to use big-O notation with respect to a function of variable $x$ that is distinct from $n$; in these situations, we will use the notation $\operatorname{O}_x\left(\cdot\right)$ to be maximally clear.

We will particularly be interested in the average-case hardness of \textsc{MAX-$k$-XOR-SAT} when the constraint matrix and the parities are drawn from some distribution. It is natural to consider when $\bm{v}$ is drawn uniformly at random from the hypercube (and independently from $\bm{B}$), i.e., that the entries of $\bm{v}$ are drawn from the product distribution $\bm{v}\sim\operatorname{Ber}\left(1/2\right)^{\otimes m}$, where $\operatorname{Ber}\left(1/2\right)$ is the Bernoulli distribution with probability of $1$ equal to $1/2$. We will therefore fix this distribution for $\bm{v}$, and write:
\begin{equation}\label{eq:par_dist}
    \bm{v}\sim\mathbb{P}_{\mathrm{par}}:=\operatorname{Ber}\left(1/2\right)^{\otimes m}.
\end{equation}

There are many natural choices of distribution for $\bm{B}$. In the paper in which DQI was introduced~\cite{jordan2025optimizationdecodedquantuminterferometry}, the authors studied \textsc{MAX-$k$-XOR-SAT} instances where $\bm{B}^\intercal$ was drawn from the \emph{Gallager ensemble} $\mathcal{G}\left(m,k,d\right)$~\cite{gallager2003low}, which we will describe in more detail in a moment (see Sec.~\ref{sec:ldpc_codes}). Here, $m$ denotes the number of columns of $\bm{B}^\intercal$, $k$ the column-sparsity, and $d=\lambda k$ the row-sparsity of matrices drawn from $\mathcal{G}\left(m,k,d\right)$. We will also assume the $\bm{B}^\intercal$ is drawn from $\mathcal{G}\left(m,k,d\right)$, and write:
\begin{equation}\label{eq:con_dist}
    \mathbb{P}_{\mathrm{con}}=\mathcal{G}\left(m,k,d\right)^\intercal\sim\bm{B}
\end{equation}
to denote the resulting distribution on $\bm{B}$. This together with $\mathbb{P}_{\mathrm{par}}$ define the distribution of \textsc{MAX-$k$-XOR-SAT} instances we are primary concerned with here.
\begin{definition}[Transposed Gallager ensemble of \textsc{MAX-$k$-XOR-SAT} instances]\label{def:trans_gall_ens}
    We call the product distribution:
    \begin{equation}
        \mathbb{P}_{\mathrm{G}}:=\mathbb{P}_{\mathrm{con}}\otimes\mathbb{P}_{\mathrm{par}}=\mathcal{G}\left(m,k,d\right)^\intercal\otimes\operatorname{Ber}\left(1/2\right)^{\otimes m}
    \end{equation}
    the \emph{transposed Gallager ensemble} of \textsc{MAX-$k$-XOR-SAT} instances.
\end{definition}

\subsection{The Gallager Ensemble of LDPC Codes}\label{sec:ldpc_codes}

\begin{figure}
    \centering
    \includegraphics[width=0.9\linewidth]{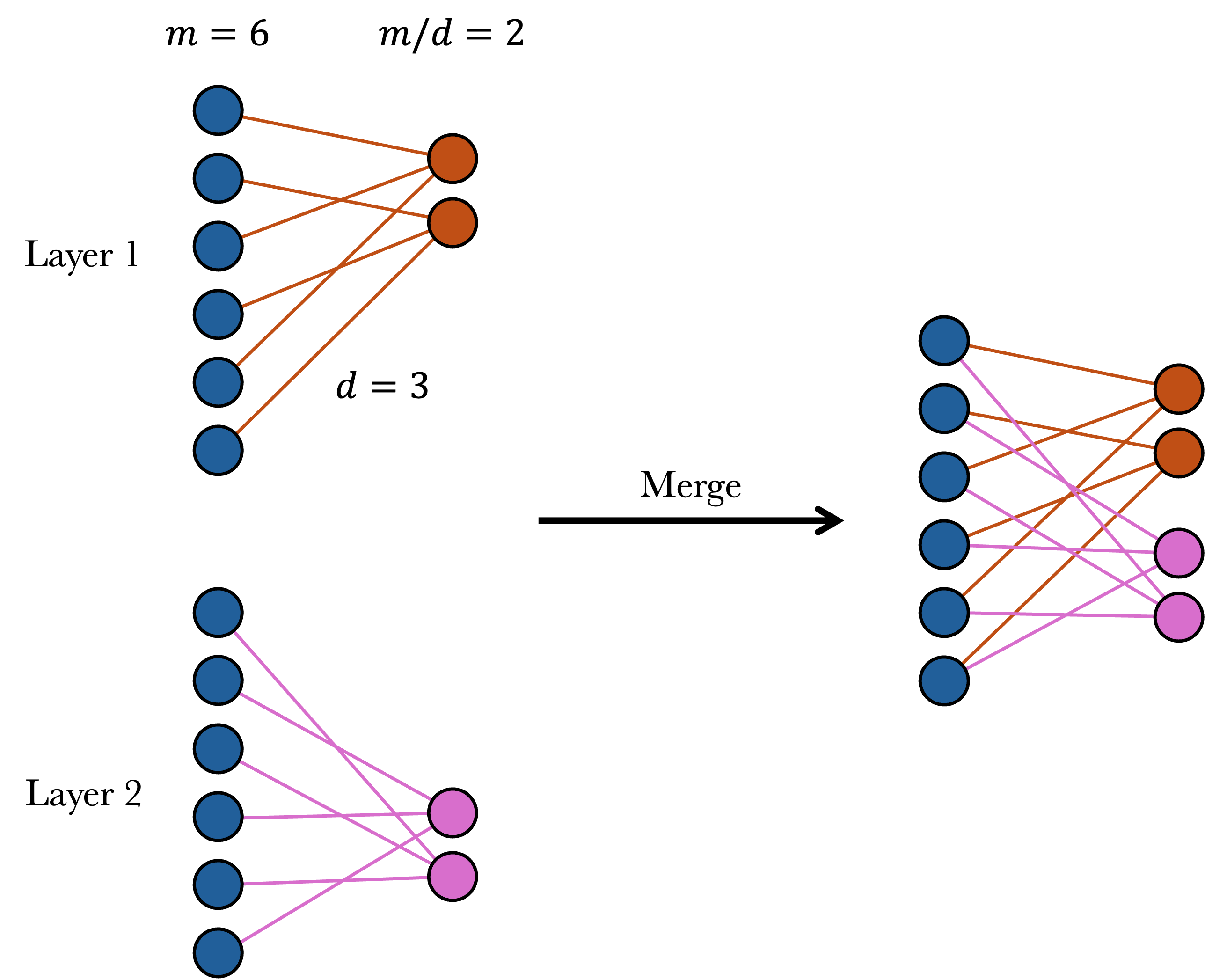}
    \caption{One step of the Gallager construction with $m = 6, d = 3$. This step is repeated for a total $k = \lambda^{-1} s$ layers, where $d$ and $\lambda$ are given parameters.}
    \label{fig:gallager}
\end{figure}

We now review both the Gallager code ensemble as well as general properties of LDPC codes. 
Define the binary matrix $\boldsymbol{H} \in \mathbb{F}_2^{n \times m}$ to be the \emph{parity check matrix} of a code $\mathcal{C}$, where \begin{align}
    \mathcal{C} = \{ \boldsymbol{v} \in \mathbb{F}_2^m \,:\, \boldsymbol{Hv} = 0 \}
\end{align}
where all algebra is done over $\mathbb{F}_2$ unless explicitly stated otherwise. 
In other words, $\mathcal{C}$ is the linear subspace of $\mathbb{F}_2^m$ given by the kernel of $\boldsymbol{H}$. 
Rows of $\boldsymbol{H}$ are known as \emph{checks}. 
We are particularly interested in codes whose parity check matrix is sparse in the following way.

\begin{definition}[Code sparsities] \label{def:code_sparsities}
    Let $\boldsymbol{H} \in \mathbb{F}_2^{n \times m}$ be a parity check matrix. Then the \emph{bit sparsity} of $\boldsymbol{H}$ and \emph{check sparsity} of $\boldsymbol{H}$ are given by \begin{align}
        \bitwt(H) := \max_{j \in [m]} \sum_{i=1}^{n} H_{ij} ,\; \checkwt(\boldsymbol{H}) := \max_{i \in [n]} \sum_{j=1}^m H_{ij} ,
    \end{align}
    where both sums are as integers (i.e., over $\mathbb{Z}$), not over $\mathbb{F}_2$.
\end{definition}
In the standard notation of coding theory, $\frac{n}{m} = \lambda^{-1} = 1 - r$, where $r$ is known as the \textit{design rate} of the family.
\begin{definition}[LDPC code] \label{def:LDPC_code}
    Let $\boldsymbol{H} \in \mathbb{F}_2^{\lambda^{-1} m \times m}$ be a parity check matrix. We say that $\boldsymbol{H}$ is a $(k, d)$-LDPC code code if \begin{align}
        \bitwt(\boldsymbol{H}) \leq k , \quad \checkwt(\boldsymbol{H}) \leq d.
    \end{align}
\end{definition}
Informally, a LDPC code satisfies the property that every check is supported on only a small number of bits, and every bit is in the support of only a small number of checks.
The Gallager ensemble with parameters $(k, d)$ is an ensemble of $(k, d)$-regular LDPC codes with particularly nice properties. 
For that reason, it is perhaps the most widely studied and used ensemble of random LDPC codes~\cite{gallager2003low}.

We now describe the Gallager ensemble $\mathcal{G}\left(m,k,d\right)$ in more detail. 
The most convenient way to define the Gallager ensemble is in the language of graphs. 
We may represent $\boldsymbol{H}$ as a bipartite graph $G = (\mathcal{L} \sqcup \mathcal{R}, \mathcal{E})$ with $m$ nodes in the left vertex set $\mathcal{L}$, $n = \lambda^{-1} m$ nodes in the right vertex set $\mathcal{R}$, and edge set $\mathcal{E} \subseteq \mathcal{L} \times \mathcal{R}$. 
The left nodes are referred to as the ``data'' nodes and the right nodes as ``syndrome'' or ``check'' nodes. 
We connect node $i \in \mathcal{L}$ to node $j \in \mathcal{R}$ if $H_{ji} = 1$. 
Such a graph is known as a \textit{Tanner graph} of $\boldsymbol{H}$ and is denoted $G(H)$. 

\begin{definition}[Regular bipartite graphs]
    A bipartite graph $G = (\mathcal{L} \sqcup \mathcal{R}, \mathcal{E})$ is (weakly) $(a, b)$-regular if $\forall i \in \mathcal{L}$, $\deg(i) = a$ ($\deg(i) \leq a$) and $\forall j \in \mathcal{R}$, $\deg(j) = b$ ($\deg(j) \leq b$).
\end{definition}

A code is $(k, d)$-LDPC if and only if its Tanner graph is weakly $(k, d)$-regular.

\begin{definition}[Bipartite graph merging] \label{def:merging_graphs}
    Let $G_1 = (\mathcal{L} \sqcup \mathcal{R}_1, \mathcal{E}_1)$ and $G_2 = (\mathcal{L} \sqcup \mathcal{R}_2, \mathcal{E}_2)$ be two bipartite graphs with the same left node set and disjoint right node sets.  The \textit{merge} of $G_1$ and $G_2$ is a graph $G = (\mathcal{L} \sqcup (\mathcal{R}_1 \cup \mathcal{R}_2), \mathcal{E}_1 \cup \mathcal{E}_2)$ given by combining the right node sets and the edges together into the same left node set.
\end{definition}
 See Fig.~\ref{fig:gallager} for a visualization of a merge. With these graph notions formulated, we now formally define the Gallager ensemble.
\begin{definition}[Gallager ensemble] \label{def:gallager_ensemble}
Let $k \geq 3$ and $d\geq k$ be constant integers.
    The Gallager ensemble $\mathcal{G}(m, k, d)$ is a distribution over parity check matrices $\boldsymbol{H} \in \mathbb{F}_2^{n \times m}$, where $n := m \frac{k}{d}$ and $m/d$ is assumed to be a positive integer, defined by a sampling process over Tanner graphs. 
    Sample $k$ uniformly random $(1, d)$-regular bipartite graphs (called \textit{layers}) $G_1, \ldots, G_k$ with $m$ left nodes and $m/d$ right nodes. The sampled code is defined to be the code whose parity check matrix corresponds to the Tanner graph obtained by merging the layers $G_1, \ldots, G_k$ in the manner of Definition~\ref{def:merging_graphs}. This graph is a $(k, d)$-regular bipartite graph with $m$ left nodes and $k\frac{m}{d} = n$ right nodes, so $\boldsymbol{H} \in \mathbb{F}_2^{n \times m}$ as needed.
\end{definition}
Here $\lambda = \frac{m}{n} = \frac{m}{m k/d} = \frac{d}{k}$.
Figure~\ref{fig:gallager} visualizes the merging of layers in a small case. The Gallager ensemble has been studied through many coding-theoretic lenses~\cite{gallager2003low,mosheiff2021low}. Note that one consequence of the requirement that $m/d\in\mathbb{N}$ is that $k$ divides $n$, which we assume throughout what follows.


\section{Stability of Decoded Quantum Interferometry}\label{sec:dqi_is_stab}

We now discuss the stability of DQI over certain code families. We emphasize that the stability of DQI holds beyond the transposed Gallager ensemble of \textsc{MAX-$k$-XOR-SAT} instances, and indeed is fairly generic; we discuss the stability of DQI over other code families in Appendix~\ref{sec:loc_rest_common_code_ens}.

Before we begin, we summarize the steps of DQI, keeping only the details that are relevant to us (for a detailed description of DQI, we refer the reader to Methods). The algorithm is parameterized by some $1\leq\ell\leq m$; the implementation of the algorithm with this choice of parameter we will denote $\operatorname{DQI}_\ell$, and takes as input a description of the problem $\left(\bm{B},\bm{v}\right)$. We refer to parts of Fig.~\ref{fig:DQI_flowchart} in the description of the algorithm which follows to aid the reader.
\begin{enumerate}
    \item A problem-independent initial state $\ket{\psi_2}$ is prepared (Fig.~\ref{fig:DQI_flowchart}, pink bubble).
    \item The gate $\bigotimes_{i=1}^m\bm{Z}_i^{v_i}$ is applied, yielding the state $\ket{\psi_3\left(\bm{v}\right)}$ (Fig.~\ref{fig:DQI_flowchart}, green bubble, first step).
    \item The isometry:
    \begin{equation}
        \bm{U}_{\bm{B}}=\sum_{\bm{y}\in\mathbb{F}_2^m}\left(\ket{\bm{y}}\otimes\ket{\bm{B}^\intercal\bm{y}}\right)\bra{\bm{y}}
    \end{equation}
    is applied, yielding the state $\ket{\psi_4\left(\bm{B},\bm{v}\right)}$ (Fig.~\ref{fig:DQI_flowchart}, green bubble, second step).
    \item A decoding operator:
    \begin{equation}
        \bm{O}_{\bm{B}}=\sum_{\bm{y}\in\mathbb{F}_2^m:\left\lvert\bm{y}\right\rvert\leq\ell}\ket{\bm{B}^\intercal\bm{y}}\left(\bra{\bm{y}}\otimes\bra{\bm{B}^\intercal\bm{y}}\right)
    \end{equation}
    is applied, yielding the state $\ket{\psi_5\left(\bm{B},\bm{v}\right)}$ (Fig.~\ref{fig:DQI_flowchart}, blue bubble, first step). Note this can only be performed efficiently if the code is efficiently decodable through $\ell$ errors.
    \item The Fourier transform $\bigotimes_{i=1}^n\bm{H}_i$ is applied, yielding the final state $\ket{\psi_6\left(\bm{B},\bm{v}\right)}$ (Fig.~\ref{fig:DQI_flowchart}, blue bubble, second step).
\end{enumerate}

Our main result is that DQI is a stable algorithm over \emph{locally restrictable} code families, which we now define.
\begin{definition}[$\left(\epsilon,\varDelta,p_{\mathrm{res}}\right)$-restrictability]\label{def:res}
    Fix $0<\epsilon\leq 1$, $1\leq\varDelta\leq m$, and $0\leq p_{\mathrm{res}}\leq 1$. If for $\bm{B}$ drawn from $\mathbb{P}_{\mathrm{con}}$ the event 
    \begin{equation}
        \{\exists S_{\bm{B}}\subseteq\left[n\right]:\left\lvert S_{\bm{B}}\right\rvert\leq\epsilon n\wedge\bm{\varPi}_{S_{\bm{B}}}\bm{B}^\intercal\in\mathcal{H}_{\geq\varDelta}\}
    \end{equation}
    has probability at least $1-p_{\mathrm{res}}$,
    we say that $\mathbb{P}_{\mathrm{con}}$ is \emph{$\left(\epsilon,\varDelta,p_{\mathrm{res}}\right)$-restrictable}. Here, $\bm{\varPi}_{S_{\bm{B}}}\in\left\{0,1\right\}^{\left\lvert S_{\bm{B}}\right\rvert\times n}$ is a projector onto the rows labeled by $S_{\bm{B}}$, and $\mathcal{H}_{\geq\varDelta}$ is the set of check matrices of codes of distance at least $\varDelta$.
\end{definition}
Informally, locally restrictable codes are those which still have lower-bounded distance ($\varDelta$) even when all but $\epsilon n$ of the syndrome bits are discarded. Our main result in this section is that $\operatorname{DQI}_\ell$ is stable (Definition~\ref{def:stab_qas_inf}, see Methods for formal definition) if $\mathbb{P}_{\mathrm{con}}$ is $\left(\epsilon,\varDelta,p_{\mathrm{res}}\right)$-restrictable and $2\ell+1\leq\varDelta$.
\begin{theorem}[DQI is stable]\label{thm:dqi_is_stable}
    Assume $\mathbb{P}_{\mathrm{con}}$ is $\left(\epsilon,\varDelta,p_{\mathrm{res}}\right)$-restrictable. Then $\operatorname{DQI}_\ell$ is $\left(\epsilon n,0,p_{\mathrm{res}}\right)$-stable if $2\ell+1\leq\varDelta$.
\end{theorem}
In the special case of a decoder that has a certain Lipschitz continuity property---which we show in Appendix~\ref{sec:decode_weight} includes belief propagation as a special case---we can prove stability with a stronger choice of parameters.
\begin{definition}[$L$-inverse Lipschitz decoder]\label{def:l_inv_lip_dec}
    Consider a decoder map $\mathcal{D}:\mathbb{F}_2^n\to\mathbb{F}_2^m$ which is one-to-one on a co-domain $\mathcal{Y}_\varDelta:=\left\{\bm{y}\in\mathbb{F}_2^m:\left\lvert\bm{y}\right\rvert\leq\varDelta\right\}$. We say $\mathcal{D}$ is $L$-inverse Lipschitz with respect to $\mathcal{Y}_\varDelta$ if:
    \begin{equation}
        \left\lvert\mathcal{D}^{-1}\left(\bm{y}\right)-\mathcal{D}^{-1}\left(\bm{y'}\right)\right\rvert_1\leq L\left\lVert\bm{y}-\bm{y'}\right\rVert_1
    \end{equation}
    for all $\bm{y}\in\mathcal{Y}_\varDelta$.
\end{definition}
\begin{theorem}[DQI with an inverse Lipschitz decoder is Lipschitz]\label{thm:dqi_is_lipschitz}
    Assume $\operatorname{DQI}_\ell$ is implemented using an $L$-inverse Lipschitz decoder. Then $\operatorname{DQI}_\ell$ is $\left(0,L,0\right)$-stable.
\end{theorem}

The remainder of this section is structured in the following way. First, in Sec.~\ref{sec:dqi_is_stable_proof} we prove Theorems~\ref{thm:dqi_is_stable} and~\ref{thm:dqi_is_lipschitz}. Then, in Sec.~\ref{sec:res_and_sparse_codes} we prove that the Gallager ensemble is locally restrictable, and instantiate Theorem~\ref{thm:dqi_is_stable} with the local restrictability parameters of the Gallager ensemble. We also remark that in Appendix~\ref{sec:loc_rest_common_code_ens} we show that other common code families are locally restrictable, demonstrating that the stability of DQI is not unique to the Gallager ensemble.

\subsection{DQI is Stable}\label{sec:dqi_is_stable_proof}

\begin{figure*}[ht!]
    \centering
    \includegraphics[width=0.8\linewidth]{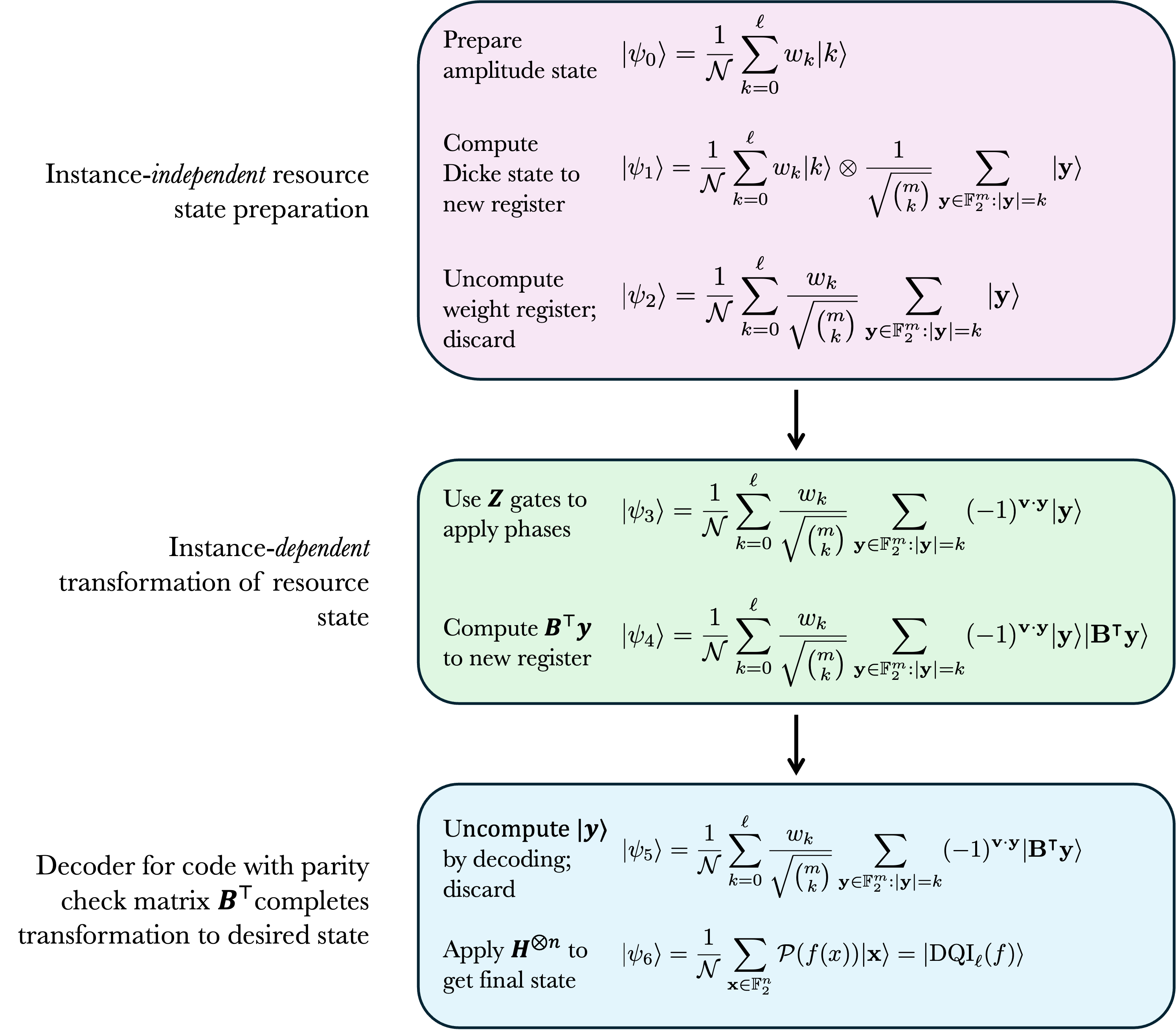}
    \caption{Workflow of the Decoded Quantum Interferometry (DQI) algorithm, which aim to produce the state $\propto \sum_{\bm{x} \in \mathbb{F}_2^n} \mathcal{P}(f(\bm{x})) \ket{\bm{x}}$, where $\mathcal{P}$ is a univariate polynomial and $f : \mathbb{F}_2^n \to \mathbb{Z}$ is an objective function.  Initially, in the pink bubble, an instance-independent resource state $\ket{\psi_2}$ is prepared. Here $w_k \in \mathbb{R}$ are constants and $1/\mathcal{N}$ is a constant that enforces normalization. In the green bubble, $\ket{\psi_2}$ is processed to encode information about the problem instance, given by $(\bm{B}, \bm{v}) \in \mathbb{F}_2^{m \times n} \times \mathbb{F}_2^{m}$. Both steps can be performed in quantum polynomial time. In the last step (blue bubble), the decoder for the code corresponding to parity check matrix $\bm{B}^\intercal$ is applied to uncompute a register. A Hadamard transform completes the transformation into the desired state $\ket{\text{DQI}_{\ell}(f)}$, where $f(\bm{x}) = m - 2\left\lVert\bm{Bx}\oplus\bm{v}\right\rVert_1$ is the difference between the number of satisfied and unsatisfied clauses.}
    \label{fig:DQI_flowchart}
\end{figure*}
We first prove Theorem~\ref{thm:dqi_is_stable}. Recall the steps of $\operatorname{DQI}_\ell$. As we are only interested in the stability of $\textsc{DQI}_\ell$ when varying $\bm{v}$, we might as well rewrite these steps in the following way.
\begin{enumerate}
    \item A phase-independent initial state $\bm{\rho}_1\left(\bm{B}\right)$ is prepared.
    \item The gate $\bigotimes_{i=1}^m\bm{Z}_i^{v_i}$ is applied to the register containing $\ket{\bm y}$, yielding the state $\bm{\rho}_2\left(\bm{B},\bm{v}\right)$.
    \item A decoding operator:
    \begin{equation}
        \bm{O}_{\bm{B}}=\sum_{\bm{y}\in\mathbb{F}_2^m:\left\lvert\bm{y}\right\rvert\leq\ell}\ket{\bm{B}^\intercal\bm{y}}\left(\bra{\bm{y}}\otimes\bra{\bm{B}^\intercal\bm{y}}\right)
    \end{equation}
    is applied, yielding the state $\bm{\rho}_3\left(\bm{B},\bm{v}\right)$.
    \item The Fourier transform $\bigotimes_{i=1}^n\bm{H}_i$ is applied, yielding the final state $\bm{\mathcal{A}}\left(\bm{B},\bm{v}\right)$.
\end{enumerate}

We claim that, conditioned on the event $\mathcal{E}_{\bm{B}}$ given by
\begin{equation}\label{eq:res_event}
    \left\{\exists S_{\bm{B}}\subseteq\left[n\right]:\left\lvert S_{\bm{B}}\right\rvert\leq\epsilon n\wedge\bm{\varPi}_{S_{\bm{B}}}\bm{B}^\intercal\in\mathcal{H}_{\geq\varDelta}\right\},
\end{equation}
where $\varPi_{S_{\bm{B}}}$ and $\mathcal{H}_{\geq\varDelta}$ are defined in Definition~\ref{def:res}, one can commute step 2 after step 3 via the application of a single operator with bounded locality.
\begin{proposition}[Commuting decoding and phasing]\label{prop:commuting_steps}
    Condition on the event $\mathcal{E}_{\bm{B}}$ (from Eq.~\eqref{eq:res_event}) occurring. 
    If $2\ell+1\leq\varDelta$, for all $\bm{v}$ there exists a state $\bm{\tilde{\rho}}\left(\bm{B}\right)$ and a unitary operator $\bm{V}_{\bm{B},\bm{v}}$ supported only on $S_{\bm{B}}$ (also from Eq.~\eqref{eq:res_event}), such that
    \begin{equation}
        \bm{\rho}_3\left(\bm{B},\bm{v}\right)=\bm{V}_{\bm{B},\bm{v}}\bm{\tilde{\rho}}\left(\bm{B}\right)\bm{V}_{\bm{B},\bm{v}}^\dagger.
    \end{equation}
\end{proposition}
\begin{proof}
    We begin with the state $\bm{\rho}_1(\bm{B})$ as above.
    Under the standard algorithmic procedure of DQI, we would next apply the phasing operation $\bigotimes_{i=1}^m \bm{Z}_i^{v_i}$ to produce $\bm{\rho}_2(\bm{B}, \bm{v})$, followed by a decoding isometry $\bm{O}_{\bm{B}}$ to yield the state $\bm{\rho}_3(\bm B, \bm v)$.
    We here will wish to instead \emph{first} apply a decoding operator directly to the state $\bm{\rho}_1(\bm B)$, producing a state \emph{independent} of $\bm v$, and then apply a different operation which restores the phases we neglected, so that the final state is exactly equal to the desired state $\bm{\rho}_3(\bm B, \bm v)$.
    We show that this latter operation can in fact be implemented with support contained entirely in $S_{\bm B}$.

    Assuming that event $\mathcal{E}_{\bm B}$ occurs, there is a set $S_{\bm B} \in [n]$ such that, restricted only those the checks in $S_{\bm B}$, the code specified by $\bm{B}^\intercal$ still has distance at least $\varDelta$.
    Further, we assume that the parameter $\ell$ of DQI is such that $2 \ell + 1 \leq \varDelta$, which means that all errors of weight $\leq \ell$ can be corrected using only the information from the syndromes computed by checks in $S_{\bm B}$.
    Therefore, there is a decoding isometry \begin{align}
        \bm{W}_{\bm B} := \sum_{\bm{y} \in \mathbb{F}_2^m : |\bm y| \leq \ell} \ket{\bm{B}^\intercal \bm y} (\bra{\bm y} \otimes \bra{\bm{B}^\intercal \bm y})
    \end{align}
    which can be implemented by acting solely on $S_{\bm B}$.
    We denote this decoder by $\bm{W}_{\bm B}$ even though it appears identical to $\bm{O}_{\bm B}$ because while its action is identical to that of $\bm{O}_{\bm B}$, it explicitly only acts on $S_{\bm B}$ and therefore may differ from $\bm{O}_{\bm B}$ in algorithmic aspects such as efficient implementability.
    Let $\tilde{\bm \rho}(\bm B) = \bm{W}_{\bm B} \bm{\rho}_1(\bm B) \bm{W}_{\bm B}^\dagger$.
    Then define $\bm{V}_{\bm B, \bm v}$ to (a) undo $\bm W_{\bm B}$, (b) compute the phases $\bigotimes_{i=1}^m \bm{Z}_i^{v_i}$ on the register holding $\ket{\bm y}$, and (c) re-do $\bm W_{\bm B}$. 
    This is a sequence of isometries which start and end on the same dimension, and thus is a unitary operation.
    Moreover, the support of this operator $\bm{V}_{\bm B, \bm v}$ is contained $S_{\bm B}$ since the only part of the input state which is acted on is $S_{\bm B}$.
    Finally, by construction, \begin{align}
        \bm{\rho}_3(\bm B, \bm v) = \bm{V}_{\bm B, \bm v} \tilde{\bm \rho}(\bm B) \bm{V}_{\bm B, \bm v}^\dagger
    \end{align}
    as claimed.
\end{proof}

Due to Proposition~\ref{prop:commuting_steps}, the result of $\operatorname{DQI}_\ell$ can, when the event $\mathcal{E}_{\bm{B}}$ (Eq.~\eqref{eq:res_event}) occurs, be written in the following way:
\begin{equation}
    \operatorname{DQI}_\ell\left(\bm{B},\bm{v}\right)=\bm{H}^{\otimes n}\bm{V}_{\bm{B},\bm{v}}\bm{\tilde{\rho}}\left(\bm{B}\right)\bm{V}_{\bm{B},\bm{v}}^\dagger\bm{H}^{\otimes n}.
\end{equation}
With this more convenient form of the DQI state, we are able to complete the proof of Theorem~\ref{thm:dqi_is_stable}.
\begin{proof}[Proof of Theorem~\ref{thm:dqi_is_stable}]
    Condition on the event $\mathcal{E}_{\bm{B}}$ (Eq.~\eqref{eq:res_event}) occurring. Then:
    \begin{equation}
        \operatorname{DQI}_\ell\left(\bm{B},\bm{v}\right)=\bm{H}^{\otimes n}\bm{V}_{\bm{B},\bm{v}}\bm{\tilde{\rho}}\left(\bm{B}\right)\bm{V}_{\bm{B},\bm{v}}^\dagger\bm{H}^{\otimes n}.
    \end{equation}
    Now consider $\Tr_{S_{\bm{B}}}\left(\operatorname{DQI}_\ell\right)$, given by
    \begin{equation}
        \left(\bigotimes_{i\not\in S_{\bm{B}}}\bm{H}_i\right)\Tr_{S_{\bm{B}}}\left(\bm{\tilde{\rho}}\left(\bm{B}\right)\right)\left(\bigotimes_{i\not\in S_{\bm{B}}}\bm{H}_i\right).
    \end{equation}
    This state is independent of $\bm{v}$, so for any $\bm{v},\bm{v'}\in\mathbb{F}_2^m$:
    \begin{equation}
        \left\lVert\Tr_{S_{\bm{B}}}\left(\operatorname{DQI}_\ell\left(\bm{B},\bm{v}\right)-\operatorname{DQI}_\ell\left(\bm{B},\bm{v'}\right)\right)\right\rVert_{W_2}=\bm{0}.
    \end{equation}
    The desired result then follows by the definition of stability (see Methods). 
\end{proof}

We now prove Theorem~\ref{thm:dqi_is_lipschitz}---stability of DQI when the decoder is inverse Lipschitz---which is much more straightforward.
\begin{proof}[Proof of Theorem~\ref{thm:dqi_is_lipschitz}]
    Fix $\bm{v}$ and $\bm{v'}$, and let $\bm{d}:=\bm{v}\oplus\bm{v'}$. We then have:
    \begin{equation}
        \bm{\rho}_2\left(\bm{B},\bm{v'}\right)=\bigotimes_{i=1}^m\bm{Z}_i^{d_i}\bm{\rho}_2\left(\bm{B},\bm{v}\right)\bigotimes_{i=1}^m\bm{Z}_i^{d_i}.
    \end{equation}
    Just as in Proposition~\ref{prop:commuting_steps}, we hope now to commute $\bigotimes_{i=1}^m\bm{Z}_i^{d_i}$ through the decoding step. Let $S_{\bm{d}}\subseteq\left[m\right]$ be the support of $\bm{d}$, and let $R_{\bm{d}}\subseteq\left[n\right]$ be the support of $\mathcal{D}^{-1}\left(\bm{d}\right)$, which by the $L$-inverse Lipschitz assumption of the decoder is of cardinality at most $\left\lvert R_{\bm{d}}\right\rvert\leq L\left\lVert\bm{d}\right\rVert_1$. Commuting $\bigotimes_{i=1}^m\bm{Z}_i^{d_i}$ through the decoding operator $\bm{O}_{\bm{B}}$ implementing $\mathcal{D}$, we therefore have:
    \begin{equation}
        \bm{O}_{\bm{B}}\bigotimes_{i=1}^m\bm{Z}_i^{d_i}=\bm{V}_{R_{\bm{d}}}\bm{O}_{\bm{B}},
    \end{equation}
    for some operator $\bm{V}_{R_{\bm{d}}}$ with support only on $R_{\bm{d}}$. By construction,
    \begin{align}
        \operatorname{DQI}_\ell\left(\bm{B},\bm{v}\right)&=\bm{H}^{\otimes n}\bm{\rho}_3\left(\bm{B},\bm{v}\right)\bm{H}^{\otimes n},\\
        \operatorname{DQI}_\ell\left(\bm{B},\bm{v'}\right)&=\bm{H}^{\otimes n}\bm{V}_{R_{\bm{d}}}\bm{\rho}_3\left(\bm{B},\bm{v}\right)\bm{V}_{R_{\bm{d}}}^\dagger\bm{H}^{\otimes n}.
    \end{align}
    Now, by the locality of $\bm{V}_{R_{\bm{d}}}$ and the invariance of the Wasserstein metric under $1$-local unitary transformations,
    \begin{align}
            & \left\lVert\operatorname{DQI}_\ell\left(\bm{B},\bm{v}\right)-\operatorname{DQI}_\ell\left(\bm{B},\bm{v'}\right)\right\rVert_{W_2}\\
            =& \left\lVert\bm{\rho}_3\left(\bm{B},\bm{v}\right)-\bm{V}_{R_{\bm{d}}}\bm{\rho}_3\left(\bm{B},\bm{v}\right)\bm{V}_{R_{\bm{d}}}^\dagger\right\rVert_{W_2}\\
            &\leq L\left\lVert\bm{d}\right\rVert_1.
    \end{align}
    Recalling $\bm{d}=\bm{v}\oplus\bm{v'}$ yields the final result.
\end{proof}

\subsection{Locally Restrictable Codes}\label{sec:res_and_sparse_codes}

The assumption of restrictability from Definition~\ref{def:res} captures a sufficient condition for $\operatorname{DQI}$ to be stable. 
We here prove that the Gallager ensemble is locally restrictable.
In particular, we show that if we remove all but the first few check nodes in the Tanner graph, then the remaining code still has very high distance.
To formalize the distance of an ensemble, we define an ensemble to be good if a draw from the ensemble has high distance with high probability.
\begin{definition}[Good ensemble] \label{def:good_ensemble}
    Let $\lambda^{-1} \in (0, 1)$ be a constant.
    An ensemble $\mathcal{D}_{m, \lambda^{-1}}$ is $(\varDelta(\lambda^{-1}), p_{\text{res}}(n, \lambda^{-1}))$-good if with probability at least $1 - p_{\text{res}}(n, \lambda^{-1})$, a parity check matrix $\boldsymbol{H} \sim \mathcal{D}_{m, \lambda^{-1}}$ has distance at least $\varDelta(\lambda^{-1})$. 
\end{definition}
\begin{figure}
    \centering
    \includegraphics[width=\linewidth]{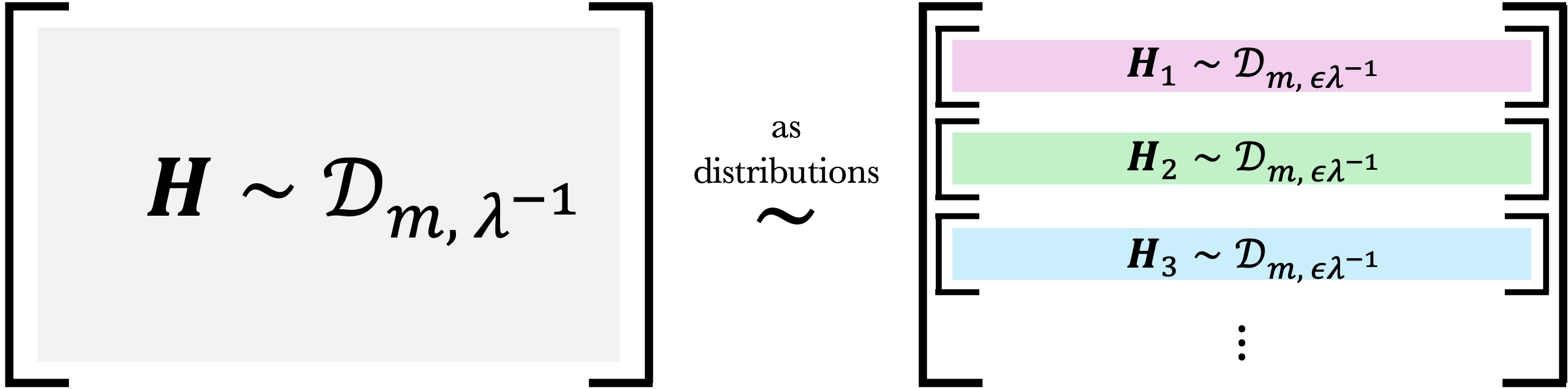}
    \caption{Informal visualization of the self-similarity property of a code. Sampling a larger code is equivalent to sampling smaller horizontal slices and then stacking them.}
    \label{fig:locally_generated}
\end{figure}
It is well known that the Gallager ensemble is good.
\begin{lemma}[Gallager ensemble is good] \label{lemma:gallager_good}
    Fix constants $\lambda^{-1}, \delta \in (0, 1/2)$. Let $\boldsymbol{H} \in \mathbb{F}_2^{\lambda^{-1} m \times m}$ be sampled from a Gallager ensemble $\mathcal{G}(m, k, d = k \lambda)$ with locality parameter $k$ and $n = \lambda^{-1} m$. Then if $\lambda^{-1} > \operatorname{H}_2(\delta)$, there exists a threshold constant $k_0(\delta) \geq 1$ (scaling as $1/\delta$) such that if $k \geq k_0$, the ensemble is $(\delta m, 1/n^{\Omega(1)})$-good.
\end{lemma}
\begin{proof}
    This claim was first shown by \cite{gallager2003low}. A simplified proof is given in \cite[Theorem~2.14]{mosheiff2021low}.
\end{proof}

\begin{theorem}[Restrictability of Gallager ensemble] \label{thm:gallager_is_restrictable}
    Let $n$ be a scaling parameter, $k, d$ be constants,$m = n d /k$, and $\lambda^{-1} = n/m \in (0, 1)$.
    Let $\epsilon$ be such that $\epsilon =  C\ln(k)/k$, where $C$ is any constant independent of $k, d, m$.
    The Gallager ensemble $\mathcal{G}(m, k, d)$ of rate parameter $\lambda^{-1}$ is $(\epsilon, \varDelta, p_{\text{res}})$-restrictable onto the first $\frac{\left\lfloor k\epsilon\right\rfloor}{k} n$ checks, where $\varDelta/m = \Theta_d(1/d)$ and $p_{\text{res}} = 1/n^{\Omega(1)}$.
\end{theorem}
\begin{proof}
The key idea behind the restrictability of the Gallager ensemble is the \emph{self-similarity} of the construction: since the checks are generated in independent blocks, removing some of the blocks is equivalent to generating a smaller instance of the code (see Fig.~\ref{fig:locally_generated}).
More precisely, the $\mathcal{G}(m, k, d)$ ensemble can be seen as, for $j$ from $1$ to $k$, partition the $m$ bits randomly into $m/d$ groups of $d$ bits, and define each group to be a new parity check.
Hence, if we restrict to the first $\epsilon$ fraction of checks, we are equivalently stopping the iterative parity-check generation process at the $(C \ln k)$th iteration.
Since the number of iterations is the $k$ parameter in the ensemble, the code generated is precisely a draw from the $\mathcal{G}(m, C \ln k, d)$ ensemble.
By Lemma~\ref{lemma:gallager_good}, for any $\delta$ such that \begin{align}
    \frac{C \ln k}{d} > \operatorname{H}_2(\delta) ,
\end{align}
this ensemble is $(\delta m, 1/n^{\Omega(1)})$-good.
In the limit of small $q$, $\operatorname{H}_2^{-1}(q) = (1 + \operatorname{o}_{1/q}(1)) \frac{q}{\log_2(1/q)}$.
Hence, if \begin{widetext}
    \begin{align}
    \delta < \operatorname{H}_2^{-1} \left( \frac{C\ln k}{d} \right) & = (1 + \operatorname{o}_k(1)) \frac{C \ln k}{d} \frac{1}{\log_2 d - \log_2 \ln k - \log_2 C} \\
    & = (1 + \operatorname{o}_k(1)) \frac{C\ln\left(2\right)}{d} = \Theta_d(1/d) ,
\end{align}
\end{widetext}
then the restricted code is $(\delta m, 1/n^{\Omega(1)})$-good.
Note that we have taken the limit $k \to \infty$ such that $\lambda^{-1} = \frac{k}{d}$ remains fixed.
\end{proof}

Taking the local restrictability of the Gallager ensemble with Theorem~\ref{thm:dqi_is_stable} then immediately implies the stability of DQI on the Gallager ensemble of codes.
\begin{corollary}[DQI is stable on Gallager ensemble]
\label{cor:DQI_stable_Gallager}
    Let $k, d$ be constants, and let $n$ be the scaling parameter.
    Define $m = n \frac{d}{k}$ and $\lambda^{-1} = \frac{n}{m} = \frac{k}{d} \in (0, 1)$.
    Let $\mathbb{P}_{\text{con}}$ be a distribution of matrices $\boldsymbol{B} \in \mathbb{F}_2^{n \times m}$ such that $\boldsymbol{B}^\intercal$ is drawn from the $\mathcal{G}(m, k, d)$ Gallager ensemble (Definition~\ref{def:gallager_ensemble}).
    Let $\epsilon = \frac{C \log(k)}{k}$, where $C$ is a constant independent of $k, d, m$.
    Then $\operatorname{DQI}_{\ell}$ is $(\epsilon n, 0, n^{-\Omega(1)})$-stable if $2\ell+1 \leq \delta m$, where $\delta = \Theta_k(1/d)$.
\end{corollary}
\begin{proof}
    Follows immediately from Theorems~\ref{thm:dqi_is_stable} and \ref{thm:gallager_is_restrictable}.
\end{proof}
\begin{remark}
    We note that local restrictability is not a property unique to the Gallager ensemble. Supporting this, in Appendix~\ref{sec:loc_rest_common_code_ens} we prove the local restrictability of other common code ensembles, the Bernoulli ensemble (Theorem~\ref{thm:Bernoulli_ensemble_good}) and the right-regular ensemble (Theorem~\ref{thm:right_regular_ensemble_good}). Analogous results to Corollary~\ref{cor:DQI_stable_Gallager} for these ensembles then also hold.
\end{remark}

\section{Algorithmic Implications and Comparison With a Classical Algorithm}\label{sec:numerics}

We here prove concrete bounds on the performance of DQI for tranposed Gallager \textsc{MAX-$k$-XOR-SAT}. We also here discuss the performance of approximate message passing (AMP), a classical algorithm widely believed to achieve the OGP threshold for combinatorial optimization problems.

To begin, we claim that the maximum attainable value of transposed Gallager \textsc{MAX-$k$-XOR-SAT} does not vanish in the $n\to\infty$ then $k\to\infty$ double limit. This was conjectured by Jordan et al.~\cite{jordan2025optimizationdecodedquantuminterferometry}; we prove this conjecture and give the exact form of the maximum attainable value in Appendix~\ref{sec:max_value_gallager_xor_sat}, and report the result here.
\begin{theorem}[Maximum value of transposed Gallager \textsc{MAX-$k$-XOR-SAT}, informal version of Theorem~\ref{thm:max_value}] \label{thm:max_value_gallager_INFORMAL}
    Define:
    \begin{equation}
        \theta^\ast:=1-\operatorname{H}_2^{-1}\left(1-\frac{1}{\lambda}\right).
    \end{equation}
    With probability exponentially close to $1$ in the $n\to\infty$ limit over the transposed Gallager ensemble of \textsc{MAX-$k$-XOR-SAT},
    \begin{equation}
        \theta^\ast-\exp\left(-\operatorname{\Omega}_k\left(k\right)\right)\leq\frac{1}{m}\max_{\bm{z}\in\mathbb{F}_2^n}\left(m-\left\lVert\bm{B}\bm{z}\oplus\bm{v}\right\rVert_1\right)\leq\theta^\ast.
    \end{equation}
\end{theorem}
In contrast, we show that the spin glass phase transition \emph{does} occur at an approximation ratio vanishing to the trivial value of $1/2$ in this double limit. In particular, we show in Methods that stable quantum algorithms are obstructed from optimizing beyond this approximation ratio.
\begin{theorem}[Stable quantum algorithms fail for \textsc{MAX-$k$-XOR-SAT}, informal version of Theorems~\ref{thm:chaos_prop} and~\ref{thm:stab_algs_fail}]\label{thm:stab_quant_algs_fail}
    Fix any $0<\nu<1/2$ and $0\leq\sigma\leq 1$, as well as any $f,L,p_{\text{st}}$ and $\gamma,p_{\text{f}}$ satisfying the following inequalities:
    \begin{align}
        p_{\text{st}}+p_{\text{f}}&\leq\frac{1}{12};\\
        \gamma>\frac{1}{2}+\left(1+\operatorname{o}_{\nu_2^{-1},\sigma^{-1}}\left(1\right)\right)&\sqrt{\frac{\nu_2\ln\left(1/\nu_2\right)+\sigma}{2\lambda}};\\
        f&\leq\left\lfloor\sigma n\right\rfloor;\\
        \lambda L\sqrt{\frac{2}{1-2\sqrt{3p_{\text{f}}}}}&<\frac{\nu}{2}.
    \end{align}
    There is no $\left(f,L,p_{\text{st}}\right)$-stable and $\left(\gamma,p_{\text{f}}\right)$-optimal quantum algorithm for tranposed Gallager \textsc{MAX-$k$-XOR-SAT}.
\end{theorem}
Combining this result with the stability of DQI over this ensemble (Corollary~\ref{cor:DQI_stable_Gallager}), we are able to show that DQI is topologically obstructed for \emph{any} choice of decoder decoding errors with relative weight at most $\frac{1}{k\lambda}$, far from the true optimum reported in Theorem~\ref{thm:max_value_gallager_INFORMAL}.
\begin{theorem}[DQI is topologically obstructed for transposed Gallager \textsc{MAX-$k$-XOR-SAT}]\label{thm:dqi_fails_gallager}
    Fix any:
    \begin{equation}\label{eq:ell_bound_dqi_fails}
        \frac{\ell}{m}\leq\frac{1}{k\lambda}.
    \end{equation}
    For sufficiently large $n$, $\operatorname{DQI}_\ell$ does not succeed in sampling a bit string $\bm{z}\in\mathbb{F}_2^n$ achieving a satisfied fraction $g(\bm{z})/m$ greater than
    \begin{align}
        \mu_{\mathrm{top}}& :=1-\operatorname{H}_2^{-1}\left(1-\left(1+\operatorname{o}_k\left(1\right)\right)\frac{\ln\left(k\right)}{k\lambda \ln 2}\right) \\
        & = \frac{1}{2}+\left(1+\operatorname{o}_k\left(1\right)\right)\sqrt{\frac{\ln\left(k\right)}{2\ln\left(2\right)k\lambda}}
    \end{align}
    with probability more than $12/13$ over both the randomness of the algorithm and $\mathbb{P}_{\mathrm{G}}$.
\end{theorem}
\begin{proof}
    Fix:
    \begin{align}\label{eq:epsilon_req}
        \epsilon&=\lambda\operatorname{H}_2\left(\frac{\ell}{m}\right)=\lambda\operatorname{H}_2\left(\frac{1}{k\lambda}\right) \\ &=\left(1+\operatorname{o}_k\left(1\right)\right)\frac{\log_2\left(k\right)}{k}.
    \end{align}
    Recall as well from Corollary~\ref{cor:DQI_stable_Gallager} that $\operatorname{DQI}_\ell$ is $\left(\epsilon n,0,n^{-\operatorname{\Omega}\left(1\right)}\right)$-stable over the Gallager ensemble at sufficiently large $k$ for this choice of $\epsilon$. The result then follows immediately from Theorem~\ref{thm:stab_quant_algs_fail} by fixing any $\sigma>\epsilon$, $p_{\text{f}}=1/13$, and $\nu$ sufficiently small.
\end{proof}
\begin{remark}
    Theorem~\ref{thm:dqi_fails_gallager} generalizes readily if one assumes the existence of a better decoder; this modifies the choice of $\epsilon$ due to Eq.~\eqref{eq:epsilon_req}. This yields another topological obstruction, though at a larger approximation ratio.
\end{remark}
Note that we prove the theorem for success probability $12/13$ for convenience, but the result holds for any constant probability.
This is because if the algorithm succeeded with any constant probability, we could run it in sequence a constant number of times, and the resultant sequentially repeated algorithm would likewise be stable and thus subject to the same topological obstruction.
Next, we show that DQI is topologically constructed when using any inverse Lipschitz decoder, which includes the belief propagation decoder as a special case. To obtain an explicit threshold we slightly modify the distribution of problem instances from the transposed Gallager ensemble of \textsc{MAX-$k$-XOR-SAT} to an ensemble known as the Poisson ensemble~\cite{jones_et_al_full} which is only local ``on average,'' though we supplement this result with a constant-$k$ bound in Methods where we show a barrier at some constant approximation ratio occurs but are unable to compute its value.

The bound we present here is expressed in terms of the \emph{$k$th algorithmic Parisi constant} $\operatorname{\normalfont\textsc{P}}_k^{\mathrm{ALG}}$, sometimes also called the \emph{algorithmic threshold of the $k$-spin model}. We refer the reader to \cite{el2021optimization} and \cite[Eq.~(2.7)]{alaoui2020algorithmicthresholdsmeanfield} for a definition. We also numerically compute and report the values of $\operatorname{\normalfont\textsc{P}}_k^{\mathrm{ALG}}$ for $k<20$ in Fig.~\ref{fig:parisi}, using code from \cite{marwaha2022boundsapproximating}. The result is proved in Appendix~\ref{sec:dqi_bp_b_ogp} as it follows almost identically to Theorem~\ref{thm:dqi_fails_gallager}.
\begin{restatable}[DQI with an inverse Lipschitz decoder is obstructed in Poisson \textsc{MAX-$k$-XOR-SAT}]{theorem}{dqibppoisson}\label{thm:bp_dqi_fails_ogp}
    Let $k$ be even and fix $\epsilon>0$. $\operatorname{DQI}_\ell$ with an inverse Lipschitz decoder is such that, for all sufficiently large $\lambda$, w.h.p.,
    \begin{align}
        \frac{g_{\bm{X}}\left(\operatorname{DQI}_\ell\left(\bm{X}\right)\right)}{m} & \leq\mu_{\mathrm{b-OGP}}+\frac{\epsilon}{\sqrt{\lambda}}\\
        & :=\frac{1}{2}+\frac{1}{2\sqrt{\lambda}}\operatorname{\normalfont\textsc{P}}_k^{\mathrm{ALG}}+\frac{\epsilon}{\sqrt{\lambda}}.
    \end{align}
\end{restatable}

\subsection{Comparison With Other Algorithms}

We first compare the topological obstruction demonstrated in Theorem~\ref{thm:dqi_fails_gallager} what DQI is known to optimally achieve in expectation with the same assumption on $\ell$~\cite[Theorem~4.1]{jordan2025optimizationdecodedquantuminterferometry}:
\begin{align}
    \frac{\left\langle g\right\rangle_{\mathrm{DQI}}}{m} & =\left(\sqrt{\frac{\ell}{2m}}+\sqrt{\frac{1}{2}\left(1-\frac{\ell}{m}\right)}\right)^2 \\
    & \leq\frac{1}{2}+\left(1+\operatorname{o}_k\left(1\right)\right)\sqrt{\frac{1}{k\lambda}}.
\end{align}
While this bound is of a slightly different nature than Theorem~\ref{thm:dqi_top_obs_intro}---which is a ``with high probability'' statement, rather than a statement only in expectation---it suggests that DQI does not perform optimally even among stable algorithms. This is because the bound shown in Theorem~\ref{thm:dqi_fails_gallager} is a function only of the stability of DQI, not any other details of the algorithm. Motivated by this observation, we show in Appendix~\ref{sec:qaoa_app_rat} that a depth-1 quantum algorithm known as QAOA~\cite{farhi2014quantumapproximateoptimizationalgorithm} achieves the same $k$ scaling as $\mu_{\mathrm{top}}$ for \textsc{MAX-$k$-XOR-SAT}, beyond what DQI can achieve in expectation:
\begin{align}
    \frac{\mathbb{E}_{\bm{v}\sim\mathbb{P}_{\mathrm{par}}}\left[\left\langle g\right\rangle_{\mathrm{QAOA}}\right]}{m} & \geq\frac{1}{2}+\left(1+\operatorname{o}_k\left(1\right)\right)\sqrt{\frac{\ln\left(k\right)}{4\ce k\lambda}} \\
    & >\frac{\left\langle g\right\rangle_{\mathrm{DQI}}}{m},
\end{align}
where the final inequality holds at sufficiently large $k$. Here, $\left\langle\cdot\right\rangle_{\mathrm{QAOA}}$ denotes an expectation over the randomness of the algorithm, and $\mathbb{E}_{\bm{v}}$ denotes an expectation over the randomness of the parity constraints $\bm{v}\sim\mathbb{P}_{\mathrm{par}}$.

We now consider what is optimally achieved by classical algorithms, and compare with both Theorems~\ref{thm:dqi_fails_gallager} and~\ref{thm:bp_dqi_fails_ogp}. It is widely believed that the classical algorithm known as approximate message passing (AMP) achieves the OGP threshold for combinatorial optimization problems~\cite{doi:10.1073/pnas.2108492118,chen2019,gamarnik2022algorithmsbarrierssymmetricbinary,10.1214/23-AAP1953,10.1002/cpa.22222,cheairi2024algorithmicuniversalitylowdegreepolynomials}. Much like DQI, AMP has the nice property that its optimal performance can be computed for Gaussian spin glass models without running the algorithm~\cite{alaoui2020algorithmicthresholdsmeanfield}. This analysis can be extended to \textsc{MAX-$k$-XOR-SAT} using the algorithmic universality result of~\cite{cheairi2024algorithmicuniversalitylowdegreepolynomials}, though unfortunately only over ensembles whose constraints are chosen independently (which is not the case for the Gallager ensemble). For instance, a constraint matrix with i.i.d.\ Bernoulli entries as considered in~\cite{cheairi2024algorithmicuniversalitylowdegreepolynomials} suffices, as well as the Poisson ensemble~\cite{jones_et_al_full} which we also considered in Theorem~\ref{thm:bp_dqi_fails_ogp}.
\begin{figure}
    \centering
    \includegraphics[width=0.9\linewidth]{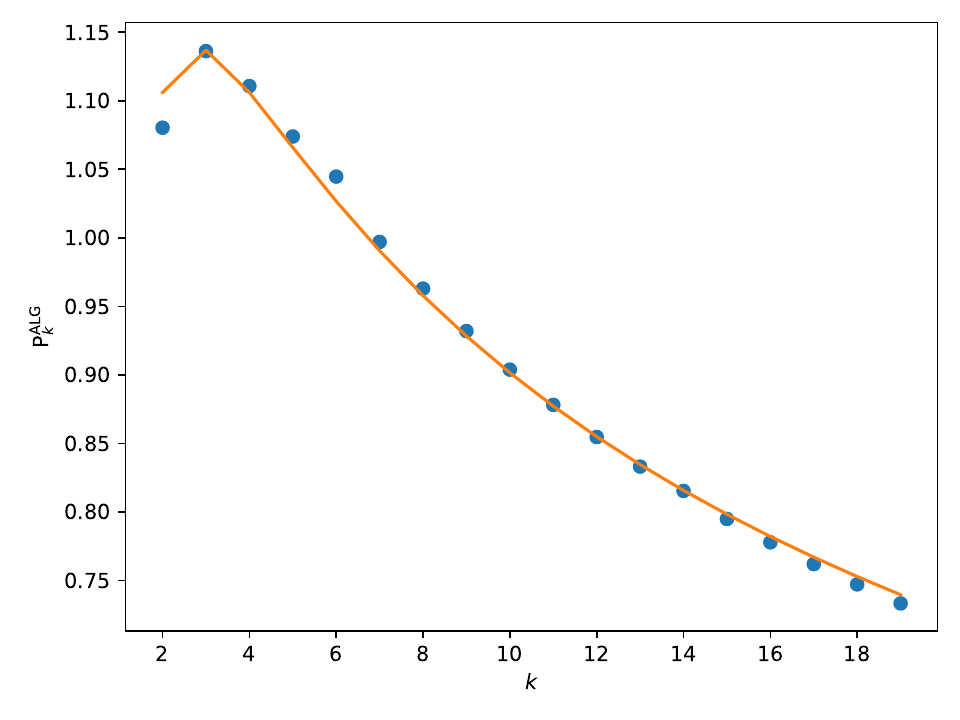}
    \caption{Numerically computed values for $\operatorname{\normalfont\textsc{P}}_k^{\mathrm{ALG}}$ for $k\in\left\{2,\ldots,19\right\}$ (blue dots). We also perform a least-square error fit to $\sqrt{\frac{c\ln\left(k\right)}{k}}$ (orange line), which we conjecture is the correct functional form at large $k$.}
    \label{fig:parisi}
\end{figure}
\begin{proposition}[Performance of AMP on the Bernoulli and Poisson ensemble~{\cite{el2021optimization,marwaha2022boundsapproximating,jones_et_al_full,cheairi2024algorithmicuniversalitylowdegreepolynomials}}]
    Consider either the Bernoulli or Poisson ensembles of \textsc{MAX-$k$-XOR-SAT}. Approximate message passing achieves a satisfied fraction:
    \begin{equation}
        \mu_{\mathrm{AMP}}:=\frac{1}{2}+\frac{1}{2\sqrt{\lambda}}\operatorname{\normalfont\textsc{P}}_k^{\mathrm{ALG}}
    \end{equation}
    with high probability over the randomness of the ensemble and the algorithm.
\end{proposition}
Note that this value is identical to the bound we show for the performance of DQI with inverse Lipschitz decoders in Theorem~\ref{thm:bp_dqi_fails_ogp}. Though we are unable to prove the satisfied fraction achieved by AMP outside of these ensembles, we conjecture it achieves the same value for the Gallager ensemble.

As $\operatorname{\normalfont\textsc{P}}_k^{\mathrm{ALG}}$ can be expressed as the minimum of a convex functional $\operatorname{\normalfont\textsc{P}}_k$ (the famous \emph{Parisi functional}) over a convex space, one can efficiently optimize the functional numerically~\cite{alaoui2020algorithmicthresholdsmeanfield,marwaha2022boundsapproximating}. This allows us to compare the performance of AMP not only against DQI with an inverse Lipschitz decoder such as belief propagation, but also against more general decoders satisfying the assumptions of Theorem~\ref{thm:dqi_fails_gallager}. We repurpose the code of \cite{marwaha2022boundsapproximating} to minimize $\operatorname{\normalfont\textsc{P}}_k$ over this domain. In Fig.~\ref{fig:parisi} we plot the calculated value of $\operatorname{\normalfont\textsc{P}}_k^{\mathrm{ALG}}$ as a function of $k$; our results at $k=2$ and $k=3$ agree with what was previously computed in the literature~\cite{alaoui2020algorithmicthresholdsmeanfield}. Note that while this is a convex optimization problem and therefore standard convex optimization algorithms should always yield the global optimum, we found in practice that it was very numerically unstable; for this reason we plot the minimum found value over $5$ runs of the optimization procedure for each data point over random initializations.

We also perform a least-squares fit of $\operatorname{\normalfont\textsc{P}}_k^{\mathrm{ALG}}$ using the conjectured functional form $\sqrt{\frac{c\ln\left(k\right)}{k}}$, yielding the numerical estimate $c\approx 3.530$. Taken together,
\begin{equation}
    \mu_{\mathrm{AMP}}\approx\frac{1}{2}+\sqrt{\frac{0.882\ln\left(k\right)}{k\lambda}}>\frac{\left\langle g\right\rangle_{\mathrm{DQI}}}{m},
\end{equation}
with the inequality holding at sufficiently large $k$. As AMP optimizes to what is known as the branching OGP threshold~\cite{el2021optimization}, we conjecture that $\mu_{\mathrm{AMP}}\geq\left\langle g\right\rangle_{\mathrm{DQI}}/m$ even at small $k$. Proof of this conjecture would follow from suitably tight bounds on the optimal OGP thresholds, which we leave for future work.

\section{Outlook}

Given the algorithmic obstructions implied by spin glass transitions, a key algorithmic question concerns techniques to overcome the spin glass phase transition barrier characterized by the overlap gap property.
Here, we show that this phase transition cannot be surpassed by DQI unless major advances are made in known decoding algorithms. One potential avenue for circumventing our no-go results consists of utilizing decoding strategies which are themselves inherently quantum.
Such decoders require a coherent superposition over errors with particular phases in exchange for improved thresholds, and since errors are themselves introduced coherently and intentionally in DQI, it may be possible to adjust DQI to become amenable to quantum decoders achieving thresholds with $k$-scalings beyond the stability barrier~\cite{chailloux2023quantum}.

More broadly, spin glass transitions and algorithmic stability serve as a rigorous diagnostic for the possibility of attaining quantum advantage via a given algorithmic proposal.
As many of the most widely considered quantum algorithms for optimization are also stable and thus subject to obstruction by spin glass phase transitions, the design of new quantum optimization algorithms that overcome these obstructions---as well as determining which optimization problems are both interesting and do not exhibit such phase transitions---remains a pressing open direction.

\section{Methods}

\subsection{Decoded Quantum Interferometry}\label{sec:dqi_rev}

The Decoded Quantum Interferometry (DQI) algorithm is a quantum reduction between the combinatorial optimization problem MAX-$k$-XOR-SAT and decoding a classical linear code~\cite{jordan2025optimizationdecodedquantuminterferometry}. 
We briefly describe the reduction, which is illustrated in Fig.~\ref{fig:DQI_flowchart}. 
A MAX-$k$-XOR-SAT instance of degree $d$ is specified by a pair $(\bm{B}, \bm{v})$ where $\bm{B} \in \mathbb{F}_2^{m \times n}$ and $\bm{v} \in \mathbb{F}_2^{m}$, for $m \geq n$. 
Here, $\bm{B}$ satisfies the constraint that the weight of every row (column) is at most $k$ ($d$). 
The task is to find $\bm{x} \in \mathbb{F}_2^{n}$ such that $\lVert \bm{Bx} - \bm{v} \rVert_1$ is minimized. 
We may define a slightly different objective function \begin{align}
    f(\bm{x}) :=  m - 2\left\lVert\bm{Bx}\oplus\bm{v}\right\rVert_1
\end{align}
which is equivalent to the previous objective function $g$ in the sense that a choice $\bm{x}$ maximizes $m-\lVert \bm{Bx} \oplus \bm{v} \rVert_1$ if and only if it maximizes $f(\bm{x})$. 
$f$ differs slightly from the function we previously considered, $g(\bm{x}) = m - \left\lVert\bm{Bx}\oplus\bm{v}\right\rVert_1$, but is equivalent to $g$ as an objective function.
We introduce $f$ because it arises naturally in the construction of DQI.
The goal of DQI is to sample solutions $\bm{x}$ from a distribution proportional to $\mathcal{P}^2(f(\bm{x}))$, where $\mathcal{P}$ is a univariate polynomial of a large degree $\ell$. 
In this way the values of $\bm{x}$ for which $f(\bm{x})$ is small have very low probability, whereas values of $\bm{x}$ for which $f(\bm{x})$ is large have comparatively a much larger probability. 
Such a sampling can be achieved by creating the quantum state \begin{align}
    \ket{\text{DQI}_{\ell}(f)} := \frac{1}{\mathcal{N}} \sum_{\bm{x} \in \mathbb{F}_2^n} \mathcal{P}(f(\bm{x})) |\bm{x}\rangle ,
\end{align}
and then subsequently measuring in the computational basis. 
Here, $1/\mathcal{N}$ is the normalization operator that enforces the correct normalization of the state.
Because $n \leq m$, $\bm{B}^\intercal \in \mathbb{F}_2^{n \times m}$ can be interpreted as the parity check matrix of a classical linear code. The central claim of DQI is that if there exists a decoder which can correctly decode all errors up to weight $\ell$ in the code given by $\bm{B}^\intercal$, then there is a quantum algorithm which samples from a distribution of solutions $\bm{x}$ with probability proportional to $\mathcal{P}^2(\bm{x})$, for any $\mathcal{P}$ such that $\deg(\mathcal{P}) \leq \ell$.
There are two key facts which together imply this reduction. 
The first is that 
\begin{theorem}[Symmetric polynomial expansion~\cite{jordan2025optimizationdecodedquantuminterferometry}] \label{thm:symmetric_polynomial_expansion}
    Let $\mathcal{P}$ be a univariate polynomial of degree $\ell$ and let $z_1, \ldots, z_m$ be formal variables satisfying $z_i^2 = 1$ for all $i$. 
    Then there exists coefficients $c_0, \ldots, c_\ell \in \mathbb{R}$, depending only on $\mathcal{P}$, such that \begin{align} \label{eq:symmetric_expansion}
        \mathcal{P}\left( \sum_{i=1}^m z_i \right) = \sum_{k=0}^{\ell} c_k \sum_{\bm{y} \in \mathbb{F}_2^{m} : |\bm{y}| = k} z_1^{y_1} \cdots z_m^{y_m} . 
    \end{align}
\end{theorem}
Denote by $\bm{b}_i$ the $i$th row of $\bm{B}$. 
Each row is known as a \textit{clause} and each corresponding $v_i$ is known as a \textit{target}.
If $\bm{b}_i \cdot \bm{x} = v_i$, then we say that the $i$th clause is \textit{satisfied}. 
Hence, $f$ counts the difference between the number of satisfied and unsatisfied clauses.
We can therefore re-express $f(\bm{x}) = \sum_{i=1}^m f_i(\bm{x})$, where \begin{align}
    f_i(\bm{x}) = (-1)^{\bm{b}_i \cdot \bm{x} + v_i} .
\end{align}
Applying Theorem~\ref{thm:symmetric_polynomial_expansion}, \begin{align}
    \mathcal{P}(f(\bm{x})) = \sum_{k=0}^{\ell} c_k \sum_{\bm{y} \in \mathbb{F}_2^{m} : |\bm{y}| = k} f_1^{y_1} \cdots f_m^{y_m} .
\end{align}
The second fact is that if we apply a Hadamard transform $\bm{H}^{\otimes n}$ to $\ket{\text{DQI}_{\ell}(f)}$, we find a state which is simpler to synthesize. By applying Eq.~(\ref{eq:symmetric_expansion}), 
\begin{align} \label{eq:target-state_DQI}
    \bm{H}^{\otimes n} \ket{\text{DQI}_{\ell}} & \propto \sum_{k=0}^{\ell} c_k \sum_{|\bm{y}| = k} (-1)^{\bm{v} \cdot \bm{y}} \ket{\bm{B}^\intercal \bm{y}}
\end{align}
Figure~\ref{fig:DQI_flowchart} describes in detail as to how the state in Eq.~(\ref{eq:target-state_DQI}) is synthesized by using a decoder for the classical linear code whose parity check matrix is given by $\bm{B}^\intercal \in \mathbb{F}_2^{n \times m}$.
Specifically, we apply the following steps.
\begin{enumerate}
    \item Define $w_k = c_k \sqrt{\binom{m}{k}}$. The state $\propto \sum_{k=0}^\ell w_k \ket{k}$ is synthesized, which can be done efficiently since $\ell = O(n)$ and so the register has only $O(\log n)$ qubits. 
    Next, create a Dicke state \begin{align}
       \ket{D^m_k} = \binom{m}{k}^{-1/2} \sum_{\bm{y} \in \mathbb{F}_2^m : |\bm{y}|=k} \ket{\bm{y}}
    \end{align} coherently, using methods such as that of \cite{bartschi2022short}. 
    The Hamming weight of every term in a Dicke state is the same, so we may compute it and thereby uncompute the $\ket{k}$ register. This leaves our initial resource state $\ket{\psi_2}$, which has no dependence on the specific problem instance $(\bm{B}, \bm{v})$.
    \item By applying $\bigotimes_{i=1}^{m} \bm{Z}^{v_i}$ on $\ket{\psi_2}$, we add a phase and create $\ket{\psi_3}$. We then add a new register and compute $\ket{\bm{y}} \ket{0} \mapsto \ket{\bm{y}} \ket{\bm{B}^\intercal \bm{y}}$.
    \item In the final step, we apply a decoder for the code $\bm{B}^\intercal$ which we assume is capable of deducing the error $\bm{y}$ from the syndrome $\bm{B}^\intercal \bm{y}$ for every $\bm{y}$ such that $\lVert \bm{y} \rVert_1 \leq \ell$. This decoder therefore enables the uncomputation of the $\ket{\bm{y}}$ register, which produces the state in Eq.~(\ref{eq:target-state_DQI}).
\end{enumerate}
After applying a final Hadamard transform, we achieve the final state $\bm{\mathcal{A}}\left(\bm{B},\bm{v}\right) = \ket{\text{DQI}_{\ell}(f)}\bra{\text{DQI}_{\ell}(f)}$.

\subsection{Topological Obstructions in Combinatorial Optimization Problems}

We begin by reviewing topological obstructions in optimization problems; in particular, we here discuss the \emph{overlap gap property} (OGP). The OGP is a topological property of the near-optimal space of optimization problems which is conjectured to tightly characterize their average-case hardness~\cite{doi:10.1073/pnas.2108492118,chen2019,gamarnik2022algorithmsbarrierssymmetricbinary,10.1214/23-AAP1953,10.1002/cpa.22222,cheairi2024algorithmicuniversalitylowdegreepolynomials}. It has also been previously seen to have implications in quantum settings, having been used to show specific quantum algorithms such as low-depth QAOA are inhibited from preparing near-optimal solutions to combinatorial optimization problems~\cite{farhi2020quantumapproximateoptimizationalgorithm,anschuetz2022critical,anschuetzkiani2022,9996946,Anshu2023concentrationbounds,chen2023localalgorithmsfailurelogdepth,goh2025overlapgappropertylimits,anschuetz2025unified,anschuetz2025efficientlearningimpliesquantum}, and even to show a weaker version of the no low-energy trivial states (NLTS) conjecture~\cite{anschuetz2024combinatorial}.

There are many different variants of the OGP. We will here be primarily concerned with the \emph{$R$-OGP} (also known as the \emph{multi-OGP})~\cite{rahman2017local,doi:10.1137/140989728}. For concreteness, we specialize the definition to the correlated ensemble of \textsc{MAX-$k$-XOR-SAT} instances we introduced in Definition~\ref{def:kappa_corr_ens}. Informally, the $R$-OGP is a statement that over $R$ correlated problem instances $g_{\bm{X}}$, solutions achieving a satisfied fraction of $\mu$ clauses are clustered w.h.p. Furthermore, near-optimal solutions from independent problem instances are far apart. We give an illustration of this phenomenon in the special case when $R=2$ in Fig.~\ref{fig:e_ogp}.

\begin{figure}
    \centering
    \includegraphics[width=0.8\linewidth]{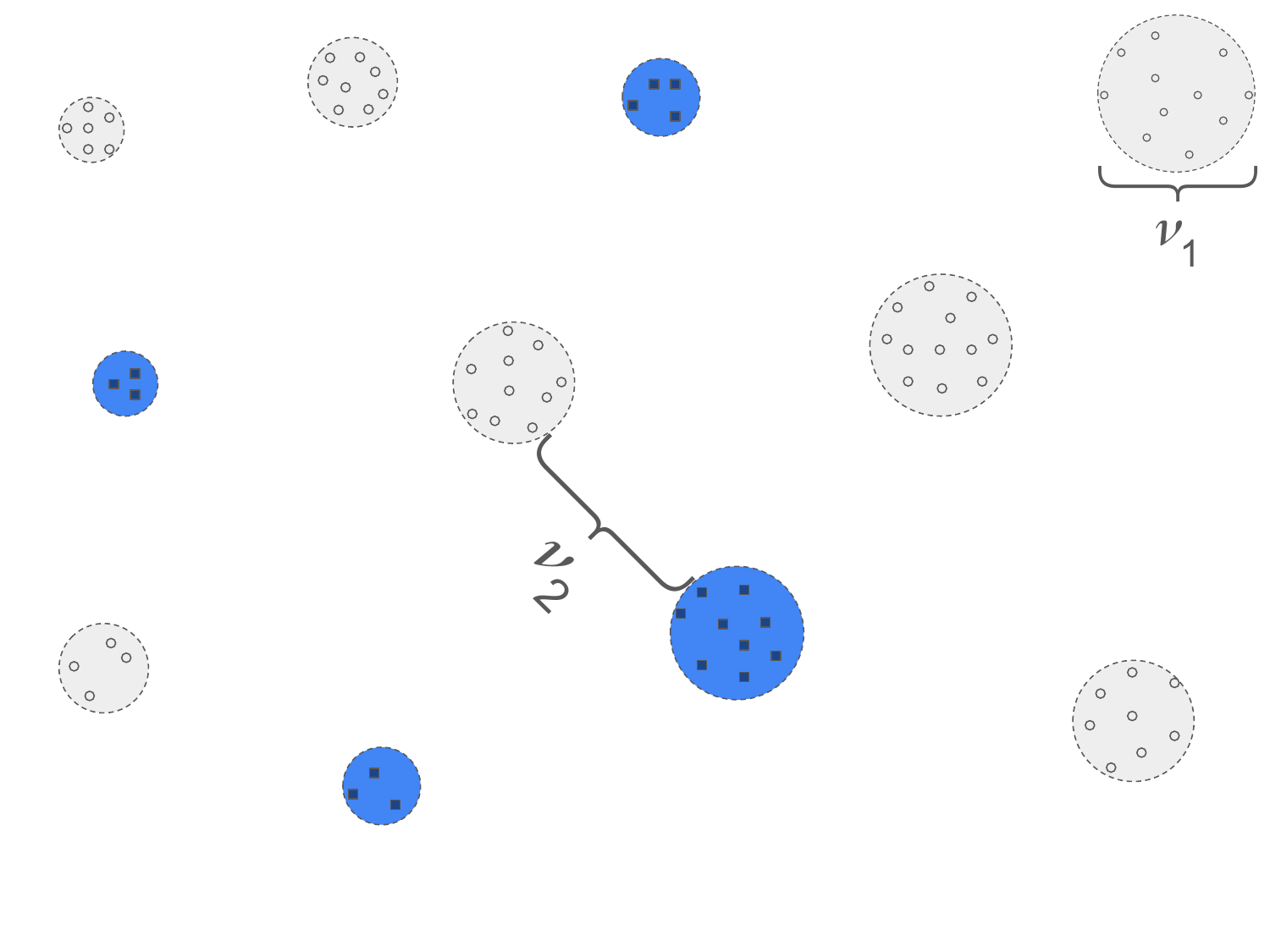}
    \caption{An illustration of the near-optimal solution space of a problem satisfying the $R$-OGP when $R=2$. For two potentially-correlated problem instances (gray circles and blue squares), tuples (small circles and squares) are either close in the semimetric $n^{-1}d_{k,S}$ (within $\nu_1$) or far apart (beyond $\nu_2$). The dotted large circles are a guide to the eye.}
    \label{fig:e_ogp}
\end{figure}
As the $R$-OGP is defined with respect to correlated problem instances, we now introduce a distribution of correlated \textsc{MAX-$k$-XOR-SAT} instances whose marginals are identical to the transposed Gallager ensemble $\mathbb{P}_{\mathrm{G}}$. For each integer $0\leq p\leq m$ define the $m\times m$ projector:
\begin{equation}\label{eq:upsilon_proj_def}
    \bm{\varUpsilon}_p:=\operatorname{diag}\left(\underbrace{1,\ldots,1}_p,\underbrace{0,\ldots,0}_{m-p}\right).
\end{equation}
\begin{definition}[$\kappa$-correlated ensemble]\label{def:kappa_corr_ens}
    Fix $\kappa\in\left[0,1\right]$, draw $\left(\bm{B},\bm{v}^{\left(1\right)}\right)\sim\mathbb{P}_{\mathrm{G}}$ and $\left\{\bm{\tilde{v}}^{\left(r\right)}\sim\mathbb{P}_{\mathrm{par}}\right\}_{r=2}^R$ independently, and let:
    \begin{equation}
        \bm{v}^{\left(r\right)}:=\bm{v}^{\left(1\right)}\oplus\bm{\varUpsilon}_{\left\lfloor\kappa m\right\rfloor}\left(\bm{v}^{\left(1\right)}\oplus\bm{\tilde{v}}^{\left(r\right)}\right).
    \end{equation}
    We call the joint distribution of $\left(\left(\bm{B},\bm{v}^{\left(r\right)}\right)\right)_{r=1}^R$ when drawn in this way the \emph{$\left(\kappa,R\right)$-correlated ensemble} $\mathbb{P}_R^{\left(\kappa\right)}$. When $R=2$, we simply say the $\kappa$-correlated ensemble $\mathbb{P}_2^{\left(\kappa\right)}$.
\end{definition}
Informally, these are $R$-tuples of \textsc{MAX-$k$-XOR-SAT} instances with identical constraint matrices, and whose parities are independent within their first $\kappa$-fraction of parities. By construction, $\mathbb{P}_R^{\left(\kappa\right)}$ marginalizes to $\mathbb{P}_{\mathrm{G}}$ in any variable for any $\kappa\in\left[0,1\right]$.

While typically the notion of distance in the OGP is given by the Hamming distance, for technical reasons we will find it much more convenient to work with what we call the \emph{$k$-minimum Hamming semimetric}. To define this, we first define the $n\times n$ projectors $\bm{\varGamma}_i$, $i\in\left[k\right]$ onto $n/k$-dimensional subspaces:
\begin{equation}\label{eq:gamma_proj_def}
    \bm{\varGamma}_i:=\operatorname{diag}\left(\underbrace{0,\ldots,0}_{\frac{\left(i-1\right)n}{k}},\underbrace{1,\ldots,1}_{\frac{n}{k}},\underbrace{0,\ldots,0}_{\left(1-\frac{i}{k}\right)n}\right).
\end{equation}
Given a subset $S\subset\left[n\right]$, we also define the $n\times n$ projector onto bits labeled by indices not in $S$:
\begin{equation}
    \bm{\varSigma}_S:=\operatorname{diag}\left(\left(\bm{1}\left\{i\notin S\right\}\right)_{i=1}^n\right).
\end{equation}
We now define the semimetric. This is formally a semimetric as it does not satisfy the triangle inequality, though we will later see that it still has nice properties. In its definition we use $\bm{1}\in\mathbb{F}_2^n$ to denote the all-$1$s vector.
\begin{definition}[$\left(k,S\right)$-minimum Hamming semimetric]\label{def:k_min_ham_semi}
    The \emph{$\left(k,S\right)$-minimum Hamming semimetric} between $\bm{z},\bm{z'}\in\mathbb{F}_2^n$ is:
    \begin{widetext}
        \begin{equation}
        d_{k,S}\left(\bm{z},\bm{z'}\right):=\sum_{i=1}^k\min\left(\left\lVert\bm{\varSigma}_S\bm{\varGamma}_i\left(\bm{z}\oplus\bm{z'}\right)\right\rVert_1,\left\lVert\bm{\varSigma}_S\bm{\varGamma}_i\left(\bm{z}\oplus\bm{z'}\oplus\bm{1}\right)\right\rVert_1\right).
    \end{equation}
    \end{widetext}
\end{definition}

We are now equipped to formally define the $R$-OGP as we use it here, specialized toward \text{MAX-$k$-XOR-SAT} over the transposed Gallager ensemble with clause density $\lambda$ and $k$-variable constraints. We consider a slightly more general version which will give us tighter bounds, but the usual multi-OGP definition can be recovered by taking $\sigma=0$ and changing the $d_{k,S}$ semimetric to the usual Hamming distance.
\begin{definition}[$R$-OGP]\label{def:multi_ogp}
    Fix any $R\geq 2\in\mathbb{N}$, $0\leq\nu_1<\nu_2<1/2$, and $0\leq\sigma\leq 1$. If there exists a $\mu^\ast$ such that the following holds for all $\kappa\in\left[0,1\right]$ and functions $S:\mathbb{F}_2^{m\times n}\to\binom{\left[n\right]}{\left\lfloor\sigma n\right\rfloor}$, then $\bm{g}_{\bm{X}}$ is said to satisfy the $R$-OGP with parameters $\left(R,\mu^\ast,\nu_1,\nu_2,\sigma\right)$.

    For any $\mu>\mu^\ast$, the quantity
    \begin{equation}
        \mathbb{P}_{\left(\bm{B},\bm{v}^{\left(r\right)}\right)_{r=1}^R\sim\mathbb{P}_R^{\left(\kappa\right)}}\left[\mathcal{S}_{\left(\bm{B},\bm{v}^{\left(r\right)}\right)_{r=1}^R}^{\left(R,\mu,\nu_1,\nu_2,S_{\bm{B}}\right)}\neq\varnothing\right]
    \end{equation}
    is at most $\exp\left(-\operatorname{\Omega}\left(n\right)\right)$,
    where $\mathcal{S}_{\left(\bm{B},\bm{v}^{\left(r\right)}\right)_{r=1}^R}^{\left(R,\mu,\nu_1,\nu_2,S_{\bm{B}}\right)}$ is the random set of $R$-tuples $\left(\bm{z}^{\left(r\right)}\right)_{r=1}^R\in\left\{-1,1\right\}^{R\times n}$ satisfying:
    \begin{enumerate}
        \item \textit{$\mu$-SAT fraction:} For all $r\in\left[R\right]$,
        \begin{equation}
            g_{\bm{X}^{\left(r\right)}}\left(\bm{z}^{\left(r\right)}\right)\geq\mu\lambda n.
        \end{equation}
        \item \textit{$\left(k,S\right)$-minimum Hamming semimetric bound:} Recalling the $\left(k,S\right)$-minimum Hamming semimetric (Definition~\ref{def:k_min_ham_semi}), for all $r\neq r'\in\left[R\right]$,
        \begin{equation}
            d_{k,S_{\bm{B}}}\left(\bm{z}^{\left(r\right)},\bm{z}^{\left(r'\right)}\right)\in\left[\nu_1 n,\nu_2 n\right].
        \end{equation}
    \end{enumerate}
\end{definition}

We will also consider the \emph{chaos property}, a weaker property than the $R$-OGP that only requires Eq.~\eqref{eq:chaos_property} be satisfied in Definition~\ref{def:multi_ogp}. The reason for this is two-fold:
\begin{enumerate}
    \item As it is weaker than the $R$-OGP, it is easier to prove; in Appendix~\ref{sec:proof_of_chaos_prop} we will compute the chaos property threshold exhibited by \textsc{MAX-$k$-XOR-SAT} for any choice of $R$, though we are only able to compute the threshold for the full $R$-OGP when $R=2$.
    \item While it yields a weaker topological obstruction than the full overlap gap property, it will suffice for obstructing DQI.
\end{enumerate}
For ease of reference later, we give a full, formal definition for the chaos property.
\begin{definition}[Chaos property]\label{def:chaos_prop}
    Fix any $R\geq 2\in\mathbb{N}$, $0<\nu_2<1/2$, and $0\leq\sigma\leq 1$. If there exists a $\mu^\ast$ such that the following holds for all functions $S:\mathbb{F}_2^{m\times n}\to\binom{\left[n\right]}{\left\lfloor\sigma n\right\rfloor}$, $\bm{g}_{\bm{X}}$ is said to satisfy the chaos property with parameters $\left(R,\mu^\ast,\nu_2,\sigma\right)$.

    For any $\mu>\mu^\ast$, the quantity
    \begin{equation}
        \mathbb{P}_{\left(\bm{B},\bm{v}^{\left(r\right)}\right)_{r=1}^R\sim\mathbb{P}_2^{\left(1\right)}}\left[\mathcal{S}_{\left(\bm{B},\bm{v}^{\left(r\right)}\right)_{r=1}^R}^{\left(R,\mu,0,\nu_2,S_{\bm{B}}\right)}\neq\varnothing\right]
    \end{equation}
    is at most $\exp\left(-\operatorname{\Omega}\left(n\right)\right)$,
    where $\mathcal{S}_{\left(\bm{B},\bm{v}^{\left(r\right)}\right)_{r=1}^R}^{\left(R,\mu,0,\nu_2,S_{\bm{B}}\right)}$ is the random set of $R$-tuples $\left(\bm{z}^{\left(r\right)}\right)_{r=1}^R\in\left\{-1,1\right\}^{R\times n}$ satisfying:
    \begin{enumerate}
        \item \textit{$\mu$-SAT fraction:} For all $r\in\left[R\right]$,
        \begin{equation}
            g_{\bm{X}^{\left(r\right)}}\left(\bm{z}^{\left(r\right)}\right)\geq\mu\lambda n.
        \end{equation}
        \item \textit{$\left(k,S\right)$-minimum Hamming semimetric bound:} Recalling the $\left(k,S\right)$-minimum Hamming semimetric (Definition~\ref{def:k_min_ham_semi}), for all $r\neq r'\in\left[R\right]$,
        \begin{equation}
            d_{k,S_{\bm{B}}}\left(\bm{z}^{\left(r\right)},\bm{z}^{\left(r'\right)}\right)\leq\nu_2 n.
        \end{equation}
    \end{enumerate}
    We emphasize that in the chaos property definition, the correlation parameter $\kappa$ is set to $1$.
    That is, unlike in the $R$-OGP, the $R$ instances are independent.
\end{definition}

\subsection{Stable Quantum Algorithms}\label{sec:stab_qas}

We now review \emph{stable quantum algorithms}. These are quantum algorithms which are Lipschitz in their inputs with respect to the \emph{quantum Wasserstein distance} on quantum states~\cite{anschuetz2025efficientlearningimpliesquantum}. In what follows we write out the definition of stability when specialized to the case of the Gallager ensemble of constraint matrices. We use $\mathcal{S}_n$ to denote the set of $n$-qubit pure states, and $\mathcal{S}_n^{\text{m}}$ to denote the set of $n$-qubit mixed states. We also use $D:=m\left(n+1\right)$ as shorthand for the total number of variables defining a \textsc{MAX-$k$-XOR-SAT} instance, i.e., the dimension of the vector $\operatorname{vec}\left(\bm{B}\right)\oplus\bm{v}$.

We consider in what follows a quantum algorithm $\bm{\mathcal{A}}$ to be a map from a space of inputs $\mathbb{F}_2^D$ as well as a probability space $\left(\varOmega,\mathbb{P}_\varOmega\right)$ to the set of quantum states $\mathcal{S}_n^{\text{m}}$. The addition of a probability space allows us to consider randomized quantum algorithms on the same footing as deterministic quantum algorithms. While in principle one can incorporate the randomness directly into $\bm{\mathcal{A}}$, it turns out doing so has implications on the stability parameters; see~\cite{anschuetz2025efficientlearningimpliesquantum} for more discussion on this technicality.

We now have all of the ingredients to define stability when specialized to \text{MAX-$k$-XOR-SAT}.
Informally, we say a quantum algorithm is stable if with high probability the algorithm is a Lipshitz function of the \textsc{MAX-$k$-XOR-SAT} parities $\bm{v}$, except perhaps within a small bad subset of qubits.
\begin{definition}[Stable quantum algorithm]\label{def:stable_qas}
    Let $\bm{\mathcal{A}}:\mathbb{F}_2^D\times\varOmega\to\mathcal{S}_n^{\mathrm{m}}$ be a quantum algorithm with associated probability space $\left(\varOmega,\mathbb{P}_\varOmega\right)$. Furthermore, let $\mathcal{K}\subset\left[0,1\right]$ be a set of $\kappa$'s, each labeling a $\kappa$-correlated  instances as defined in Definition~\ref{def:kappa_corr_ens}.

    $\bm{\mathcal{A}}$ is said to be \emph{$\left(f,L,\mathcal{K},p_{\mathrm{st}}\right)$-stable} if, for all $\kappa\in\mathcal{K}$,
    \begin{widetext}
    \begin{equation}
        \mathbb{P}_{\left(\left(\bm{B},\bm{v}\right),\left(\bm{B},\bm{v'}\right),\omega\right)\sim\mathbb{P}_2^{\left(\kappa\right)}\otimes\mathbb{P}_\varOmega}\left[\exists S_{\bm{B}}\in\binom{\left[n\right]}{f}:\left\lVert\Tr_{S_{\bm{B}}}\left(\bm{\mathcal{A}}\left(\bm{B},\bm{v},\omega\right)-\bm{\mathcal{A}}\left(\bm{B},\bm{v'},\omega\right)\right)\right\rVert_{W_2}\leq L\left\lVert\bm{v}-\bm{v'}\right\rVert_1\right]\geq 1-p_{\mathrm{st}}.
    \end{equation}
    \end{widetext}
    If $\mathcal{K}$ can be arbitrary, we say $\bm{\mathcal{A}}$ is $\left(f,L,p_{\mathrm{st}}\right)$-stable. Furthermore, if $\bm{\mathcal{A}}$ is a constant function of $\omega\in\varOmega$, we say $\bm{\mathcal{A}}$ is deterministic.
\end{definition}
Here, $\left\lVert\cdot\right\rVert_{W_2}$ induces the quantum Wasserstein distance of order $2$~\cite{anschuetz2025efficientlearningimpliesquantum} which we review in Appendix~\ref{sec:quant_wass_dist}. In~\cite{anschuetz2025efficientlearningimpliesquantum}, it was shown that many standard quantum algorithms are stable according to Definition~\ref{def:stable_qas}, including logarithmic-depth variational quantum algorithms, phase estimation, and logarithmic-depth Lindbladian dynamics.

We also here review the concept of a near-optimal quantum algorithm~\cite{anschuetz2025efficientlearningimpliesquantum}. We specialize to the case when the output is a classical state as we are here interested in the hardness of classical problems; this will allow us to directly incorporate the effects of shot noise in our definition. In particular, let $\mathcal{M}:\mathcal{S}_n^{\text{m}}\to\operatorname{Conv}\left(\mathcal{B}\right)$ be the channel denoting measurement in the computational basis, where $\mathcal{B}$ is the set of pure computational basis states on $n$ qubits and $\operatorname{Conv}\left(\cdot\right)$ denotes the convex hull. We can equivalently consider this as a stochastic map to pure computational basis states $\bm{\widetilde{\mathcal{M}}}:\mathcal{S}_n^{\text{m}}\times\left[0,1\right]\to\mathcal{B}$, where:
\begin{equation}\label{eq:comp_meas_chan_def}
    \mathbb{E}_{\upsilon\sim\mathcal{U}}\left[\bm{\widetilde{\mathcal{M}}}\left(\bm{\rho},\upsilon\right)\right]=\bm{\mathcal{M}}\left(\bm{\rho}\right),
\end{equation}
with $\mathcal{U}$ denoting the uniform distribution over $\left[0,1\right]$. We then have the following definition, where for simplicity we slightly abuse notation and write $g_{\bm{X}}$ as a map $\mathcal{B}\to\mathbb{Z}$ rather than $\mathbb{F}_2^n\to\mathbb{Z}$.
\begin{definition}[Near-optimal quantum algorithm]\label{def:no_qas}
    Let $\bm{\mathcal{A}}:\mathbb{F}_2^D\times\varOmega\to\mathcal{S}_n^{\text{m}}$ be a quantum algorithm with associated probability space $\left(\varOmega,\mathbb{P}_\varOmega\right)$, and let $\bm{\widetilde{\mathcal{M}}}$ be as in Eq.~\eqref{eq:comp_meas_chan_def}. Then, $\bm{\mathcal{A}}$ is said to be \emph{$\left(\gamma,p_{\text{f}}\right)$-optimal} for $g_{\bm{X}}$ over $\mathbb{P}_{\mathrm{G}}$ if
    \begin{widetext}
        \begin{equation}
        \mathbb{P}_{\left(\bm{X},\omega,\upsilon\right)\sim\mathbb{P}_{\mathrm{G}}\otimes\mathbb{P}_\varOmega\otimes\mathcal{U}}\left[g_{\bm{X}}\left(\bm{\widetilde{\mathcal{M}}}\left(\bm{\mathcal{A}}\left(\bm{X},\omega\right),\upsilon\right)\right)\geq\gamma\lambda n\right]\geq 1-p_{\text{f}}.
    \end{equation}
    \end{widetext}
\end{definition}
Informally, a quantum algorithm is $\left(\gamma,p_{\text{f}}\right)$-optimal if it satisfies a fraction $\gamma$ of clauses with probability $1-p_{\text{f}}$ over both the randomness of the algorithm and the randomness of drawing problem instances from $\mathbb{P}_{\mathrm{G}}$.

\subsection{Stable Quantum Algorithms Fail to Optimize \texorpdfstring{\textsc{MAX-$k$-XOR-SAT}}{MAX-k-XOR-SAT}}\label{sec:stable_algs_fail}

In this section we show that \textsc{MAX-$k$-XOR-SAT} exhibits a spin glass phase transition over the transposed Gallager ensemble. We then provide a proof sketch that stable quantum algorithms cannot prepare states beyond this transition, allowing us to prove average-case hardness results for \textsc{MAX-$k$-XOR-SAT} at fixed $k$ for the first time for both quantum and classical stable algorithms.

We first claim that the transposed Gallager ensemble of \textsc{MAX-$k$-XOR-SAT} exhibits the chaos property, which will suffice to later show that DQI fails to prepare near-optimal states of this ensemble. We provide a proof sketch here and leave the full proof for Appendix~\ref{sec:proof_of_chaos_prop}, where we also prove that transposed Gallager \textsc{MAX-$k$-XOR-SAT} exhibits an overlap gap property for completeness.
\begin{restatable}[Transposed Gallager \textsc{MAX-$k$-XOR-SAT} exhibits the chaos property]{theorem}{tgxorsatchaosprop}\label{thm:chaos_prop}
    Fix any $0<\nu_2<1/2$, $R\in\mathbb{N}$, and $0\leq\sigma\leq 1$. Transposed Gallager \textsc{MAX-$k$-XOR-SAT} satisfies the chaos properties with parameters $\left(R,\mu,\nu_2,\sigma\right)$ for all $\mu$ greater than:
    \begin{equation}
        1-\operatorname{H}_2^{-1}\left(1-\frac{1}{R\lambda}-\frac{1}{\lambda}\left(1-\frac{1}{R}\right)\left(\operatorname{H}_2\left(\nu_2\right)+\sigma\right)\right).
    \end{equation}
    In particular, choosing:
    \begin{equation}
        R=\max\left(\left\lceil\nu_2^{-2}\right\rceil,\left\lceil\sigma^{-2}\right\rceil\right),
    \end{equation}
    the chaos property holds for any:
    \begin{equation}
        \begin{aligned}
            \mu&>1-\operatorname{H}_2^{-1}\left(1-\left(1+\operatorname{o}_{\nu_2^{-1},\sigma^{-1}}\left(1\right)\right)\left(\frac{\operatorname{H}_2\left(\nu_2\right)+\sigma}{\lambda}\right)\right)\\
            &=\frac{1}{2}+\left(1+\operatorname{o}_{\nu_2^{-1},\sigma^{-1}}\left(1\right)\right)\sqrt{\frac{\nu_2\ln\left(1/\nu_2\right)+\sigma}{2\lambda}}.
        \end{aligned}
    \end{equation}
\end{restatable}
\begin{proof}[Proof sketch]
    We show the bound holds uniformly for any fixed $\bm{B}$. A simple combinatorial argument shows that the number $Z$ of $\left(\bm{z}^{\left(r\right)}\right)_{r=1}^R\in\left\{-1,1\right\}^{R\times n}$ satisfying $d_{k,S_{\bm{B}}}\left(\bm{z}^{\left(r\right)},\bm{z}^{\left(r'\right)}\right)\leq\nu_2 n$ for all $r\neq r'\in\left[R\right]$ is bounded by:
    \begin{equation}
        Z\leq\exp_2\left(n+R\left(\operatorname{H}_2\left(\nu_2\right)+\sigma\right)n\right).
    \end{equation}
    Now, let $X_r$ be the random variable counting the number of entries for which $\bm{v}^{\left(r\right)}$ equals $\bm{B}\bm{z}^{\left(r\right)}$. As each $\bm{v}^{\left(r\right)}$ has i.i.d.\ Bernoulli random entries and the $\bm{v}^{\left(r\right)}$ are independent, $X_r$ are i.i.d.\ binomial random variables. A standard tail bound then says:
    \begin{widetext}
        \begin{equation}
        \mathbb{P}_{\left(\bm{v}^{\left(r\right)}\right)_{r=1}^R}\left[\bigcap_{r=1}^R\left\lVert\bm{B}\bm{z}^{\left(r\right)}\oplus\bm{v}^{\left(r\right)}\right\rVert_1\leq\left(1-\mu\right)\lambda n\right]\leq\exp_2\left(R\operatorname{H}_2\left(\mu\right)\lambda n-R\lambda n+\operatorname{O}\left(\log\left(n\right)\right)\right).
    \end{equation}
    \end{widetext}
    Union bounding over configurations then gives a requirement:
    \begin{equation}
        1+R\left(\operatorname{H}_2\left(\nu_2\right)+\sigma\right)+R\lambda\operatorname{H}_2\left(\mu\right)-R\lambda<0
    \end{equation}
    for the probability to be exponentially vanishing. The result then follows by solving for $\mu$.
\end{proof}

We now show that stable quantum algorithms fail to optimize transposed Gallager \textsc{MAX-$k$-XOR-SAT} at approximation ratios for which the model exhibits an overlap gap property. We also show that the weaker chaos property inhibits optimization by quantum algorithms with even stricter stability properties, which we will later see includes DQI. The full technical details of the proof are provided in Appendix~\ref{sec:ogp_implies_alg_hardness}.
\begin{restatable}[Stable Quantum Algorithms Fail for Transposed Gallager \textsc{MAX-$k$-XOR-SAT}]{theorem}{stableqasfailgallager}\label{thm:stab_algs_fail}
    Fix any $\left(R,\mu,\nu_1,\nu_2,\sigma\right)$ such that \textsc{MAX-$k$-XOR-SAT} satisfies the multi-OGP with these parameters (Definition~\ref{def:multi_ogp}). Fix any $Q\in\mathbb{N}$. There exists a constant $\delta>0$ depending only on $R$ and $Q$ such that the following holds at sufficiently large $n$. Fix $f,L,p_{\text{st}}$ and $\gamma,p_{\text{f}}$ satisfying the following inequalities:
    \begin{align}
        p_{\text{st}}+p_{\text{f}}&\leq\delta;\\
        \gamma&\geq\mu;\\
        f&\leq\left\lfloor\sigma n\right\rfloor;\\
        \frac{\lambda L}{Q}\sqrt{\frac{Q+1}{1-\left(Q+1\right)\sqrt{3p_{\text{f}}}}}&<\frac{\nu_2-\nu_1}{2}.
    \end{align}
    There is no $\left(f,L,p_{\text{st}}\right)$-stable and $\left(\gamma,p_{\text{f}}\right)$-optimal quantum algorithm for tranposed Gallager \textsc{MAX-$k$-XOR-SAT}.

    Similarly, fix any $\left(R,\mu,\nu_2,\sigma\right)$ such that \textsc{MAX-$k$-XOR-SAT} satisfies the chaos property with these parameters (Definition~\ref{def:chaos_prop}). Fix $f,L,p_{\text{st}}$ and $\gamma,p_{\text{f}}$ satisfying the following inequalities:
    \begin{align}
        p_{\text{st}}+p_{\text{f}}&\leq\frac{1}{12};\\
        \gamma&\geq\mu;\\
        f&\leq\left\lfloor\sigma n\right\rfloor;\\
        \lambda L\sqrt{\frac{2}{1-2\sqrt{3p_{\text{f}}}}}&<\frac{\nu_2}{2}.
    \end{align}
    There is no $\left(f,L,p_{\text{st}}\right)$-stable and $\left(\gamma,p_{\text{f}}\right)$-optimal quantum algorithm for tranposed Gallager \textsc{MAX-$k$-XOR-SAT}.
\end{restatable}
\begin{proof}[Proof sketch]
    At a high level, our strategy follows that of~\cite[Theorem~16]{anschuetz2025efficientlearningimpliesquantum}, but differs in a few key ways:
    \begin{itemize}
        \item As we only consider classical problems, the proof can be simplified substantially.
        \item As the eigenbasis of the problem is fixed, we can account for shot noise in the failure probability of the algorithm.
        \item The $\left(k,S\right)$-minimum Hamming semimetric (Definition~\ref{def:k_min_ham_semi}) which we use to define (and later prove) our OGP and chaos property does not satisfy the triangle inequality, introducing complications to the proof.
        \item $f$ no longer need depend on $Q$, which greatly improves the satisfied fraction at which DQI (or any other stable algorithm with nontrivial $f$) is obstructed. This is achieved by slightly modifying the definition of stability from~\cite{anschuetz2025efficientlearningimpliesquantum}, which requires further modifications to the proof.
        \item Our strategy requires a novel strategy to interpolate between correlated problem instances as we are no longer considering Gaussian randomness.
    \end{itemize}

    We now give a sketch of our approach. We proceed by contradiction, assuming there exists a quantum algorithm $\bm{\mathcal{A}}$ that is $\left(f,L,p_{\text{st}}\right)$-stable and $\left(\gamma,p_{\text{f}}\right)$-optimal for \textsc{MAX-$k$-XOR-SAT} instances drawn from $\mathbb{P}_{\mathrm{G}}$.
    \begin{enumerate}
        \item We show that the existence of $\bm{\mathcal{A}}$ implies the existence of a stable, near-optimal quantum algorithm which is also deterministic. This follows by a pigeonhole principle-like argument to show that there must exist some good random seed $\omega\in\varOmega$ for which $\bm{\mathcal{A}}$ has only slightly worse stability parameters.
        \item We show that the existence of a stable, near-optimal, deterministic quantum algorithm for this problem implies the existence of a stable, near-optimal classical algorithm $\bm{\mathcal{I}}$ (at the cost of worse constants). This follows by defining $\bm{\mathcal{I}}$ as composing $\bm{\mathcal{A}}$ with a dephasing channel. By the locality of the dephasing channel, this algorithm does not change the stability parameters by much.
        \item We consider an interpolation path over many replicas, and show that $\bm{\mathcal{I}}$ is stable and near-optimal over the replicas with high probability. The path is constructed by successively resampling steps of the Gallager construction as described in Sec.~\ref{sec:ldpc_codes}. The result then follows by a union bound.
        \item Due to \textsc{MAX-$k$-XOR-SAT} with independent parities having distant solutions (as implied by the chaos property, Definition~\ref{def:chaos_prop}), we show that with high probability all $R$-tuples of $T$ independently-sampled instances have near-optimal states which are distant in $\left(k,S\right)$-minimum Hamming semimetric with high probability. Just as the previous step, this follows from a simple union bound and the definition of the chaos property.
        \item We show that with high probability there exists some point along the interpolation path where this algorithm outputs a configuration disallowed by the multi-OGP due to the pairwise-stability and near-optimality of $\bm{\mathcal{I}}$. In particular, by the previous step, $\bm{\mathcal{I}}$ must output distant solutions at the final point of the interpolation path---when the problem instances are independent---and must output identical instances at the initial point of the interpolation path, when all of the problem instances are identical. By the Lipschitzness of $\bm{\mathcal{I}}$, it must therefore output solutions in the disallowed subspace at some point along the interpolation path; in other words, $\bm{\mathcal{I}}$ is unable to ``jump'' between clusters along the interpolation path due to its stability.
        \item Finally, we show that there exist choices of parameters such that all ``with high probability'' events have a nontrivial intersection. This contradicts the assumption of the existence of $\bm{\mathcal{A}}$.
    \end{enumerate}
\end{proof}
Finally, to supplement Theorem~\ref{thm:bp_dqi_fails_ogp} of the main text, we show directly that DQI with an inverse Lipschitz decoder is obstructed by the multi-OGP of the transposed Gallager ensemble, though unfortunately we are unable to compute explicit threshold.
\begin{theorem}[DQI with an inverse Lipschitz decoder is obstructed by the overlap gap property]\label{thm:inv_lip_const_k_ogp_bound}
    Consider $\operatorname{DQI}_\ell$ with an $L$-inverse Lipschitz decoder (Definition~\ref{def:l_inv_lip_dec}) for any $n$-independent choice of $L$. Let $\left(R,\mu_{\mathrm{OGP}},\nu_1,\nu_2,\sigma\right)$ be parameters for which the $R$-OGP is satisfied for the transposed Gallager ensemble of \textsc{MAX-$k$-XOR-SAT}. For sufficiently large $n$, $\operatorname{DQI}_\ell$ does not succeed in sampling a bit string $\bm{z}\in\mathbb{F}_2^n$ achieving a satisfied fraction:
    \begin{equation}
        \frac{g\left(\bm{z}\right)}{m}>\mu_{\mathrm{OGP}}
    \end{equation}
    with constant probability over both the randomness of the algorithm and $\mathbb{P}_{\mathrm{G}}$.
\end{theorem}
\begin{proof}
    By Theorem~\ref{thm:dqi_is_lipschitz}, $\operatorname{DQI}_\ell$ with an $L$-inverse Lipschitz decoder is $\left(0,L,0\right)$-stable. The first result then follows from Theorem~\ref{thm:stab_algs_fail} by taking $Q$ sufficiently large and taking $p_{\mathrm{st}}+p_{\mathrm{f}}$ to be smaller than the constant $\delta>0$ defined in the theorem statement.
\end{proof}

\subsection*{Acknowledgments}

The authors thank Madhu Sudan for an insightful discussion and Kunal Marwaha for comments on a draft of this work. 
E.R.A. is funded in part by the Walter Burke Institute for Theoretical Physics at Caltech. 
J.Z.L. is funded by a National Defense Science and Engineering Graduate (NDSEG) fellowship and by the U.S. Department of Energy, Office of Science, National Quantum Information Science Research Centers, Co-design Center for Quantum Advantage (C2QA) under contract number DE-SC0012704.

\bibliography{main}

\begin{thebibliography}{54}%
\makeatletter
\providecommand \@ifxundefined [1]{%
 \@ifx{#1\undefined}
}%
\providecommand \@ifnum [1]{%
 \ifnum #1\expandafter \@firstoftwo
 \else \expandafter \@secondoftwo
 \fi
}%
\providecommand \@ifx [1]{%
 \ifx #1\expandafter \@firstoftwo
 \else \expandafter \@secondoftwo
 \fi
}%
\providecommand \natexlab [1]{#1}%
\providecommand \enquote  [1]{``#1''}%
\providecommand \bibnamefont  [1]{#1}%
\providecommand \bibfnamefont [1]{#1}%
\providecommand \citenamefont [1]{#1}%
\providecommand \href@noop [0]{\@secondoftwo}%
\providecommand \href [0]{\begingroup \@sanitize@url \@href}%
\providecommand \@href[1]{\@@startlink{#1}\@@href}%
\providecommand \@@href[1]{\endgroup#1\@@endlink}%
\providecommand \@sanitize@url [0]{\catcode `\\12\catcode `\$12\catcode `\&12\catcode `\#12\catcode `\^12\catcode `\_12\catcode `\%12\relax}%
\providecommand \@@startlink[1]{}%
\providecommand \@@endlink[0]{}%
\providecommand \url  [0]{\begingroup\@sanitize@url \@url }%
\providecommand \@url [1]{\endgroup\@href {#1}{\urlprefix }}%
\providecommand \urlprefix  [0]{URL }%
\providecommand \Eprint [0]{\href }%
\providecommand \doibase [0]{https://doi.org/}%
\providecommand \selectlanguage [0]{\@gobble}%
\providecommand \bibinfo  [0]{\@secondoftwo}%
\providecommand \bibfield  [0]{\@secondoftwo}%
\providecommand \translation [1]{[#1]}%
\providecommand \BibitemOpen [0]{}%
\providecommand \bibitemStop [0]{}%
\providecommand \bibitemNoStop [0]{.\EOS\space}%
\providecommand \EOS [0]{\spacefactor3000\relax}%
\providecommand \BibitemShut  [1]{\csname bibitem#1\endcsname}%
\let\auto@bib@innerbib\@empty
\bibitem [{\citenamefont {Gamarnik}\ \emph {et~al.}(2025)\citenamefont {Gamarnik}, \citenamefont {Jagannath},\ and\ \citenamefont {K{\i}z{\i}lda{\u{g}}}}]{gamarnik2023shatteringisingpurepspin}%
  \BibitemOpen
  \bibfield  {author} {\bibinfo {author} {\bibfnamefont {D.}~\bibnamefont {Gamarnik}}, \bibinfo {author} {\bibfnamefont {A.}~\bibnamefont {Jagannath}},\ and\ \bibinfo {author} {\bibfnamefont {E.~C.}\ \bibnamefont {K{\i}z{\i}lda{\u{g}}}},\ }\bibfield  {title} {\bibinfo {title} {Shattering in the {Ising} $p$-spin glass model},\ }\href@noop {} {\bibfield  {journal} {\bibinfo  {journal} {Probab. Theory Relat. Fields}\ }\textbf {\bibinfo {volume} {193}},\ \bibinfo {pages} {89} (\bibinfo {year} {2025})}\BibitemShut {NoStop}%
\bibitem [{\citenamefont {Jones}\ \emph {et~al.}(2023)\citenamefont {Jones}, \citenamefont {Marwaha}, \citenamefont {Sandhu},\ and\ \citenamefont {Shi}}]{jones_et_al}%
  \BibitemOpen
  \bibfield  {author} {\bibinfo {author} {\bibfnamefont {C.}~\bibnamefont {Jones}}, \bibinfo {author} {\bibfnamefont {K.}~\bibnamefont {Marwaha}}, \bibinfo {author} {\bibfnamefont {J.~S.}\ \bibnamefont {Sandhu}},\ and\ \bibinfo {author} {\bibfnamefont {J.}~\bibnamefont {Shi}},\ }\bibfield  {title} {\bibinfo {title} {Random {Max-CSPs} inherit algorithmic hardness from spin glasses},\ }in\ \href {https://doi.org/10.4230/LIPIcs.ITCS.2023.77} {\emph {\bibinfo {booktitle} {14th Innovations in Theoretical Computer Science Conference (ITCS 2023)}}},\ \bibinfo {series} {Leibniz International Proceedings in Informatics (LIPIcs)}, Vol.\ \bibinfo {volume} {251},\ \bibinfo {editor} {edited by\ \bibinfo {editor} {\bibfnamefont {Y.}~\bibnamefont {Tauman~Kalai}}}\ (\bibinfo  {publisher} {Schloss Dagstuhl -- Leibniz-Zentrum f{\"u}r Informatik},\ \bibinfo {address} {Dagstuhl, Germany},\ \bibinfo {year} {2023})\ pp.\ \bibinfo {pages} {77:1--77:26}\BibitemShut {NoStop}%
\bibitem [{\citenamefont {Shor}(1999)}]{shor1999polynomial}%
  \BibitemOpen
  \bibfield  {author} {\bibinfo {author} {\bibfnamefont {P.~W.}\ \bibnamefont {Shor}},\ }\bibfield  {title} {\bibinfo {title} {Polynomial-time algorithms for prime factorization and discrete logarithms on a quantum computer},\ }\href@noop {} {\bibfield  {journal} {\bibinfo  {journal} {SIAM review}\ }\textbf {\bibinfo {volume} {41}},\ \bibinfo {pages} {303} (\bibinfo {year} {1999})}\BibitemShut {NoStop}%
\bibitem [{\citenamefont {Gyurik}\ \emph {et~al.}(2024)\citenamefont {Gyurik}, \citenamefont {Schmidhuber}, \citenamefont {King}, \citenamefont {Dunjko},\ and\ \citenamefont {Hayakawa}}]{gyurik2024quantum}%
  \BibitemOpen
  \bibfield  {author} {\bibinfo {author} {\bibfnamefont {C.}~\bibnamefont {Gyurik}}, \bibinfo {author} {\bibfnamefont {A.}~\bibnamefont {Schmidhuber}}, \bibinfo {author} {\bibfnamefont {R.}~\bibnamefont {King}}, \bibinfo {author} {\bibfnamefont {V.}~\bibnamefont {Dunjko}},\ and\ \bibinfo {author} {\bibfnamefont {R.}~\bibnamefont {Hayakawa}},\ }\bibfield  {title} {\bibinfo {title} {Quantum computing and persistence in topological data analysis},\ }\href@noop {} {\bibfield  {journal} {\bibinfo  {journal} {arXiv preprint arXiv:2410.21258}\ } (\bibinfo {year} {2024})}\BibitemShut {NoStop}%
\bibitem [{\citenamefont {Berry}\ \emph {et~al.}(2024)\citenamefont {Berry}, \citenamefont {Su}, \citenamefont {Gyurik}, \citenamefont {King}, \citenamefont {Basso}, \citenamefont {Barba}, \citenamefont {Rajput}, \citenamefont {Wiebe}, \citenamefont {Dunjko},\ and\ \citenamefont {Babbush}}]{berry2024analyzing}%
  \BibitemOpen
  \bibfield  {author} {\bibinfo {author} {\bibfnamefont {D.~W.}\ \bibnamefont {Berry}}, \bibinfo {author} {\bibfnamefont {Y.}~\bibnamefont {Su}}, \bibinfo {author} {\bibfnamefont {C.}~\bibnamefont {Gyurik}}, \bibinfo {author} {\bibfnamefont {R.}~\bibnamefont {King}}, \bibinfo {author} {\bibfnamefont {J.}~\bibnamefont {Basso}}, \bibinfo {author} {\bibfnamefont {A.~D.~T.}\ \bibnamefont {Barba}}, \bibinfo {author} {\bibfnamefont {A.}~\bibnamefont {Rajput}}, \bibinfo {author} {\bibfnamefont {N.}~\bibnamefont {Wiebe}}, \bibinfo {author} {\bibfnamefont {V.}~\bibnamefont {Dunjko}},\ and\ \bibinfo {author} {\bibfnamefont {R.}~\bibnamefont {Babbush}},\ }\bibfield  {title} {\bibinfo {title} {Analyzing prospects for quantum advantage in topological data analysis},\ }\href@noop {} {\bibfield  {journal} {\bibinfo  {journal} {PRX Quantum}\ }\textbf {\bibinfo {volume} {5}},\ \bibinfo {pages} {010319} (\bibinfo {year} {2024})}\BibitemShut {NoStop}%
\bibitem [{\citenamefont {Harrow}\ \emph {et~al.}(2009)\citenamefont {Harrow}, \citenamefont {Hassidim},\ and\ \citenamefont {Lloyd}}]{harrow2009quantum}%
  \BibitemOpen
  \bibfield  {author} {\bibinfo {author} {\bibfnamefont {A.~W.}\ \bibnamefont {Harrow}}, \bibinfo {author} {\bibfnamefont {A.}~\bibnamefont {Hassidim}},\ and\ \bibinfo {author} {\bibfnamefont {S.}~\bibnamefont {Lloyd}},\ }\bibfield  {title} {\bibinfo {title} {Quantum algorithm for linear systems of equations},\ }\href@noop {} {\bibfield  {journal} {\bibinfo  {journal} {Physical review letters}\ }\textbf {\bibinfo {volume} {103}},\ \bibinfo {pages} {150502} (\bibinfo {year} {2009})}\BibitemShut {NoStop}%
\bibitem [{\citenamefont {Gamarnik}(2021)}]{doi:10.1073/pnas.2108492118}%
  \BibitemOpen
  \bibfield  {author} {\bibinfo {author} {\bibfnamefont {D.}~\bibnamefont {Gamarnik}},\ }\bibfield  {title} {\bibinfo {title} {The overlap gap property: A topological barrier to optimizing over random structures},\ }\href {https://doi.org/10.1073/pnas.2108492118} {\bibfield  {journal} {\bibinfo  {journal} {Proc. Natl. Acad. Sci. U.S.A.}\ }\textbf {\bibinfo {volume} {118}},\ \bibinfo {pages} {e2108492118} (\bibinfo {year} {2021})}\BibitemShut {NoStop}%
\bibitem [{\citenamefont {Chen}\ \emph {et~al.}(2019)\citenamefont {Chen}, \citenamefont {Gamarnik}, \citenamefont {Panchenko},\ and\ \citenamefont {Rahman}}]{chen2019}%
  \BibitemOpen
  \bibfield  {author} {\bibinfo {author} {\bibfnamefont {W.-K.}\ \bibnamefont {Chen}}, \bibinfo {author} {\bibfnamefont {D.}~\bibnamefont {Gamarnik}}, \bibinfo {author} {\bibfnamefont {D.}~\bibnamefont {Panchenko}},\ and\ \bibinfo {author} {\bibfnamefont {M.}~\bibnamefont {Rahman}},\ }\bibfield  {title} {\bibinfo {title} {Suboptimality of local algorithms for a class of max-cut problems},\ }\href {https://doi.org/10.1214/18-AOP1291} {\bibfield  {journal} {\bibinfo  {journal} {Ann. Probab.}\ }\textbf {\bibinfo {volume} {47}},\ \bibinfo {pages} {1587} (\bibinfo {year} {2019})}\BibitemShut {NoStop}%
\bibitem [{\citenamefont {Gamarnik}\ \emph {et~al.}(2022)\citenamefont {Gamarnik}, \citenamefont {Kızıldağ}, \citenamefont {Perkins},\ and\ \citenamefont {Xu}}]{gamarnik2022algorithmsbarrierssymmetricbinary}%
  \BibitemOpen
  \bibfield  {author} {\bibinfo {author} {\bibfnamefont {D.}~\bibnamefont {Gamarnik}}, \bibinfo {author} {\bibfnamefont {E.~C.}\ \bibnamefont {Kızıldağ}}, \bibinfo {author} {\bibfnamefont {W.}~\bibnamefont {Perkins}},\ and\ \bibinfo {author} {\bibfnamefont {C.}~\bibnamefont {Xu}},\ }\href@noop {} {\bibinfo {title} {Algorithms and barriers in the symmetric binary perceptron model}} (\bibinfo {year} {2022}),\ \Eprint {https://arxiv.org/abs/2203.15667} {arXiv:2203.15667 [cs.CC]} \BibitemShut {NoStop}%
\bibitem [{\citenamefont {Gamarnik}\ and\ \citenamefont {Kızıldağ}(2023)}]{10.1214/23-AAP1953}%
  \BibitemOpen
  \bibfield  {author} {\bibinfo {author} {\bibfnamefont {D.}~\bibnamefont {Gamarnik}}\ and\ \bibinfo {author} {\bibfnamefont {E.~C.}\ \bibnamefont {Kızıldağ}},\ }\bibfield  {title} {\bibinfo {title} {Algorithmic obstructions in the random number partitioning problem},\ }\href {https://doi.org/10.1214/23-AAP1953} {\bibfield  {journal} {\bibinfo  {journal} {Ann. Appl. Probab.}\ }\textbf {\bibinfo {volume} {33}},\ \bibinfo {pages} {5497} (\bibinfo {year} {2023})}\BibitemShut {NoStop}%
\bibitem [{\citenamefont {Huang}\ and\ \citenamefont {Sellke}(2025)}]{10.1002/cpa.22222}%
  \BibitemOpen
  \bibfield  {author} {\bibinfo {author} {\bibfnamefont {B.}~\bibnamefont {Huang}}\ and\ \bibinfo {author} {\bibfnamefont {M.}~\bibnamefont {Sellke}},\ }\bibfield  {title} {\bibinfo {title} {Tight {Lipschitz} hardness for optimizing mean field spin glasses},\ }\href {https://doi.org/10.1002/cpa.22222} {\bibfield  {journal} {\bibinfo  {journal} {Commun. Pure Appl. Math.}\ }\textbf {\bibinfo {volume} {78}},\ \bibinfo {pages} {60} (\bibinfo {year} {2025})}\BibitemShut {NoStop}%
\bibitem [{\citenamefont {Cheairi}\ and\ \citenamefont {Gamarnik}(2024)}]{cheairi2024algorithmicuniversalitylowdegreepolynomials}%
  \BibitemOpen
  \bibfield  {author} {\bibinfo {author} {\bibfnamefont {H.~E.}\ \bibnamefont {Cheairi}}\ and\ \bibinfo {author} {\bibfnamefont {D.}~\bibnamefont {Gamarnik}},\ }\href@noop {} {\bibinfo {title} {Algorithmic universality, low-degree polynomials, and {Max-Cut} in sparse random graphs}} (\bibinfo {year} {2024}),\ \Eprint {https://arxiv.org/abs/2412.18014} {arXiv:2412.18014 [cs.DS]} \BibitemShut {NoStop}%
\bibitem [{\citenamefont {Farhi}\ \emph {et~al.}(2014)\citenamefont {Farhi}, \citenamefont {Goldstone},\ and\ \citenamefont {Gutmann}}]{farhi2014quantumapproximateoptimizationalgorithm}%
  \BibitemOpen
  \bibfield  {author} {\bibinfo {author} {\bibfnamefont {E.}~\bibnamefont {Farhi}}, \bibinfo {author} {\bibfnamefont {J.}~\bibnamefont {Goldstone}},\ and\ \bibinfo {author} {\bibfnamefont {S.}~\bibnamefont {Gutmann}},\ }\href {https://arxiv.org/abs/1411.4028} {\bibinfo {title} {A quantum approximate optimization algorithm}} (\bibinfo {year} {2014}),\ \Eprint {https://arxiv.org/abs/1411.4028} {arXiv:1411.4028 [quant-ph]} \BibitemShut {NoStop}%
\bibitem [{\citenamefont {Cerezo}\ \emph {et~al.}(2021)\citenamefont {Cerezo}, \citenamefont {Arrasmith}, \citenamefont {Babbush}, \citenamefont {Benjamin}, \citenamefont {Endo}, \citenamefont {Fujii}, \citenamefont {McClean}, \citenamefont {Mitarai}, \citenamefont {Yuan}, \citenamefont {Cincio} \emph {et~al.}}]{cerezo2021variational}%
  \BibitemOpen
  \bibfield  {author} {\bibinfo {author} {\bibfnamefont {M.}~\bibnamefont {Cerezo}}, \bibinfo {author} {\bibfnamefont {A.}~\bibnamefont {Arrasmith}}, \bibinfo {author} {\bibfnamefont {R.}~\bibnamefont {Babbush}}, \bibinfo {author} {\bibfnamefont {S.~C.}\ \bibnamefont {Benjamin}}, \bibinfo {author} {\bibfnamefont {S.}~\bibnamefont {Endo}}, \bibinfo {author} {\bibfnamefont {K.}~\bibnamefont {Fujii}}, \bibinfo {author} {\bibfnamefont {J.~R.}\ \bibnamefont {McClean}}, \bibinfo {author} {\bibfnamefont {K.}~\bibnamefont {Mitarai}}, \bibinfo {author} {\bibfnamefont {X.}~\bibnamefont {Yuan}}, \bibinfo {author} {\bibfnamefont {L.}~\bibnamefont {Cincio}}, \emph {et~al.},\ }\bibfield  {title} {\bibinfo {title} {Variational quantum algorithms},\ }\href@noop {} {\bibfield  {journal} {\bibinfo  {journal} {Nature Reviews Physics}\ }\textbf {\bibinfo {volume} {3}},\ \bibinfo {pages} {625} (\bibinfo {year} {2021})}\BibitemShut {NoStop}%
\bibitem [{\citenamefont {Farhi}\ \emph {et~al.}(2020)\citenamefont {Farhi}, \citenamefont {Gamarnik},\ and\ \citenamefont {Gutmann}}]{farhi2020quantumapproximateoptimizationalgorithm}%
  \BibitemOpen
  \bibfield  {author} {\bibinfo {author} {\bibfnamefont {E.}~\bibnamefont {Farhi}}, \bibinfo {author} {\bibfnamefont {D.}~\bibnamefont {Gamarnik}},\ and\ \bibinfo {author} {\bibfnamefont {S.}~\bibnamefont {Gutmann}},\ }\href@noop {} {\bibinfo {title} {The quantum approximate optimization algorithm needs to see the whole graph: A typical case}} (\bibinfo {year} {2020}),\ \Eprint {https://arxiv.org/abs/2004.09002} {arXiv:2004.09002 [quant-ph]} \BibitemShut {NoStop}%
\bibitem [{\citenamefont {Anschuetz}(2022)}]{anschuetz2022critical}%
  \BibitemOpen
  \bibfield  {author} {\bibinfo {author} {\bibfnamefont {E.~R.}\ \bibnamefont {Anschuetz}},\ }\bibfield  {title} {\bibinfo {title} {Critical points in quantum generative models},\ }in\ \href {https://openreview.net/forum?id=2f1z55GVQN} {\emph {\bibinfo {booktitle} {{I}nternational {C}onference on {L}earning {R}epresentations}}},\ \bibinfo {editor} {edited by\ \bibinfo {editor} {\bibfnamefont {K.}~\bibnamefont {Hofmann}}, \bibinfo {editor} {\bibfnamefont {A.}~\bibnamefont {Rush}}, \bibinfo {editor} {\bibfnamefont {Y.}~\bibnamefont {Liu}}, \bibinfo {editor} {\bibfnamefont {C.}~\bibnamefont {Finn}}, \bibinfo {editor} {\bibfnamefont {Y.}~\bibnamefont {Choi}},\ and\ \bibinfo {editor} {\bibfnamefont {M.}~\bibnamefont {Deisenroth}}}\ (\bibinfo  {publisher} {OpenReview},\ \bibinfo {year} {2022})\BibitemShut {NoStop}%
\bibitem [{\citenamefont {Anschuetz}\ and\ \citenamefont {Kiani}(2022)}]{anschuetzkiani2022}%
  \BibitemOpen
  \bibfield  {author} {\bibinfo {author} {\bibfnamefont {E.~R.}\ \bibnamefont {Anschuetz}}\ and\ \bibinfo {author} {\bibfnamefont {B.~T.}\ \bibnamefont {Kiani}},\ }\bibfield  {title} {\bibinfo {title} {Quantum variational algorithms are swamped with traps},\ }\href {https://doi.org/10.1038/s41467-022-35364-5} {\bibfield  {journal} {\bibinfo  {journal} {Nat. Commun.}\ }\textbf {\bibinfo {volume} {13}},\ \bibinfo {pages} {7760} (\bibinfo {year} {2022})}\BibitemShut {NoStop}%
\bibitem [{\citenamefont {Basso}\ \emph {et~al.}(2022)\citenamefont {Basso}, \citenamefont {Gamarnik}, \citenamefont {Mei},\ and\ \citenamefont {Zhou}}]{9996946}%
  \BibitemOpen
  \bibfield  {author} {\bibinfo {author} {\bibfnamefont {J.}~\bibnamefont {Basso}}, \bibinfo {author} {\bibfnamefont {D.}~\bibnamefont {Gamarnik}}, \bibinfo {author} {\bibfnamefont {S.}~\bibnamefont {Mei}},\ and\ \bibinfo {author} {\bibfnamefont {L.}~\bibnamefont {Zhou}},\ }\bibfield  {title} {\bibinfo {title} {Performance and limitations of the {QAOA} at constant levels on large sparse hypergraphs and spin glass models},\ }in\ \href {https://doi.org/10.1109/FOCS54457.2022.00039} {\emph {\bibinfo {booktitle} {2022 IEEE 63rd Annual Symposium on Foundations of Computer Science (FOCS)}}}\ (\bibinfo {year} {2022})\ pp.\ \bibinfo {pages} {335--343}\BibitemShut {NoStop}%
\bibitem [{\citenamefont {Anshu}\ and\ \citenamefont {Metger}(2023)}]{Anshu2023concentrationbounds}%
  \BibitemOpen
  \bibfield  {author} {\bibinfo {author} {\bibfnamefont {A.}~\bibnamefont {Anshu}}\ and\ \bibinfo {author} {\bibfnamefont {T.}~\bibnamefont {Metger}},\ }\bibfield  {title} {\bibinfo {title} {Concentration bounds for quantum states and limitations on the {QAOA} from polynomial approximations},\ }\href {https://doi.org/10.22331/q-2023-05-11-999} {\bibfield  {journal} {\bibinfo  {journal} {{Quantum}}\ }\textbf {\bibinfo {volume} {7}},\ \bibinfo {pages} {999} (\bibinfo {year} {2023})}\BibitemShut {NoStop}%
\bibitem [{\citenamefont {Chen}\ \emph {et~al.}(2023)\citenamefont {Chen}, \citenamefont {Huang},\ and\ \citenamefont {Marwaha}}]{chen2023localalgorithmsfailurelogdepth}%
  \BibitemOpen
  \bibfield  {author} {\bibinfo {author} {\bibfnamefont {A.}~\bibnamefont {Chen}}, \bibinfo {author} {\bibfnamefont {N.}~\bibnamefont {Huang}},\ and\ \bibinfo {author} {\bibfnamefont {K.}~\bibnamefont {Marwaha}},\ }\href@noop {} {\bibinfo {title} {Local algorithms and the failure of log-depth quantum advantage on sparse random {CSPs}}} (\bibinfo {year} {2023}),\ \Eprint {https://arxiv.org/abs/2310.01563} {arXiv:2310.01563 [quant-ph]} \BibitemShut {NoStop}%
\bibitem [{\citenamefont {Goh}(2025)}]{goh2025overlapgappropertylimits}%
  \BibitemOpen
  \bibfield  {author} {\bibinfo {author} {\bibfnamefont {M.}~\bibnamefont {Goh}},\ }\href {https://arxiv.org/abs/2404.06087} {\bibinfo {title} {The overlap gap property limits limit swapping in the qaoa}} (\bibinfo {year} {2025}),\ \Eprint {https://arxiv.org/abs/2404.06087} {arXiv:2404.06087 [quant-ph]} \BibitemShut {NoStop}%
\bibitem [{\citenamefont {Anschuetz}(2025)}]{anschuetz2025unified}%
  \BibitemOpen
  \bibfield  {author} {\bibinfo {author} {\bibfnamefont {E.~R.}\ \bibnamefont {Anschuetz}},\ }\bibfield  {title} {\bibinfo {title} {A unified theory of quantum neural network loss landscapes},\ }in\ \href {https://openreview.net/forum?id=fv8TTt9srF} {\emph {\bibinfo {booktitle} {{I}nternational {C}onference on {L}earning {R}epresentations}}},\ \bibinfo {editor} {edited by\ \bibinfo {editor} {\bibfnamefont {Y.}~\bibnamefont {Yue}}, \bibinfo {editor} {\bibfnamefont {A.}~\bibnamefont {Garg}}, \bibinfo {editor} {\bibfnamefont {N.}~\bibnamefont {Peng}}, \bibinfo {editor} {\bibfnamefont {F.}~\bibnamefont {Sha}},\ and\ \bibinfo {editor} {\bibfnamefont {R.}~\bibnamefont {Yu}}}\ (\bibinfo  {publisher} {OpenReview},\ \bibinfo {year} {2025})\BibitemShut {NoStop}%
\bibitem [{\citenamefont {Anschuetz}(2026)}]{anschuetz2025efficientlearningimpliesquantum}%
  \BibitemOpen
  \bibfield  {author} {\bibinfo {author} {\bibfnamefont {E.~R.}\ \bibnamefont {Anschuetz}},\ }\bibfield  {title} {\bibinfo {title} {Quantum glassiness from efficient learning},\ }\href {https://doi.org/10.1007/s00220-025-05536-7} {\bibfield  {journal} {\bibinfo  {journal} {Commun. Math. Phys.}\ }\textbf {\bibinfo {volume} {407}},\ \bibinfo {pages} {30} (\bibinfo {year} {2026})}\BibitemShut {NoStop}%
\bibitem [{\citenamefont {Jordan}\ \emph {et~al.}(2025)\citenamefont {Jordan}, \citenamefont {Shutty}, \citenamefont {Wootters}, \citenamefont {Zalcman}, \citenamefont {Schmidhuber}, \citenamefont {King}, \citenamefont {Isakov}, \citenamefont {Khattar},\ and\ \citenamefont {Babbush}}]{jordan2025optimizationdecodedquantuminterferometry}%
  \BibitemOpen
  \bibfield  {author} {\bibinfo {author} {\bibfnamefont {S.~P.}\ \bibnamefont {Jordan}}, \bibinfo {author} {\bibfnamefont {N.}~\bibnamefont {Shutty}}, \bibinfo {author} {\bibfnamefont {M.}~\bibnamefont {Wootters}}, \bibinfo {author} {\bibfnamefont {A.}~\bibnamefont {Zalcman}}, \bibinfo {author} {\bibfnamefont {A.}~\bibnamefont {Schmidhuber}}, \bibinfo {author} {\bibfnamefont {R.}~\bibnamefont {King}}, \bibinfo {author} {\bibfnamefont {S.~V.}\ \bibnamefont {Isakov}}, \bibinfo {author} {\bibfnamefont {T.}~\bibnamefont {Khattar}},\ and\ \bibinfo {author} {\bibfnamefont {R.}~\bibnamefont {Babbush}},\ }\href@noop {} {\bibinfo {title} {Optimization by {Decoded Quantum Interferometry}}} (\bibinfo {year} {2025}),\ \Eprint {https://arxiv.org/abs/2408.08292} {arXiv:2408.08292 [quant-ph]} \BibitemShut {NoStop}%
\bibitem [{\citenamefont {Reed}\ and\ \citenamefont {Solomon}(1960)}]{reed1960polynomial}%
  \BibitemOpen
  \bibfield  {author} {\bibinfo {author} {\bibfnamefont {I.~S.}\ \bibnamefont {Reed}}\ and\ \bibinfo {author} {\bibfnamefont {G.}~\bibnamefont {Solomon}},\ }\bibfield  {title} {\bibinfo {title} {Polynomial codes over certain finite fields},\ }\href@noop {} {\bibfield  {journal} {\bibinfo  {journal} {Journal of the society for industrial and applied mathematics}\ }\textbf {\bibinfo {volume} {8}},\ \bibinfo {pages} {300} (\bibinfo {year} {1960})}\BibitemShut {NoStop}%
\bibitem [{\citenamefont {De~Palma}\ \emph {et~al.}(2021)\citenamefont {De~Palma}, \citenamefont {Marvian}, \citenamefont {Trevisan},\ and\ \citenamefont {Lloyd}}]{9420734}%
  \BibitemOpen
  \bibfield  {author} {\bibinfo {author} {\bibfnamefont {G.}~\bibnamefont {De~Palma}}, \bibinfo {author} {\bibfnamefont {M.}~\bibnamefont {Marvian}}, \bibinfo {author} {\bibfnamefont {D.}~\bibnamefont {Trevisan}},\ and\ \bibinfo {author} {\bibfnamefont {S.}~\bibnamefont {Lloyd}},\ }\bibfield  {title} {\bibinfo {title} {The quantum {Wasserstein} distance of order 1},\ }\href {https://doi.org/10.1109/TIT.2021.3076442} {\bibfield  {journal} {\bibinfo  {journal} {{IEEE} Trans. Inf. Theory}\ }\textbf {\bibinfo {volume} {67}},\ \bibinfo {pages} {6627} (\bibinfo {year} {2021})}\BibitemShut {NoStop}%
\bibitem [{\citenamefont {Zyablov}\ and\ \citenamefont {Pinsker}(1975)}]{zyablov1975estimation}%
  \BibitemOpen
  \bibfield  {author} {\bibinfo {author} {\bibfnamefont {V.~V.}\ \bibnamefont {Zyablov}}\ and\ \bibinfo {author} {\bibfnamefont {M.~S.}\ \bibnamefont {Pinsker}},\ }\bibfield  {title} {\bibinfo {title} {Estimation of the error-correction complexity for {Gallager} low-density codes},\ }\href {https://www.mathnet.ru/eng/ppi1568} {\bibfield  {journal} {\bibinfo  {journal} {Probl. Peredachi Inf.}\ }\textbf {\bibinfo {volume} {11}},\ \bibinfo {pages} {23} (\bibinfo {year} {1975})}\BibitemShut {NoStop}%
\bibitem [{\citenamefont {Richardson}\ and\ \citenamefont {Urbanke}(2008)}]{richardson2008modern}%
  \BibitemOpen
  \bibfield  {author} {\bibinfo {author} {\bibfnamefont {T.}~\bibnamefont {Richardson}}\ and\ \bibinfo {author} {\bibfnamefont {R.}~\bibnamefont {Urbanke}},\ }\href@noop {} {\emph {\bibinfo {title} {Modern coding theory}}}\ (\bibinfo  {publisher} {Cambridge university press},\ \bibinfo {year} {2008})\BibitemShut {NoStop}%
\bibitem [{\citenamefont {Feldman}\ \emph {et~al.}(2003)\citenamefont {Feldman}, \citenamefont {Wainwright},\ and\ \citenamefont {Karger}}]{feldman2003using}%
  \BibitemOpen
  \bibfield  {author} {\bibinfo {author} {\bibfnamefont {J.}~\bibnamefont {Feldman}}, \bibinfo {author} {\bibfnamefont {M.}~\bibnamefont {Wainwright}},\ and\ \bibinfo {author} {\bibfnamefont {D.~R.}\ \bibnamefont {Karger}},\ }\bibfield  {title} {\bibinfo {title} {Using linear programming to decode linear codes},\ }in\ \href@noop {} {\emph {\bibinfo {booktitle} {37th annual Conference on Information Sciences and Systems (CISS’03)}}}\ (\bibinfo {year} {2003})\BibitemShut {NoStop}%
\bibitem [{\citenamefont {Koetter}\ \emph {et~al.}(2005)\citenamefont {Koetter}, \citenamefont {LI}, \citenamefont {Vontobel},\ and\ \citenamefont {Walker}}]{koetter2005characterizations}%
  \BibitemOpen
  \bibfield  {author} {\bibinfo {author} {\bibfnamefont {R.}~\bibnamefont {Koetter}}, \bibinfo {author} {\bibfnamefont {W.}~\bibnamefont {LI}}, \bibinfo {author} {\bibfnamefont {P.}~\bibnamefont {Vontobel}},\ and\ \bibinfo {author} {\bibfnamefont {J.}~\bibnamefont {Walker}},\ }\bibfield  {title} {\bibinfo {title} {Characterizations of pseudo-codewords of ldpc codes, arxiv report},\ }\href@noop {} {\bibfield  {journal} {\bibinfo  {journal} {arXiv preprint cs.IT/0508049}\ } (\bibinfo {year} {2005})}\BibitemShut {NoStop}%
\bibitem [{\citenamefont {Ghazi}\ and\ \citenamefont {Lee}(2017)}]{ghazi2017lp}%
  \BibitemOpen
  \bibfield  {author} {\bibinfo {author} {\bibfnamefont {B.}~\bibnamefont {Ghazi}}\ and\ \bibinfo {author} {\bibfnamefont {E.}~\bibnamefont {Lee}},\ }\bibfield  {title} {\bibinfo {title} {Lp/sdp hierarchy lower bounds for decoding random ldpc codes},\ }\href@noop {} {\bibfield  {journal} {\bibinfo  {journal} {IEEE Transactions on Information Theory}\ }\textbf {\bibinfo {volume} {64}},\ \bibinfo {pages} {4423} (\bibinfo {year} {2017})}\BibitemShut {NoStop}%
\bibitem [{\citenamefont {El~Alaoui}\ \emph {et~al.}(2021)\citenamefont {El~Alaoui}, \citenamefont {Montanari},\ and\ \citenamefont {Sellke}}]{el2021optimization}%
  \BibitemOpen
  \bibfield  {author} {\bibinfo {author} {\bibfnamefont {A.}~\bibnamefont {El~Alaoui}}, \bibinfo {author} {\bibfnamefont {A.}~\bibnamefont {Montanari}},\ and\ \bibinfo {author} {\bibfnamefont {M.}~\bibnamefont {Sellke}},\ }\bibfield  {title} {\bibinfo {title} {Optimization of mean-field spin glasses},\ }\href {https://doi.org/10.1214/21-AOP1519} {\bibfield  {journal} {\bibinfo  {journal} {Ann. Probab.}\ }\textbf {\bibinfo {volume} {49}},\ \bibinfo {pages} {2922} (\bibinfo {year} {2021})}\BibitemShut {NoStop}%
\bibitem [{\citenamefont {Alaoui}\ and\ \citenamefont {Montanari}(2020)}]{alaoui2020algorithmicthresholdsmeanfield}%
  \BibitemOpen
  \bibfield  {author} {\bibinfo {author} {\bibfnamefont {A.~E.}\ \bibnamefont {Alaoui}}\ and\ \bibinfo {author} {\bibfnamefont {A.}~\bibnamefont {Montanari}},\ }\href {https://arxiv.org/abs/2009.11481} {\bibinfo {title} {Algorithmic thresholds in mean field spin glasses}} (\bibinfo {year} {2020}),\ \Eprint {https://arxiv.org/abs/2009.11481} {arXiv:2009.11481 [cond-mat.stat-mech]} \BibitemShut {NoStop}%
\bibitem [{\citenamefont {Marwaha}\ and\ \citenamefont {Hadfield}(2022)}]{marwaha2022boundsapproximating}%
  \BibitemOpen
  \bibfield  {author} {\bibinfo {author} {\bibfnamefont {K.}~\bibnamefont {Marwaha}}\ and\ \bibinfo {author} {\bibfnamefont {S.}~\bibnamefont {Hadfield}},\ }\bibfield  {title} {\bibinfo {title} {Bounds on approximating {M}ax {$k$}{XOR} with quantum and classical local algorithms},\ }\href {https://doi.org/10.22331/q-2022-07-07-757} {\bibfield  {journal} {\bibinfo  {journal} {{Quantum}}\ }\textbf {\bibinfo {volume} {6}},\ \bibinfo {pages} {757} (\bibinfo {year} {2022})}\BibitemShut {NoStop}%
\bibitem [{\citenamefont {Marwaha}\ \emph {et~al.}(2025)\citenamefont {Marwaha}, \citenamefont {Fefferman}, \citenamefont {Gheorghiu},\ and\ \citenamefont {Havlicek}}]{marwaha2025complexitydecodedquantuminterferometry}%
  \BibitemOpen
  \bibfield  {author} {\bibinfo {author} {\bibfnamefont {K.}~\bibnamefont {Marwaha}}, \bibinfo {author} {\bibfnamefont {B.}~\bibnamefont {Fefferman}}, \bibinfo {author} {\bibfnamefont {A.}~\bibnamefont {Gheorghiu}},\ and\ \bibinfo {author} {\bibfnamefont {V.}~\bibnamefont {Havlicek}},\ }\href@noop {} {\bibinfo {title} {On the complexity of {Decoded Quantum Interferometry}}} (\bibinfo {year} {2025}),\ \Eprint {https://arxiv.org/abs/2509.14443} {arXiv:2509.14443 [quant-ph]} \BibitemShut {NoStop}%
\bibitem [{\citenamefont {Parekh}(2025)}]{parekh2025quantumadvantagedecodedquantum}%
  \BibitemOpen
  \bibfield  {author} {\bibinfo {author} {\bibfnamefont {O.}~\bibnamefont {Parekh}},\ }\href@noop {} {\bibinfo {title} {No quantum advantage in decoded quantum interferometry for maxcut}} (\bibinfo {year} {2025}),\ \Eprint {https://arxiv.org/abs/2509.19966} {arXiv:2509.19966 [quant-ph]} \BibitemShut {NoStop}%
\bibitem [{\citenamefont {Gallager}(1962)}]{gallager2003low}%
  \BibitemOpen
  \bibfield  {author} {\bibinfo {author} {\bibfnamefont {R.}~\bibnamefont {Gallager}},\ }\bibfield  {title} {\bibinfo {title} {Low-density parity-check codes},\ }\href {https://doi.org/10.1109/TIT.1962.1057683} {\bibfield  {journal} {\bibinfo  {journal} {IRE Trans. Inf. Theory}\ }\textbf {\bibinfo {volume} {8}},\ \bibinfo {pages} {21} (\bibinfo {year} {1962})}\BibitemShut {NoStop}%
\bibitem [{\citenamefont {Mosheiff}\ \emph {et~al.}(2021)\citenamefont {Mosheiff}, \citenamefont {Resch}, \citenamefont {Ron-Zewi}, \citenamefont {Silas},\ and\ \citenamefont {Wootters}}]{mosheiff2021low}%
  \BibitemOpen
  \bibfield  {author} {\bibinfo {author} {\bibfnamefont {J.}~\bibnamefont {Mosheiff}}, \bibinfo {author} {\bibfnamefont {N.}~\bibnamefont {Resch}}, \bibinfo {author} {\bibfnamefont {N.}~\bibnamefont {Ron-Zewi}}, \bibinfo {author} {\bibfnamefont {S.}~\bibnamefont {Silas}},\ and\ \bibinfo {author} {\bibfnamefont {M.}~\bibnamefont {Wootters}},\ }\bibfield  {title} {\bibinfo {title} {Low-density parity-check codes achieve list-decoding capacity},\ }\href@noop {} {\bibfield  {journal} {\bibinfo  {journal} {SIAM Journal on Computing}\ }\textbf {\bibinfo {volume} {53}},\ \bibinfo {pages} {FOCS20} (\bibinfo {year} {2021})}\BibitemShut {NoStop}%
\bibitem [{\citenamefont {Jones}\ \emph {et~al.}(2022)\citenamefont {Jones}, \citenamefont {Marwaha}, \citenamefont {Sandhu},\ and\ \citenamefont {Shi}}]{jones_et_al_full}%
  \BibitemOpen
  \bibfield  {author} {\bibinfo {author} {\bibfnamefont {C.}~\bibnamefont {Jones}}, \bibinfo {author} {\bibfnamefont {K.}~\bibnamefont {Marwaha}}, \bibinfo {author} {\bibfnamefont {J.~S.}\ \bibnamefont {Sandhu}},\ and\ \bibinfo {author} {\bibfnamefont {J.}~\bibnamefont {Shi}},\ }\href@noop {} {\bibinfo {title} {Random {Max-CSPs} inherit algorithmic hardness from spin glasses}} (\bibinfo {year} {2022})\BibitemShut {NoStop}%
\bibitem [{\citenamefont {Chailloux}\ and\ \citenamefont {Tillich}(2023)}]{chailloux2023quantum}%
  \BibitemOpen
  \bibfield  {author} {\bibinfo {author} {\bibfnamefont {A.}~\bibnamefont {Chailloux}}\ and\ \bibinfo {author} {\bibfnamefont {J.-P.}\ \bibnamefont {Tillich}},\ }\bibfield  {title} {\bibinfo {title} {The quantum decoding problem},\ }\href@noop {} {\bibfield  {journal} {\bibinfo  {journal} {arXiv preprint arXiv:2310.20651}\ } (\bibinfo {year} {2023})}\BibitemShut {NoStop}%
\bibitem [{\citenamefont {B{\"a}rtschi}\ and\ \citenamefont {Eidenbenz}(2022)}]{bartschi2022short}%
  \BibitemOpen
  \bibfield  {author} {\bibinfo {author} {\bibfnamefont {A.}~\bibnamefont {B{\"a}rtschi}}\ and\ \bibinfo {author} {\bibfnamefont {S.}~\bibnamefont {Eidenbenz}},\ }\bibfield  {title} {\bibinfo {title} {Short-depth circuits for dicke state preparation},\ }in\ \href@noop {} {\emph {\bibinfo {booktitle} {2022 IEEE International Conference on Quantum Computing and Engineering (QCE)}}}\ (\bibinfo {organization} {IEEE},\ \bibinfo {year} {2022})\ pp.\ \bibinfo {pages} {87--96}\BibitemShut {NoStop}%
\bibitem [{\citenamefont {Anschuetz}\ \emph {et~al.}(2024)\citenamefont {Anschuetz}, \citenamefont {Gamarnik},\ and\ \citenamefont {Kiani}}]{anschuetz2024combinatorial}%
  \BibitemOpen
  \bibfield  {author} {\bibinfo {author} {\bibfnamefont {E.~R.}\ \bibnamefont {Anschuetz}}, \bibinfo {author} {\bibfnamefont {D.}~\bibnamefont {Gamarnik}},\ and\ \bibinfo {author} {\bibfnamefont {B.}~\bibnamefont {Kiani}},\ }\bibfield  {title} {\bibinfo {title} {Combinatorial {NLTS} from the overlap gap property},\ }\href@noop {} {\bibfield  {journal} {\bibinfo  {journal} {Quantum}\ }\textbf {\bibinfo {volume} {8}},\ \bibinfo {pages} {1527} (\bibinfo {year} {2024})}\BibitemShut {NoStop}%
\bibitem [{\citenamefont {Rahman}\ and\ \citenamefont {Vir{\'a}g}(2017)}]{rahman2017local}%
  \BibitemOpen
  \bibfield  {author} {\bibinfo {author} {\bibfnamefont {M.}~\bibnamefont {Rahman}}\ and\ \bibinfo {author} {\bibfnamefont {B.}~\bibnamefont {Vir{\'a}g}},\ }\bibfield  {title} {\bibinfo {title} {Local algorithms for independent sets are half-optimal},\ }\href {https://doi.org/10.1214/16-AOP1094} {\bibfield  {journal} {\bibinfo  {journal} {Ann. Probab.}\ }\textbf {\bibinfo {volume} {45}},\ \bibinfo {pages} {1543} (\bibinfo {year} {2017})}\BibitemShut {NoStop}%
\bibitem [{\citenamefont {Gamarnik}\ and\ \citenamefont {Sudan}(2017)}]{doi:10.1137/140989728}%
  \BibitemOpen
  \bibfield  {author} {\bibinfo {author} {\bibfnamefont {D.}~\bibnamefont {Gamarnik}}\ and\ \bibinfo {author} {\bibfnamefont {M.}~\bibnamefont {Sudan}},\ }\bibfield  {title} {\bibinfo {title} {Performance of sequential local algorithms for the random {NAE-\$K\$-SAT} problem},\ }\href {https://doi.org/10.1137/140989728} {\bibfield  {journal} {\bibinfo  {journal} {SIAM Journal on Computing}\ }\textbf {\bibinfo {volume} {46}},\ \bibinfo {pages} {590} (\bibinfo {year} {2017})}\BibitemShut {NoStop}%
\bibitem [{\citenamefont {Burshtein}\ and\ \citenamefont {Miller}(2002)}]{burshtein2002bounds}%
  \BibitemOpen
  \bibfield  {author} {\bibinfo {author} {\bibfnamefont {D.}~\bibnamefont {Burshtein}}\ and\ \bibinfo {author} {\bibfnamefont {G.}~\bibnamefont {Miller}},\ }\bibfield  {title} {\bibinfo {title} {Bounds on the performance of belief propagation decoding},\ }\href@noop {} {\bibfield  {journal} {\bibinfo  {journal} {IEEE Transactions on Information Theory}\ }\textbf {\bibinfo {volume} {48}},\ \bibinfo {pages} {112} (\bibinfo {year} {2002})}\BibitemShut {NoStop}%
\bibitem [{\citenamefont {Richardson}\ and\ \citenamefont {Urbanke}(2001)}]{richardson2001capacity}%
  \BibitemOpen
  \bibfield  {author} {\bibinfo {author} {\bibfnamefont {T.~J.}\ \bibnamefont {Richardson}}\ and\ \bibinfo {author} {\bibfnamefont {R.~L.}\ \bibnamefont {Urbanke}},\ }\bibfield  {title} {\bibinfo {title} {The capacity of low-density parity-check codes under message-passing decoding},\ }\href@noop {} {\bibfield  {journal} {\bibinfo  {journal} {IEEE Transactions on information theory}\ }\textbf {\bibinfo {volume} {47}},\ \bibinfo {pages} {599} (\bibinfo {year} {2001})}\BibitemShut {NoStop}%
\bibitem [{\citenamefont {Papadimitriou}(1981)}]{papadimitriou1981complexity}%
  \BibitemOpen
  \bibfield  {author} {\bibinfo {author} {\bibfnamefont {C.~H.}\ \bibnamefont {Papadimitriou}},\ }\bibfield  {title} {\bibinfo {title} {On the complexity of integer programming},\ }\href@noop {} {\bibfield  {journal} {\bibinfo  {journal} {Journal of the ACM (JACM)}\ }\textbf {\bibinfo {volume} {28}},\ \bibinfo {pages} {765} (\bibinfo {year} {1981})}\BibitemShut {NoStop}%
\bibitem [{\citenamefont {Khachiyan}(1979)}]{khachiyan1979polynomial}%
  \BibitemOpen
  \bibfield  {author} {\bibinfo {author} {\bibfnamefont {L.~G.}\ \bibnamefont {Khachiyan}},\ }\bibfield  {title} {\bibinfo {title} {A polynomial algorithm in linear programming},\ }\href {http://mi.mathnet.ru/dan42319} {\bibfield  {journal} {\bibinfo  {journal} {Proc. USSR Acad. Sci.}\ }\textbf {\bibinfo {volume} {244}},\ \bibinfo {pages} {1093} (\bibinfo {year} {1979})}\BibitemShut {NoStop}%
\bibitem [{\citenamefont {Feldman}\ \emph {et~al.}(2007)\citenamefont {Feldman}, \citenamefont {Malkin}, \citenamefont {Servedio}, \citenamefont {Stein},\ and\ \citenamefont {Wainwright}}]{feldman2007lp}%
  \BibitemOpen
  \bibfield  {author} {\bibinfo {author} {\bibfnamefont {J.}~\bibnamefont {Feldman}}, \bibinfo {author} {\bibfnamefont {T.}~\bibnamefont {Malkin}}, \bibinfo {author} {\bibfnamefont {R.~A.}\ \bibnamefont {Servedio}}, \bibinfo {author} {\bibfnamefont {C.}~\bibnamefont {Stein}},\ and\ \bibinfo {author} {\bibfnamefont {M.~J.}\ \bibnamefont {Wainwright}},\ }\bibfield  {title} {\bibinfo {title} {Lp decoding corrects a constant fraction of errors},\ }\href@noop {} {\bibfield  {journal} {\bibinfo  {journal} {IEEE Transactions on Information Theory}\ }\textbf {\bibinfo {volume} {53}},\ \bibinfo {pages} {82} (\bibinfo {year} {2007})}\BibitemShut {NoStop}%
\bibitem [{\citenamefont {Vershynin}(2018)}]{vershynin_vectors}%
  \BibitemOpen
  \bibfield  {author} {\bibinfo {author} {\bibfnamefont {R.}~\bibnamefont {Vershynin}},\ }\bibinfo {title} {Concentration without independence},\ in\ \href {https://doi.org/10.1017/9781108231596.008} {\emph {\bibinfo {booktitle} {High-Dimensional Probability: An Introduction with Applications in Data Science}}},\ \bibinfo {series and number} {Cambridge Series in Statistical and Probabilistic Mathematics}\ (\bibinfo  {publisher} {Cambridge University Press},\ \bibinfo {year} {2018})\ pp.\ \bibinfo {pages} {98--126}\BibitemShut {NoStop}%
\bibitem [{\citenamefont {Frieze}(1990)}]{frieze1990171}%
  \BibitemOpen
  \bibfield  {author} {\bibinfo {author} {\bibfnamefont {A.}~\bibnamefont {Frieze}},\ }\bibfield  {title} {\bibinfo {title} {On the independence number of random graphs},\ }\href {https://doi.org/10.1016/0012-365X(90)90149-C} {\bibfield  {journal} {\bibinfo  {journal} {Discrete Math.}\ }\textbf {\bibinfo {volume} {81}},\ \bibinfo {pages} {171} (\bibinfo {year} {1990})}\BibitemShut {NoStop}%
\bibitem [{\citenamefont {Calkin}(1997)}]{calkin1997}%
  \BibitemOpen
  \bibfield  {author} {\bibinfo {author} {\bibfnamefont {N.~J.}\ \bibnamefont {Calkin}},\ }\bibfield  {title} {\bibinfo {title} {Dependent sets of constant weight binary vectors},\ }\href {https://doi.org/10.1017/S0963548397003040} {\bibfield  {journal} {\bibinfo  {journal} {Comb. Probab. Comput.}\ }\textbf {\bibinfo {volume} {6}},\ \bibinfo {pages} {263} (\bibinfo {year} {1997})}\BibitemShut {NoStop}%
\bibitem [{\citenamefont {Kirshner}\ and\ \citenamefont {Samorodnitsky}(2021)}]{9398654}%
  \BibitemOpen
  \bibfield  {author} {\bibinfo {author} {\bibfnamefont {N.}~\bibnamefont {Kirshner}}\ and\ \bibinfo {author} {\bibfnamefont {A.}~\bibnamefont {Samorodnitsky}},\ }\bibfield  {title} {\bibinfo {title} {A moment ratio bound for polynomials and some extremal properties of {Krawchouk} polynomials and {Hamming} spheres},\ }\href {https://doi.org/10.1109/TIT.2021.3071597} {\bibfield  {journal} {\bibinfo  {journal} {IEEE Trans. Inf. Theory}\ }\textbf {\bibinfo {volume} {67}},\ \bibinfo {pages} {3509} (\bibinfo {year} {2021})}\BibitemShut {NoStop}%
\bibitem [{\citenamefont {Sándor}(2017)}]{sandor2017two}%
  \BibitemOpen
  \bibfield  {author} {\bibinfo {author} {\bibfnamefont {J.}~\bibnamefont {Sándor}},\ }\bibfield  {title} {\bibinfo {title} {Two applications of the {Hadamard} integral inequality},\ }\href@noop {} {\bibfield  {journal} {\bibinfo  {journal} {Notes Number Theory Discrete Math.}\ }\textbf {\bibinfo {volume} {23}},\ \bibinfo {pages} {52} (\bibinfo {year} {2017})}\BibitemShut {NoStop}%
\end{thebibliography}%

\newpage
\onecolumngrid
\appendix

\section*{Supplementary Material}

\section{Maximum Decodable Error Weight for Gallager Ensemble Codes}\label{sec:decode_weight}

Recall that we are working with the $(k, d)$ Gallager ensemble, an ensemble of LDPC codes whose parity check matrices form $(k, d)$-regular bipartite graphs in their Tanner graph representation.
Each check acts on $d$ bits, and each bit is acted on by $k$ checks.
Equivalently, every data (left) node has degree $k$ and every check (right) node has degree $d$.

\subsection{Limitations of Belief Propagation}

We here study the asymptotic performance of belief propagation decoding on the $(k, d)$ Gallager ensemble.
We specialize to the Gallager ensemble for simplicity, but our result holds for any LDPC ensemble which is asymptotically locally tree-like, which includes not only the Gallager ensemble but also the random $(k, d)$-regular Tanner graph ensemble and irregular degree distribution ensembles~\cite{burshtein2002bounds}.

Fix a design rate $r$ and let the number of data bits be $n$.
Belief propagation (BP) executes in a round-by-round fashion, the total number of which $r$ are taken to be fixed ahead of time.
The fixing of rounds precludes the technical possibility of letting $r$ scale with $n$, but is a necessary artifact for the theoretical analysis of BP only because it implies that the Tanner graph is locally tree-like (i.e. has no cycles of length $\sim r$) asymptotically, thereby enabling inductive probability computations.
As a consequence, a fixed round complexity is a standard model choice in the density evolution analysis of BP~\cite{burshtein2002bounds,richardson2008modern,richardson2001capacity} as a method of enforcing asymptotic local tree-like structure of the Tanner graph.
In practice, if cycles are present, the performance of BP generally worsens, as the posterior probabilities necessary for deciding the decoded message may not converge at all.
In what follows, we define $p_{\text{BP}}^*(k)$ to be the maximum error probability of a binary symmetric channel in which a $r$-round belief propagation decoder has a failure probability that vanishes as $n \to \infty$.
The quantity does not depend on $r$, so long as it is fixed while $n$ tends to $\infty$.
This quantity's dependence on $d$ is implicitly defined by its dependence on $k$ when $r$ is fixed, since by regularity $k = (1-r)d = \lambda^{-1} d$.

\begin{theorem}[Limitations of belief propagation]
\label{thm:BP_limits}
    Fix any number of rounds $r$.
    There exists $k_0 > 0$ such that for all $k \geq k_0$, $p_{\operatorname{BP}}^*(k) \leq \frac{1}{4} \lambda^{-1} \frac{\ln k}{k}$.
\end{theorem}

Before we prove the above theorem, we introduce some notation involved in the execution of belief propagation (BP).
Each round involves a round of \emph{right-bound} messages sent from data nodes to check nodes, and a round of \emph{left-bound} messages sent from check nodes to data nodes. 
Fix an edge $(v, w)$ between a data node $v$ and a check node $w$.
Since we are concerned only with the asymptotic limit $n \to \infty$ and the Gallager ensemble is locally tree-like as $n \to \infty$, each data node has a unique parent node, which we define to be $w$.
A right-bound message $m_v \in [0,1]$ from $v$ to $w$ is interpreted as the posterior probability that the corresponding data bit is $1$.
For a fixed such data node $v$, let $X^{(t)}$ be the message sent by $v$ to $w$.
The final decoded message is given by $v = 1$ if $X^{(r)}  \geq 1/2$ and $v = 0$ otherwise.
Let $x_t := \mathbb{E}[X^{(t)}]$, where the expectation is taken over the randomness of the error.
The algorithm is initialized with $X^{(0)}$, which is defined as the prior probability that the transmitted bit corresponding to $v$ is $1$, given the noisy bit received under the binary symmetric channel with error probability $p$.
Since we work with linear codes, we may assume without loss of generality that the transmitted codeword is $0^n$.
Consequently, $X^{(0)}$ is $p$ if the received bit is $0$, and $1-p$ otherwise.
Hence, \begin{align}
    x_0 = p(1-p) + (1-p) p = 2 p (1 - p) .
\end{align}
Moreover, the decoding fails if $X^{(r)} \geq 1/2$, since the correct codeword is $0^n$.

The finite-$k$ upper bound on BP has been characterized analytically by \cite{burshtein2002bounds}.
Our proof proceeds by an asymptotic analysis of their finite bound; we rely on the following two facts from the finite analysis.
\begin{lemma}[\cite{burshtein2002bounds}, Theorem 3]
\label{lemma:burshtein_thm_3}
For any $\epsilon > 0$, $r \geq 1$, $k, d$, the following inequality holds for $n$ sufficiently large.
\begin{align}
    x_{r + 1} > (1 - \epsilon) x_0 [1 - (1 - 2 x_{r})^{d-1}]^{k-1} .
\end{align}
\end{lemma}

The choice of fixed $r$ is irrelevant in our proof as long as it is fixed so that the cycle-free assumption utilized in the proof of the above lemma holds; we henceforth disregard the exact choice of $r$.

\begin{lemma}[\cite{burshtein2002bounds}, Lemma 3]
\label{lemma:burshtein_lemma_3}
For $X^{(t)}$ as defined above, suppose that $x_t \geq \gamma$.
Then \begin{align}
    \Pr[X^{(t)} \geq \frac{1}{2}] \geq \frac{1}{2} (1 - \sqrt{1 - 2 \gamma}) .
\end{align}
\end{lemma}

With these lemmas in place, we next prove Theorem~\ref{thm:BP_limits}.
\begin{proof}[Proof of Theorem~\ref{thm:BP_limits}]
We consider a binary symmetric channel of probability $p = \alpha \frac{\ln k}{k}$, where $\alpha > 0$ is a constant to be determined.
Then \begin{align}
    x_0 = 2 p (1 - p) = 2 \alpha \frac{\ln k}{k} \left(1 - \alpha \frac{\ln k}{k} \right) = (2 \alpha - \operatorname{o}_k(1)) \frac{\ln k}{k} .
\end{align}
We prove by induction that $x_{t} \geq (1 - \operatorname{o}_k(1)) x_0$ for all $t \leq r$, when $n$ is sufficiently large.
More precisely, we first fix $k$, and take $n \to \infty$ so that we may apply Lemma~\ref{lemma:burshtein_thm_3}.
We then take the limit $k \to \infty$, giving rise to $\operatorname{o}_k(1)$ factors.
The induction base case is immediate; for the induction hypothesis, suppose that $x_{t} \geq (1 - \epsilon) x_0$.
Choose any $\epsilon > 0$ such that $\beta := 4 \lambda \alpha (1 - \epsilon) > 1$.
Then by Lemma~\ref{lemma:burshtein_thm_3}, with the notation $\lambda^{-1} = 1 - r$ so that $\lambda = \frac{d}{k}$, \begin{align}
    x_{t+1} & > (1 - \epsilon) x_0 [1 - (1 - 2 x_t)^{d-1}]^{k-1} \\
    & \geq (1 - \epsilon) x_0 [1 - \exp(-2 (d-1) x_t)]^{k-1} \\
    & \geq (1 - \epsilon) x_0 [1 - \exp(-2 (d-1) (1 - \epsilon) x_0)]^{k-1} \\
    & \geq (1 - \epsilon) x_0 \left[1 - \exp(-2 (d-1) (1 - \epsilon) (2 \alpha - \operatorname{o}_k(1)) \frac{\ln k}{k}) \right]^{k-1} \\
    & = (1 - \epsilon) x_0 \left[1 - \exp(-2 \lambda (1 - \epsilon) (2 \alpha - \operatorname{o}_k(1)) \ln k) \right]^{k-1} \\
    & = (1 - \epsilon) x_0 \left[1 - k^{- (1 - \epsilon) (4 \lambda \alpha - \operatorname{o}_k(1))} \right]^{k-1} \\
    & = (1 - \epsilon) x_0 \left[1 - k^{- (\beta -\operatorname{o}_k(1))} \right]^{k-1}  .
\end{align}
Here we have used the inequality $1 - x \leq e^{-x}$ and applied the induction hypothesis.
For all sufficiently large $k$, $\beta - \operatorname{o}_k(1) < 0$, and thus $k^{- (\beta -\operatorname{o}_k(1))} \in [0, 1]$.
We next apply the inequality $(1 - q)^m \geq 1 - mq$ for all positive integers $m$ and $q \in [0, 1]$.
This follows by a probabilistic interpretation: for $m$ independent coin flips with probability $q$ of heads, the probability that at least one heads appears is $1 - (1 - q)^m$.
At the same time, by a union bound, $1 - (1 - q)^m \leq mq$, from which the bound follows.
In our case, \begin{align}
    x_{t+1} & > (1 - \epsilon) x_0 [1 - (k-1) k^{- (\beta -\operatorname{o}_k(1))}] \\
    & \geq (1 - \epsilon) x_0 [1 - k \cdot k^{- (\beta -\operatorname{o}_k(1))}] \\
    & \geq (1 - \epsilon) x_0 [1 - k^{- (\beta - 1 -\operatorname{o}_k(1))}] \\
    & = (1 - \epsilon) x_0 (1 - \operatorname{o}_k(1)) ,
\end{align}
where the last line follows because $\beta > 1$ by definition.
Since we have freedom of choosing $\epsilon$ to be arbitrarily small so long as it does not depend on $n$, we choose $\epsilon = \operatorname{o}_k(1)$, e.g. $\epsilon = \frac{1}{10k}$.
This choice makes not issue in the requirement that $\beta > 1$, since $\beta$ increases as $\epsilon$ decreases.
Thus, as claimed, $x_{t+1} \geq (1 - \operatorname{o}_k(1)) x_0$.
To instantiate $\alpha$, we choose $\alpha = c \lambda^{-1}$ for any $c > 1/4$.
Next, we apply Lemma~\ref{lemma:burshtein_lemma_3}, so that \begin{align}
    \Pr[X^{(r)} \geq \frac{1}{2}] & \geq \frac{1}{2} \left( 1 - \sqrt{1 - 2 (1 - \operatorname{o}_k(1)) x_0} \right) \\
    & \geq \frac{1}{2} \left( 1 - \sqrt{1 - 2 (1 - \operatorname{o}_k(1)) (2 \alpha - \operatorname{o}_k(1)) \frac{\ln k}{k}} \right)  \\
    & = \frac{1}{2} \left( 1 - \sqrt{1 - (4 \alpha - \operatorname{o}_k(1)) \frac{\ln k}{k}} \right) \\
    & \geq \frac{1}{2} \left( (2 \alpha - \operatorname{o}_k(1)) \frac{\ln k}{k} \right) \\
    & = (1 - \operatorname{o}_k(1)) p ,
\end{align}
where we have used the inequality $\sqrt{1 - x} \leq 1 - \frac{x}{2}$ for $x \leq 1$.
Hence, the probability of failure fails to vanish as $n \to \infty$, when $k$ is sufficiently large.
In fact, the above calculation shows that for sufficiently large $k$, as $n \to \infty$, the decoded error probability is effectively no better than the error probability with no decoding at all.
This threshold is given by \begin{align}
    p^*_{\text{BP}}(k) \leq p = \alpha \frac{\ln k}{k} = c \lambda^{-1} \frac{\ln k}{k} .
\end{align}
Since this bound holds for all $c > 1/4$, it holds for $c = 1/4$.
\end{proof}

We further observe that, due to its algorithmic locality, BP is an inverse Lipschitz decoder, and this property is insensitive to the details of the decoder implementation.

\begin{lemma}[BP algorithms are inverse Lipschitz]
\label{kemma:BP_inverse_Lipschitz}
Let $\mathcal{D}$ be any decoding algorithm on a $(k, d)$-regular Tanner graph which operates as follows.
Each node is equipped with an internal state.
In each of $r$ rounds, each node sends a message to each of its neighbors, where the message is a function only of the node's state.
Then, each node updates its state based on all received messages.
At the beginning of the algorithm, each check node is initialized with the measured syndrome.
Then $\mathcal{D}$ is $L$-inverse Lipschitz, where $L \leq (kd)^{r/2 + 1}$.
\end{lemma}

\begin{proof}
We wish to show that $\left\lVert\mathcal{D}^{-1}(\bm{y}) - \mathcal{D}^{-1}(\bm{y}')\right\rVert_1 \leq L \lVert \bm{y} - \bm{y}' \rVert_1$ for any $\bm{y}, \bm{y}'$ in the range of $\mathcal{D}$.
Suppose that $\bm{y}$ and $\bm{y}'$ differ at a single bit.
We consider the backwards light cone of that bit in the algorithm.
The act of modifying one data bit at the end of the algorithm can affect at most $k$ check bits in the previous round, $r-1$.
Then, in round $r-2$, this effect can propagate to at most $d$ data bits per check bit, for a total of $kd$ bits.
Then, in round $r-3$, at most $kd \cdot k$ check bits are affected.
In general, every 2 rounds, the initial modification can propagate to affect at most $kd$ other check bits.
Hence, once we propagate to the beginning of the execution, at most $(kd)^{\lceil r/2 \rceil} \leq (kd)^{r/2 + 1}$ check bits are affected, and hence the initial syndrome given by $\mathcal{D}^{-1}$ differs by at most $(kd)^{r/2 + 1}$ bits.
More generally, each different bit between $\bm{y}$ and $\bm{y}'$ induces a backpropagating light cone of width $\leq (kd)^{r/2 + 1}$.
Their intersections can only decrease the number of bits affected, and thus \begin{align}
    \left\lVert\mathcal{D}^{-1}(\bm{y}) - \mathcal{D}^{-1}(\bm{y}')\right\rVert_1 \leq L \lVert \bm{y} - \bm{y}' \rVert_1
\end{align}
as claimed.
\end{proof}

\subsection{Limitations of Non-Local Decoders}

Beyond belief propagation-type approximate iterative decoders, state-of-the-art exact decoders in the literature are primarily based on linear programming (LP) relaxations and their semidefinite programming (SDP) generalizations~\cite{feldman2003using,koetter2005characterizations,ghazi2017lp}.
While these exact decoders are known to correct a typical $(k, d)$-LDPC code for all errors up to weight $\Omega(n/d)$, they are also upper bounded by the same threshold $\operatorname{O}(n/d)$.
We here provide a self-contained proof for the LP relaxation decoder; their generalizations in the LP/SDP hierarchies have been studied on an ensemble very similar to the Gallager ensemble and shown to fail on errors of relative weight $\frac{3}{d} (1 + \operatorname{o}_d(1))$~\cite{ghazi2017lp}.
Hence, we expect the bound to hold similarly for the Gallager ensemble.

The LP decoder arises from the observation that the problem of decoding can be expressed as an integer linear program (ILP).
Although the general ILP problem is NP-hard~\cite{papadimitriou1981complexity}, it becomes easy if we do not restrict the solutions to be integers~\cite{khachiyan1979polynomial}.
LP decoders thus express the decoding problem as an ILP, ``relax'' the constraints to not necessarily be integer-valued, and then run a LP solver in hopes that the solution returned is nevertheless the desired error.
More precisely, given a syndrome $\bm{s} \in \mathbb{Z}_2^n$, the ILP is given by \begin{align}
    \min_{\bm{x} \in \mathbb{Z}_2^m} \lVert \bm{x} \rVert_1 \; \text{subject to} \; \bm{Hx} = \bm{s} .
\end{align}
In the LP relaxation, the domain of $x_i$ broadens from $\{0,1\}$ to $[0,1]$ and the constraint is that the solution lies in the convex hull of bitstrings with syndrome $\bm{s}$.

It is known that for a typical $(k, d)$-LDPC code, the LP decoder can decode all errors up to relative weight $\Omega_d(1/d)$ with high probability over the choice of code~\cite{feldman2007lp}.
(This ensemble is a uniformly random $(k, d)$-regular Tanner graph, which differs slightly from the Gallager ensemble construction.)
At the same time, this threshold is optimal up to a universal constant: for \emph{any} $(k, d)$-regular Tanner graph, there exists an error of relative weight $\operatorname{O}_d(1/d)$ for which the LP decoder will fail to correct.
(This statement is independent of the choice of ensemble, and thus applies equally to the Gallager ensemble.)

\begin{theorem}[Limitations of the LP decoder]
\label{thm:LP_limitations}
For any $(k, d)$-regular Tanner graph, there exists an error of relative weight $\frac{1}{d} (1 + \operatorname{o}_d(1))$ such that the LP decoder will fail on $\bm{e}$ (i.e. output a vector which is not $\bm{e}$).
\end{theorem}

\begin{proof}
The relaxation implies that feasible solutions (convex combinations of true binary solutions) outputted by the decoder may not be bitstrings at all, and instead may be fractions.
Consider the non-integer error \begin{align}
    \Bar{\bm{x}} := \left( \frac{1}{d}, \dots, \frac{1}{d} \right) \in [0,1]^m .
\end{align}
We claim that regardless of the choice $\bm{s}$ and the choice of $\bm{H}$, so long as every check has sparsity $d$, $\Bar{\bm{x}}$ satisfies $\bm{Hx} = \bm{s}$.
For any given check $\bm{h}$ with support $S$ such that $|S| = d$, the corresponding syndrome bit is either 1 or 0.
In the former case, we may write $\Bar{\bm{x}}$ restricted to $S$ as \begin{align}
    \Bar{\bm{x}} \big\rvert_S = \frac{1}{d} \sum_{i \in S} \bm{1}_i ,
\end{align}
where $\bm{1}_i$ is 1 on the $i$th entry and 0 everywhere else.
This is a convex combination of errors which give syndrome bit 1.
In the latter case, we may also write $\Bar{\bm{x}}$ restricted to $S$ as \begin{align}
    \Bar{\bm{x}} \big\rvert_S = \frac{1}{d(d-1)} \sum_{i < j : i,j \in S} \bm{1}_{ij} ,
\end{align}
where $\bm{1}_{ij}$ have 1's in the $i$th and $j$th entries and 0's everywhere else.
Each $i \in S$ is represented $d-1$ times in this sum, so every entry uniformly has magnitude $1/d$ as required.
Each $\bm{1}_{ij}$ has syndrome bit 0, so this is a convex combination of errors which give syndrome bit 0.
Consequently, $\Bar{\bm{x}}$ satisfies the constraints of the relaxed LP problem irrespective of the details of the choice of $\bm{H}$.

For any error $\bm{e}$ of weight $w > \frac{m}{d}$, the LP decoder cannot output $\bm{e}$ as the solution, because$\Bar{\bm{x}}$ is a valid solution whose weight $m/d$ is smaller than that of $\bm{e}$.
Thus, the LP decoder fails to find $\bm{e}$.
The weight $w$ may be taken arbitrarily close to $m/d$ as claimed.
\end{proof}

\section{Background on the Quantum Wasserstein Distance}\label{sec:quant_wass_dist}

We here give background on the quantum Wasserstein distance. The quantum Wasserstein distance was originally introduced in~\cite{9420734}, though here we will mainly consider the slightly more general definition introduced in~\cite{anschuetz2025efficientlearningimpliesquantum}.

Informally, the quantum Wasserstein distance is a quantum ``earth mover's'' metric in that states which differ only by a channel acting on $\ell$ qubits differ in Wasserstein distance by $\operatorname{O}\left(\ell\right)$; in this way, it can be thought of as a quantum generalization of the Hamming distance (and indeed, it reduces to the Hamming distance on bit strings). More formally, it is defined in the following way. Here, $\left\lVert\cdot\right\rVert_1$ denotes the trace norm, and $\mathcal{O}_n$ is the space of Hermitian observables on $n$ qubits. We also use the notation $\Tr_{\mathcal{I}}$ to denote the partial trace when tracing out the qubits labeled by the index set $\mathcal{I}$. To begin, we define the quantum Wasserstein $F$-norm---note that this is not a true norm as it is not homogeneous.
\begin{definition}[Quantum Wasserstein $F$-norm~{\cite[Definition~43]{anschuetz2025efficientlearningimpliesquantum}}]
    Let $\bm{X}$ be a Hermitian, traceless observable on $n$ qubits. The \emph{quantum Wasserstein $F$-norm of order $p$} is defined as:
    \begin{equation}
        \left\lVert\bm{X}\right\rVert_{W_p}:=\min_{\left\{\bm{X}_i\right\}_{i=1}^n\in\mathcal{B}\left(\bm{X}\right)}\left(\sum_{i=1}^n\left\lVert\frac{1}{2}\bm{X}_i\right\rVert_1^{\frac{1}{p}}\right),
    \end{equation}
    where
    \begin{equation}
        \mathcal{B}\left(\bm{X}\right)=\left\{\left\{\bm{X}_i\right\}_{i=1}^n:\bm{X}=\sum_{i=1}^n\bm{X}_i\wedge\bm{X}_i\in\mathcal{O}_n\wedge\Tr_{\left\{i\right\}}\left(\bm{X}_i\right)=0\right\}.
    \end{equation}
\end{definition}
We can define a metric using this $F$-norm, which we call the quantum Wasserstein distance of order $p$.
\begin{definition}[Quantum Wasserstein distance of order $p$~{\cite[Definition~44]{anschuetz2025efficientlearningimpliesquantum}}]
    For $\bm{\rho},\bm{\sigma}\in\mathcal{S}_n^{\text{m}}$, their \emph{quantum Wasserstein distance of order $p$} is:
    \begin{equation}
        W_p\left(\bm{\rho},\bm{\sigma}\right):=\left\lVert\bm{\rho}-\bm{\sigma}\right\rVert_{W_p}.
    \end{equation}
\end{definition}

Unlike more traditional metrics on the space of quantum states---such as the trace distance---the quantum Wasserstein distance is not unitarily invariant, i.e., $\left\lVert\bm{U}\bm{X}\bm{U}^\dagger\right\rVert_{W_p}$ does not necessarily equal $\left\lVert\bm{X}\right\rVert_{W_p}$ for unitary $\bm{U}$. Furthermore, the norm is not necessarily contractive under quantum channels. That said, the metric still has some nice properties which we review in what follows.

First, there is an equivalence of quantum Wasserstein norms.
\begin{proposition}[Equivalence of quantum Wasserstein norms~{\cite[Proposition~45]{anschuetz2025efficientlearningimpliesquantum}}]\label{prop:wass_norm_equiv}
    For any $q\leq p$,
    \begin{equation}
        \left\lVert\bm{X}\right\rVert_{W_q}^q\leq\left\lVert\bm{X}\right\rVert_{W_p}^p\leq n^{p-q}\left\lVert\bm{X}\right\rVert_{W_q}^q.
    \end{equation}
\end{proposition}
Furthermore, there is an equivalence of norms between the quantum Wasserstein and trace norms.
\begin{proposition}[Equivalence of trace and quantum Wasserstein norms~{\cite[Proposition~2]{9420734}}]\label{prop:wass_trace_norm_equiv}
    For any traceless $\bm{X}\in\mathcal{O}_n$,
    \begin{equation}
        \frac{1}{2}\left\lVert\bm{X}\right\rVert_1\leq\left\lVert\bm{X}\right\rVert_{W_1}\leq\frac{n}{2}\left\lVert\bm{X}\right\rVert_1.
    \end{equation}
\end{proposition}

Second, though the quantum Wasserstein norm is generally not contractive under the action of quantum channels, it is contractive under the action of tensor-product channels.
\begin{proposition}[Contractivity under tensor product channels~{\cite[Proposition~47]{anschuetz2025efficientlearningimpliesquantum}}]\label{prop:contractivity_w2}
    For any channel of the form
    \begin{equation}
        \bm{\varLambda}=\frac{1}{B}\sum_{b=1}^B\bigotimes_{i=1}^n\bm{\varLambda}_i^{\left(b\right)},
    \end{equation}
    we have:
    \begin{equation}
        \left\lVert\bm{\varLambda}\left(\bm{X}\right)\right\rVert_{W_p}\leq\left\lVert\bm{X}\right\rVert_{W_p}.
    \end{equation}
    This inequality is saturated when $\bm{\varLambda}$ is a tensor product of unitary channels.
\end{proposition}
More generally, local operations have bounded influence on the quantum Wasserstein distance.
\begin{proposition}[Continuity of the $W_p$ distance]\label{prop:cont_wp}
    Let $\mathcal{I}\subseteq\left[n\right]$, and let $\bm{X}$ be a traceless Hermitian operator on $n$ qubits with the property $\Tr_{\mathcal{I}}\left(\bm{X}\right)=0$. Then,
    \begin{equation}
        \left\lVert\bm{X}\right\rVert_{W_p}\leq\left\lvert\mathcal{I}\right\rvert\left(\frac{3}{4}\right)^{\frac{1}{p}}\left\lVert\bm{X}\right\rVert_1.
    \end{equation}
\end{proposition}
\begin{proof}
    This is an immediate generalization of Proposition~5 of~\cite{9420734}.
\end{proof}
In particular, quantum channels acting on $k$ qubits only change the quantum $W_2$ distance by $\operatorname{O}\left(k\right)$; this is an immediate generalization of Corollary~2 of~\cite{9420734}.
\begin{corollary}[$W_p$ bound for local channels]\label{cor:bounded_w_p_change_loc}
    Let $\bm{\varLambda}$ be a superoperator acting on at most $k$ qubits. Then:
    \begin{equation}
        \left\lVert\bm{\rho}-\bm{\varLambda}\left(\bm{\rho}\right)\right\rVert_{W_p}\leq 2k\left(\frac{3}{4}\right)^{\frac{1}{p}}\left\lVert\bm{\varLambda}\left(\bm{\rho}\right)\right\rVert_{1\to 1},
    \end{equation}
    where $\left\lVert\cdot\right\rVert_{1\to 1}$ denotes the superoperator norm with respect to the trace norm.
\end{corollary}

Finally, the quantum Wasserstein distance over mixtures of product states in a shared basis upper bounds the classical Wasserstein distance. To state this result, we first define the notion of a \emph{coupling} between two probability distributions $p$ and $q$, specializing to discrete spaces for simplicity.
\begin{definition}[Coupling on a discrete space]
    Let $p$ and $q$ be probability distributions over a set $\mathcal{X}$ of finite cardinality. A probability distribution $\pi$ on $\mathcal{X}\times\mathcal{X}$ is called a \emph{coupling} between $p$ and $q$ if:
    \begin{align}
        p\left(x\right)&=\sum_{y\in\mathcal{X}}\pi\left(x,y\right),\\
        q\left(y\right)&=\sum_{x\in\mathcal{X}}\pi\left(x,y\right).
    \end{align}
\end{definition}
Couplings are used to define the classical Wasserstein distance, summarized as follows. Here, $d_{\mathrm{H}}$ denotes the Hamming distance.
\begin{definition}[Classical Wasserstein distances~{\cite[Definition~2]{9420734}}]
    The classical Wasserstein distance of order $\alpha$ between two distributions $p$ and $q$ over a discrete space $\mathcal{X}$ is defined as:
    \begin{equation}
        W_\alpha\left(p,q\right):=\inf_{\pi\in\mathcal{C}\left(p,q\right)}\left(\mathbb{E}_{\left(x,y\right)\sim\pi}d_{\text{H}}\left(x,y\right)^\alpha\right)^{\frac{1}{\alpha}}.
    \end{equation}
\end{definition}
We now state the relation between the quantum and classical Wasserstein distances.
\begin{proposition}[Quantum Wasserstein distance over mixtures of product states~{\cite[Proposition~54]{anschuetz2025efficientlearningimpliesquantum}}]\label{prop:quant_wass_prod_states}
    Consider quantum states $\bm{\rho}$ and $\bm{\sigma}$ mutually diagonalized by the same product state basis $\left\{\bm{s}\right\}_{\bm{s}\in\left\{0,1\right\}^{\times n}}$:
    \begin{align}
        \bm{\rho}&=\sum_{\bm{s}\in\left\{0,1\right\}^{\times n}} p\left(\bm{s}\right)\ket{\bm{s}}\bra{\bm{s}},\\
        \bm{\sigma}&=\sum_{\bm{s}\in\left\{0,1\right\}^{\times n}} q\left(\bm{s}\right)\ket{\bm{s}}\bra{\bm{s}}.
    \end{align}
    Let $\mathcal{C}\left(p,q\right)$ be the set of couplings between $p$ and $q$. Then, for any $\alpha\geq 1$,
    \begin{equation}\label{eq:quant_wass_red}
        \inf_{\pi\in\mathcal{C}\left(p,q\right)}\left(\mathbb{E}_{\left(\bm{s},\bm{t}\right)\sim\pi}\left\lVert\ket{\bm{s}}\bra{\bm{s}}-\ket{\bm{t}}\bra{\bm{t}}\right\rVert_{W_1}^\alpha\right)^{\frac{1}{\alpha}}\leq W_\alpha\left(\bm{\rho},\bm{\sigma}\right)\leq\inf_{\pi\in\mathcal{C}\left(p,q\right)}\sum_{\bm{s},\bm{t}\in\left\{0,1\right\}^{\times n}}\pi\left(\bm{s},\bm{t}\right)^{\frac{1}{\alpha}}\left\lVert\ket{\bm{s}}\bra{\bm{s}}-\ket{\bm{t}}\bra{\bm{t}}\right\rVert_{W_1}.
    \end{equation}
\end{proposition}

\section{Proof of Theorem~\ref{thm:stab_algs_fail}}\label{sec:ogp_implies_alg_hardness}

In this Appendix we prove that stable quantum algorithms fail to optimize \textsc{MAX-$k$-XOR-SAT} at any approximation ratio exhibiting an $R$-OGP. We also show that stable quantum algorithms with sufficiently good stability parameters are obstructed by the weaker chaos property; we will later see that the fixed-subset chaos property suffices for obstructing DQI. These results are presented as Theorem~\ref{thm:stab_algs_fail} of the main text, which we restate here for convenience.
\stableqasfailgallager*

We now prove Theorem~\ref{thm:stab_algs_fail}. At a high level, our strategy follows that of~\cite[Theorem~16]{anschuetz2025efficientlearningimpliesquantum}, but differs in a few key ways:
\begin{itemize}
    \item As we only consider classical problems, the proof can be simplified substantially.
    \item As the eigenbasis of the problem is fixed, we can account for shot noise in the failure probability of the algorithm.
    \item The $\left(k,S\right)$-minimum Hamming semimetric (Definition~\ref{def:k_min_ham_semi}) which we use to define (and later prove) our OGP and chaos property does not satisfy the triangle inequality, introducing complications to the proof.
    \item $f$ no longer need depend on $Q$, which greatly improves the satisfied fraction at which DQI (or any other stable algorithm with nontrivial $f$) is obstructed. This is achieved by slightly modifying the definition of stability from~\cite{anschuetz2025efficientlearningimpliesquantum}, which requires further modifications to the proof.
    \item Our strategy requires a novel strategy to interpolate between correlated problem instances as we are no longer considering Gaussian randomness.
\end{itemize}
Before proceeding with the proof, we give a sketch of our approach, lettered by subsection in what follows. We proceed by contradiction, assuming there exists a quantum algorithm $\bm{\mathcal{A}}$ that is $\left(f,L,p_{\text{st}}\right)$-stable and $\left(\gamma,p_{\text{f}}\right)$-optimal for \textsc{MAX-$k$-XOR-SAT} instances drawn from $\mathbb{P}_{\mathrm{G}}$.
\begin{enumerate}[label=(\Alph*)]
    \item We show that the existence of $\bm{\mathcal{A}}$ implies the existence of a stable, near-optimal quantum algorithm which is also deterministic.
    \item We show that the existence of a stable, near-optimal, deterministic quantum algorithm for this problem implies the existence of a stable, near-optimal classical algorithm $\bm{\mathcal{I}}$ (at the cost of worse constants).
    \item We consider an interpolation path over many replicas, and show that $\bm{\mathcal{I}}$ is stable and near-optimal over the replicas with high probability.
    \item Due to \textsc{MAX-$k$-XOR-SAT} with independent parities having distant solutions (as implied by the chaos property, Definition~\ref{def:chaos_prop}), we show that with high probability all $R$-tuples of $T$ independently-sampled instances have near-optimal states which are distant in $\left(k,S\right)$-minimum Hamming semimetric with high probability.
    \item We show that with high probability there exists some point along the interpolation path where this algorithm outputs a configuration disallowed by the multi-OGP due to the pairwise-stability and near-optimality of $\bm{\mathcal{I}}$.
    \item Finally, we show that there exist choices of parameters such that all ``with high probability'' events have a nontrivial intersection. This contradicts the assumption of the existence of $\bm{\mathcal{A}}$.
\end{enumerate}

\subsection{Reduction to Deterministic Quantum Algorithms}

We first prove that, WLOG, one can consider deterministic quantum algorithms. The proof follows a similar strategy as~\cite[Lemma~6.11]{gamarnik2022algorithmsbarrierssymmetricbinary}.
\begin{lemma}[Reduction to deterministic quantum algorithms]\label{lem:det_rand_alg_red}
    Let $\bm{\mathcal{A}}\left(\bm{X},\omega\right)$ be a quantum algorithm that is both $\left(f,L,p_{\text{st}}\right)$-stable and $\left(\gamma,p_{\text{f}}\right)$-optimal for $g_{\bm{X}}$ over $\bm{X}\sim\mathbb{P}_{\mathrm{G}}$. Fix any $\mathcal{K}=\left\{\kappa\right\}\subseteq\left[0,1\right]$. Then, there exists a deterministic quantum algorithm $\bm{\widetilde{\mathcal{A}}}\left(\bm{X}\right)$ that is both $\left(f,L,\mathcal{K},3p_{\text{st}}\right)$-stable and $\left(\gamma,3p_{\text{f}}\right)$-optimal for $g_{\bm{X}}$.
\end{lemma}
\begin{proof}
    Let $\left(\varOmega,\mathbb{P}_\varOmega\right)$ be the probability space associated with $\bm{\mathcal{A}}$ and let $g^\ast n$ be the expected maximal value of $g_{\bm{X}}$ over $\bm{X}\sim\mathbb{P}_{\mathrm{G}}$. Recall the definition of the measurement channel $\bm{\widetilde{\mathcal{M}}}$ (Eq.~\eqref{eq:comp_meas_chan_def}). For notational convenience, we define the events for all $\omega\in\varOmega$, $\nu\in\left[0,1\right]$, and $\bm{X}=\left(\bm{B},\bm{v}\right),\bm{Y}=\left(\bm{B},\bm{v'}\right)\in\mathbb{F}_2^D$:
    \begin{align}
        \mathcal{E}_{\text{st}}^{\left(\omega\right)}\left(\bm{X},\bm{Y}\right)&:=\left\{\exists S_{\bm{B}}\in\binom{\left[n\right]}{f}:\left\lVert\Tr_{S_{\bm{B}}}\left(\bm{\mathcal{A}}\left(\bm{X},\omega\right)-\bm{\mathcal{A}}\left(\bm{Y},\omega\right)\right)\right\rVert_{W_2}\leq L\left\lVert\bm{v}-\bm{v'}\right\rVert_1\right\},\\
        \mathcal{E}_{\text{no}}^{\left(\omega\right)}\left(\bm{X},\upsilon\right)&:=\left\{g_{\bm{X}}\left(\bm{\widetilde{\mathcal{M}}}\left(\bm{\mathcal{A}}\left(\bm{X},\omega\right),\upsilon\right)\right)\geq\gamma \lambda n\right\},
    \end{align}
    and the events for all $\omega\in\varOmega$:
    \begin{align}
        \mathcal{T}_{\text{st}}\left(\omega\right)&:=\left\{\mathbb{P}_{\left(\bm{X},\bm{Y}\right)\sim\mathbb{P}_2^{\left(\kappa\right)}}\left[\mathcal{E}_{\text{st}}^{\left(\omega\right)}\left(\bm{X},\bm{Y}\right)^\complement\right]>3p_{\text{st}}\right\}=\left\{\mathbb{P}_{\left(\bm{X},\bm{Y}\right)\sim\mathbb{P}_2^{\left(\kappa\right)}}\left[\mathcal{E}_{\text{st}}^{\left(\omega\right)}\left(\bm{X},\bm{Y}\right)\right]\leq 1-3p_{\text{st}}\right\},\\
        \mathcal{T}_{\text{no}}\left(\omega\right)&:=\left\{\mathbb{P}_{\left(\bm{X},\upsilon\right)\sim\mathbb{P}_{\mathrm{G}}\otimes\mathcal{U}}\left[\mathcal{E}_{\text{no}}^{\left(\omega\right)}\left(\bm{X},\upsilon\right)^\complement\right]>3p_{\text{f}}\right\}=\left\{\mathbb{P}_{\left(\bm{X},\upsilon\right)\sim\mathbb{P}_{\mathrm{G}}\otimes\mathcal{U}}\left[\mathcal{E}_{\text{no}}^{\left(\omega\right)}\left(\bm{X},\upsilon\right)\right]\leq 1-3p_{\text{f}}\right\}.
    \end{align}
    By the law of total probability,
    \begin{align}
        \mathbb{E}_{\omega\sim\mathbb{P}_\varOmega}\left[\mathbb{P}_{\left(\bm{X},\bm{Y}\right)\sim\mathbb{P}_2^{\left(\kappa\right)}}\left[\mathcal{E}_{\text{st}}^{\left(\omega\right)}\left(\bm{X},\bm{Y}\right)^\complement\right]\right]&=\mathbb{P}_{\left(\bm{X},\bm{Y},\omega\right)\sim\mathbb{P}_2^{\left(\kappa\right)}\otimes\mathbb{P}_\varOmega}\left[\mathcal{E}_{\text{st}}^{\left(\omega\right)}\left(\bm{X},\bm{Y}\right)^\complement\right]\leq p_{\text{st}},\\
        \mathbb{E}_{\omega\sim\mathbb{P}_\varOmega}\left[\mathbb{P}_{\left(\bm{X},\upsilon\right)\sim\mathbb{P}_{\mathrm{G}}\otimes\mathcal{U}}\left[\mathcal{E}_{\text{no}}^{\left(\omega\right)}\left(\bm{X},\upsilon\right)^\complement\right]\right]&=\mathbb{P}_{\left(\bm{X},\omega,\upsilon\right)\sim\mathbb{P}_{\mathrm{G}}\otimes\mathbb{P}_\varOmega\otimes\mathcal{U}}\left[\mathcal{E}_{\text{no}}^{\left(\omega\right)}\left(\bm{X},\upsilon\right)^\complement\right]\leq p_{\text{f}},
    \end{align}
    where the inequalities follow from the stability and near-optimality of $\bm{\mathcal{A}}$.
    By Markov's inequality,
    \begin{align}
        \mathbb{P}_{\omega\sim\mathbb{P}_\varOmega}\left[\mathcal{T}_{\text{st}}\left(\omega\right)\right]&\leq\frac{\mathbb{E}_{\omega\sim\mathbb{P}_\varOmega}\left[\mathbb{P}_{\left(\bm{X},\bm{Y}\right)\sim\mathbb{P}_2^{\left(\kappa\right)}}\left[\mathcal{E}_{\text{st}}^{\left(\omega\right)}\left(\bm{X},\bm{Y}\right)^\complement\right]\right]}{3p_{\text{st}}}\leq\frac{1}{3},\\
        \mathbb{P}_{\omega\sim\mathbb{P}_\varOmega}\left[\mathcal{T}_{\text{no}}\left(\omega\right)\right]&\leq\frac{\mathbb{E}_{\omega\sim\mathbb{P}_\varOmega}\left[\mathbb{P}_{\left(\bm{X},\upsilon\right)}\left[\mathcal{E}_{\text{no}}^{\left(\omega\right)}\left(\bm{X},\upsilon\right)^\complement\right]\right]}{3p_{\text{f}}}\leq\frac{1}{3}.
    \end{align}
    Furthermore, by the union bound,
    \begin{equation}
        \mathbb{P}_{\omega\sim\mathbb{P}_\varOmega}\left[\mathcal{T}_{\text{no}}\left(\omega\right)\cup\mathcal{T}_{\text{st}}\left(\omega\right)\right]\leq\frac{2}{3}.
    \end{equation}
    In particular, there must exist $\omega^\ast\in\varOmega$ such that the event
    \begin{equation}
        \left(\mathcal{T}_{\text{no}}\left(\omega^\ast\right)\cup\mathcal{T}_{\text{st}}\left(\omega^\ast\right)\right)^\complement=\mathcal{T}_{\text{no}}\left(\omega^\ast\right)^\complement\cap\mathcal{T}_{\text{st}}\left(\omega^\ast\right)^\complement
    \end{equation}
    occurs. By definition, then,
    \begin{equation}
        \bm{\widetilde{\mathcal{A}}}\left(\bm{X}\right):=\bm{\mathcal{A}}\left(\bm{X},\omega^\ast\right)
    \end{equation}
    is a deterministic quantum algorithm that is both $\left(f,L,\mathcal{K},3p_{\text{st}}\right)$-stable and $\left(\gamma,3p_{\text{f}}\right)$-optimal for $g_{\bm{X}}$.
\end{proof}

We also bound the probability over $\bm{X}$ that the probability of achieving a high-energy state over $\upsilon\sim\mathcal{U}$ is large.
\begin{lemma}[Controlling shot noise]\label{lem:controlling_shot_noise}
    Let $\bm{\mathcal{A}}\left(\bm{X}\right)$ be a $\left(\gamma,3p_{\text{f}}\right)$-optimal deterministic quantum algorithm for $g_{\bm{X}}$ over $\bm{X}\sim\mathbb{P}_{\mathrm{G}}$. Then:
    \begin{equation}
        \mathbb{P}_{\bm{X}\sim\mathbb{P}_{\mathrm{G}}}\left[\mathbb{P}_{\upsilon\sim\mathcal{U}}\left[g_{\bm{X}}\left(\bm{\widetilde{\mathcal{M}}}\left(\bm{\mathcal{A}}\left(\bm{X}\right),\upsilon\right)\right)\geq\gamma\lambda n\right]\geq 1-\sqrt{3p_{\text{f}}}\right]\geq 1-\sqrt{3p_{\text{f}}}.
    \end{equation}
\end{lemma}
\begin{proof}
    Recall the definition of the measurement channel $\bm{\widetilde{\mathcal{M}}}$ (Eq.~\eqref{eq:comp_meas_chan_def}). For notational convenience, we define the event for all $\nu\in\left[0,1\right]$ and $\bm{X}$:
    \begin{equation}
        \mathcal{E}_{\text{no}}^{\left(\bm{X}\right)}\left(\upsilon\right):=\left\{g_{\bm{X}}\left(\bm{\widetilde{\mathcal{M}}}\left(\bm{\mathcal{A}}\left(\bm{X}\right),\upsilon\right)\right)\geq\gamma\lambda n\right\},
    \end{equation}
    as well as the event for all $\bm{X}$:
    \begin{equation}\label{eq:near_opt_event}
        \mathcal{T}_{\text{no}}\left(\bm{X}\right):=\left\{\mathbb{P}_{\upsilon\sim\mathcal{U}}\left[\mathcal{E}_{\text{no}}^{\left(\bm{X}\right)}\left(\upsilon\right)^\complement\right]>\sqrt{3p_{\text{f}}}\right\}=\left\{\mathbb{P}_{\upsilon\sim\mathcal{U}}\left[\mathcal{E}_{\text{no}}^{\left(\bm{X}\right)}\left(\upsilon\right)\right]\leq 1-\sqrt{3p_{\text{f}}}\right\}.
    \end{equation}
    By the law of total probability and the $\left(\gamma,3p_{\text{f}}\right)$-near optimality of $\bm{\mathcal{A}}$,
    \begin{equation}
        \mathbb{E}_{\bm{X}\sim\mathbb{P}_{\mathrm{G}}}\left[\mathbb{P}_{\upsilon}\left[\mathcal{E}_{\text{no}}^{\left(\bm{X}\right)}\left(\upsilon\right)^\complement\right]\right]=\mathbb{P}_{\left(\bm{X},\upsilon\right)\sim\mathbb{P}_{\mathrm{G}}\otimes\mathcal{U}}\left[\mathcal{E}_{\text{no}}^{\left(\bm{X}\right)}\left(\upsilon\right)^\complement\right]\leq 3p_{\text{f}}.
    \end{equation}
    By Markov's inequality,
    \begin{equation}
        \mathbb{P}_{\bm{X}\sim\mathbb{P}_{\mathrm{G}}}\left[\mathcal{T}_{\text{no}}\left(\bm{X}\right)\right]\leq\frac{\mathbb{E}_{\bm{X}\sim\mathbb{P}_{\mathrm{G}}}\left[\mathbb{P}_{\upsilon\sim\mathcal{U}}\left[\mathcal{E}_{\text{no}}^{\left(\bm{X}\right)}\left(\upsilon\right)^\complement\right]\right]}{\sqrt{3p_{\text{f}}}}\leq\sqrt{3p_f}.
    \end{equation}
    Substituting the definition of $\mathcal{T}_{\text{no}}\left(\bm{X}\right)$ yields the final result.
\end{proof}

\subsection{Measurement Reduction}

We now prove that, if there exists a stable, near-optimal quantum algorithm for a classical problem class, there exists a near-optimal classical algorithm which satisfies a weaker notion of stability than that given in Definition~\ref{def:stable_qas}. In what follows, we use $\mathcal{U}$ to denote the uniform distribution over $\left[0,1\right]$, and we use $\mathcal{B}$ to denote the set of pure computational basis states.
\begin{lemma}[Nondeterministic classical shadows reduction]\label{lem:class_shads_red}
    Assume there exists an $\left(f,L,\left\{\kappa\right\},3Qp_{\text{st}}\right)$-stable and $\left(\gamma,3p_{\text{f}}\right)$-near optimal deterministic quantum algorithm $\bm{\mathcal{A}}\left(\bm{X}\right)$ for $g_{\bm{X}}$ over $\bm{X}\sim\mathbb{P}_{\mathrm{G}}$. Then, there exists a pure quantum algorithm $\bm{\mathcal{G}}:\mathbb{R}^D\times\left[0,1\right]\to\mathcal{B}$ such that, for $\left(\bm{X},\bm{Y}\right)=\left(\left(\bm{B},\bm{v}\right),\left(\bm{B},\bm{v'}\right)\right)\sim\mathbb{P}_2^{\left(\kappa\right)}$,
    \begin{equation}\label{eq:g_close_to_stab}
        \mathbb{P}_{\left(\bm{X},\bm{Y}\right)\sim\mathbb{P}_2^{\left(\kappa\right)}}\left[\exists S_{\bm{B}}\in\binom{\left[n\right]}{f}:\left\lVert\Tr_{S_{\bm{B}}}\left(\bm{\mathcal{G}}\left(\bm{X},\nu\right)-\bm{\mathcal{G}}\left(\bm{Y},\nu\right)\right)\right\rVert_{W_2}\leq L\left\lVert\bm{v}-\bm{v'}\right\rVert_1\right]\geq 1-3p_{\text{st}}.
    \end{equation}
    Furthermore, for any $\bm{X}$ and $\upsilon\in\left[0,1\right]$ satisfying the probability $1-\sqrt{3p_{\text{f}}}$ event $\mathcal{T}_{\text{no}}\left(\bm{X}\right)^\complement$ (Eq.~\eqref{eq:near_opt_event}),
    \begin{equation}\label{eq:g_near_opt}
        \mathbb{P}_{\upsilon\sim\mathcal{U}}\left[g_{\bm{X}}\left(\bm{\mathcal{G}}\left(\bm{X},\upsilon\right)\right)\geq\gamma\lambda n\right]\geq 1-\sqrt{3p_{\text{f}}}.
    \end{equation}
\end{lemma}
\begin{proof}
    We begin by considering Eq.~\eqref{eq:g_close_to_stab}. As the quantum Wasserstein distance is nonincreasing under convex combinations of tensor product channels (see Proposition~\ref{prop:contractivity_w2}, reviewed in Appendix~\ref{sec:quant_wass_dist}), we have for all $\bm{\rho},\bm{\sigma}\in\mathcal{S}_n^{\text{m}}$:
    \begin{equation}
        \left\lVert\mathbb{E}_{\upsilon\sim\mathcal{U}}\left[\bm{\widetilde{\mathcal{M}}}\left(\bm{\rho},\upsilon\right)-\bm{\widetilde{\mathcal{M}}}\left(\bm{\sigma},\upsilon\right)\right]\right\rVert_{W_2}:=\left\lVert\bm{\mathcal{M}}\left(\bm{\rho}\right)-\bm{\mathcal{M}}\left(\bm{\sigma}\right)\right\rVert_{W_2}\leq\left\lVert\bm{\rho}-\bm{\sigma}\right\rVert_{W_2}.
    \end{equation}
    In particular, defining:
    \begin{equation}\label{eq:g_def}
        \bm{\mathcal{G}}\left(\bm{X},\upsilon\right):=\bm{\widetilde{\mathcal{M}}}\left(\bm{\mathcal{A}}\left(\bm{X}\right),\upsilon\right),
    \end{equation}
    we have:
    \begin{equation}\label{eq:quant_stab_red}
        \mathbb{P}_{\left(\bm{X},\bm{Y}\right)\sim\mathbb{P}_2^{\left(\kappa\right)}}\left[\exists S_{\bm{B}}\in\binom{\left[n\right]}{f}:\left\lVert\Tr_{S_{\bm{B}}}\left(\bm{\mathcal{G}}\left(\bm{X},\nu\right)-\bm{\mathcal{G}}\left(\bm{Y},\nu\right)\right)\right\rVert_{W_2}\leq L\left\lVert\bm{v}-\bm{v'}\right\rVert_1\right]\geq 1-3p_{\text{st}}
    \end{equation}
    since $\bm{\mathcal{A}}\left(\bm{X}\right)$ is $\left(f,L,\mathcal{K},3p_{\text{st}}\right)$-stable.

    Finally, Eq.~\eqref{eq:g_near_opt} follows immediately from substituting Eq.~\eqref{eq:g_def} into Lemma~\ref{lem:controlling_shot_noise}, yielding the final result.
\end{proof}

Ideally, we would strengthen this notion of stability to that of Definition~\ref{def:stable_qas}, which would allow us to directly leverage the machinery of classical algorithmic hardness results based on OGPs. Unfortunately, it is generally the case that:
\begin{equation}
    \left\lVert\mathbb{E}_{\nu\sim\mathcal{U}}\left[\Tr_{S_{\bm{B}}}\left(\bm{\mathcal{G}}\left(\bm{X},\nu\right)-\bm{\mathcal{G}}\left(\bm{Y},\nu\right)\right)\right]\right\rVert_{W_2}\neq\mathbb{E}_{\nu\sim\mathcal{U}}\left[\left\lVert\Tr_{S_{\bm{B}}}\left(\bm{\mathcal{G}}\left(\bm{X},\nu\right)-\bm{\mathcal{G}}\left(\bm{Y},\nu\right)\right)\right\rVert_{W_2}\right],
\end{equation}
making such a strengthening generally impossible. However, we can get close; as the quantum Wasserstein distance upper bounds the classical Wasserstein distance for quantum states in the computational basis (see Proposition~\ref{prop:quant_wass_prod_states}), we have the following consequence of stability according to $\left\lVert\mathbb{E}_{\nu\sim\mathcal{U}}\left[\Tr_{S_{\bm{B}}}\left(\bm{\mathcal{G}}\left(\bm{X},\nu\right)-\bm{\mathcal{G}}\left(\bm{Y},\nu\right)\right)\right]\right\rVert_{W_2}$. Note that the quantum Wasserstein distance of order $1$ is exactly the Hamming distance on computational basis states~\cite{9420734}, so here $\left\lVert\ket{\bm{s}}\bra{\bm{s}}-\ket{\bm{t}}\bra{\bm{t}}\right\rVert_{W_1}$ can simply be thought of as $d_{\mathrm{H}}\left(\bm{s},\bm{t}\right)$.
\begin{proposition}[Expected Wasserstein distance~{\cite[Proposition~26]{anschuetz2025efficientlearningimpliesquantum}}]\label{prop:exp_wass_dist}
    For each $\bm{X}$, let $p_{\bm{X}}\left(\ket{\bm{s}}\bra{\bm{s}}\right)$ be the distribution of $\bm{\mathcal{G}}\left(\bm{X},\nu\right)$ over $\nu\sim\mathcal{U}$, i.e.,
    \begin{equation}
        \mathbb{E}_{\nu\sim\mathcal{U}}\left[\bm{\mathcal{G}}\left(\bm{X},\nu\right)\right]=:\sum_{\bm{s}\in\mathbb{F}_2^n}p_{\bm{X}}\left(\ket{\bm{s}}\bra{\bm{s}}\right)\ket{\bm{s}}\bra{\bm{s}}.
    \end{equation}
    Then, for every pair $\left(\bm{X},\bm{Y}\right)$, there exists some probability distribution $\pi_{\left(\bm{X},\bm{Y}\right)}\left(\ket{\bm{s}}\bra{\bm{s}},\ket{\bm{t}}\bra{\bm{t}}\right)$ over $\left(\ket{\bm{s}}\bra{\bm{s}},\ket{\bm{t}}\bra{\bm{t}}\right)\in\mathcal{B}^2$ satisfying:
    \begin{align}
        \mathbb{E}_{\left(\ket{\bm{s}}\bra{\bm{s}},\ket{\bm{t}}\bra{\bm{t}}\right)\sim\pi_{\left(\bm{X},\bm{Y}\right)}}\left[\left\lVert\ket{\bm{s}}\bra{\bm{s}}-\ket{\bm{t}}\bra{\bm{t}}\right\rVert_{W_1}^2\right]&\leq\left\lVert\mathbb{E}_{\nu\sim\mathcal{U}}\left[\bm{\mathcal{G}}\left(\bm{X},\nu\right)-\bm{\mathcal{G}}\left(\bm{Y},\nu\right)\right]\right\rVert_{W_2}^2,\\
        \sum_{\ket{\bm{t}}\bra{\bm{t}}\in\mathcal{B}}\pi_{\left(\bm{X},\bm{Y}\right)}\left(\ket{\bm{s}}\bra{\bm{s}},\ket{\bm{t}}\bra{\bm{t}}\right)&=p_{\bm{X}}\left(\ket{\bm{s}}\bra{\bm{s}}\right),\label{eq:first_coupling_marg}\\
        \sum_{\ket{\bm{s}}\bra{\bm{s}}\in\mathcal{B}}\pi_{\left(\bm{X},\bm{Y}\right)}\left(\ket{\bm{s}}\bra{\bm{s}},\ket{\bm{t}}\bra{\bm{t}}\right)&=p_{\bm{Y}}\left(\ket{\bm{t}}\bra{\bm{t}}\right).\label{eq:sec_coupling_marg}
    \end{align}
\end{proposition}
Unfortunately, $\pi_{\left(\bm{X},\bm{Y}\right)}$ depends on $\left(\bm{X},\bm{Y}\right)$, so this fact cannot be immediately applied to demonstrate the stability of $\bm{\mathcal{G}}$. However, this fact will be enough to allow us to prove the existence of a different stable, pure quantum algorithm $\bm{\mathcal{I}}$. We formalize this as the following lemma.
\begin{lemma}[Reduction to pure, deterministic algorithms]\label{lem:stab_loc_shad_red_along_path}
    Let $\bm{\mathcal{A}}$, $\mathcal{T}_{\text{no}}\left(\bm{X}\right)$, $\kappa$, and $\bm{\mathcal{G}}$ be as in Lemma~\ref{lem:class_shads_red}. Fix $Q\in\mathbb{N}$. Let $\mathbb{P}_Q$ be a distribution over $\bm{X}=\left(\bm{X}_q\right)_{q=0}^Q=\left(\left(\bm{B},\bm{v}_q\right)\right)_{q=0}^Q\in\mathbb{R}^{D\times\left(Q+1\right)}$ such that the marginal distribution over any pair $\left(\bm{X}_q,\bm{X}_{q+1}\right)$ is $\mathbb{P}_2^{\left(\kappa\right)}$. Assume that
    \begin{equation}
        p_{\text{f}}<\frac{1}{3\left(Q+1\right)^2},
    \end{equation}
    and fix any
    \begin{equation}
        \beta>\sqrt{\frac{Q}{1-\left(Q+1\right)\sqrt{3p_{\text{f}}}}}.
    \end{equation}
    Consider as well the event:
    \begin{equation}\label{eq:exp_ir_stab_small}
        \mathcal{C}_{\bm{X}}:=\bigcap_{q=0}^{Q-1}\left\{\exists S_{\bm{B}}\in\binom{\left[n\right]}{f}:\left\lVert\Tr_{S_{\bm{B}}}\left(\bm{\mathcal{A}}\left(\bm{X}_q\right)-\bm{\mathcal{A}}\left(\bm{X}_{q+1}\right)\right)\right\rVert_{W_2}\leq L\left\lVert\bm{v}_q-\bm{v}_{q+1}\right\rVert_1\right\}\cap\bigcap_{q=0}^Q\mathcal{T}_{\text{no}}\left(\bm{X}_q\right)^\complement,
    \end{equation}
    which occurs with probability:
    \begin{equation}\label{eq:prob_c_event}
        \mathbb{P}_{\bm{X}\sim\mathbb{P}_Q}\left[\mathcal{C}_{\bm{X}}\right]\geq 1-\left(3Qp_{\text{st}}+\left(Q+1\right)\sqrt{3p_{\text{f}}}+Q\exp\left(-\operatorname{\Omega}\left(n\right)\right)\right).
    \end{equation}

    For any $\bm{X}\sim\mathbb{P}_Q$ where the event $\mathcal{C}_{\bm{X}}$ occurs, there exists a pure, deterministic quantum algorithm $\bm{\mathcal{I}}:\left\{q\right\}_{q=0}^Q\to\mathcal{B}$ satisfying for all integer $0\leq q\leq Q-1$:
    \begin{equation}
        \left\lVert\Tr_{S_{\bm{B}}}\left(\bm{\mathcal{I}}\left(q\right)-\bm{\mathcal{I}}\left(q+1\right)\right)\right\rVert_{W_1}\leq\beta L\left\lVert\bm{v}_q-\bm{v}_{q+1}\right\rVert_1
    \end{equation}
    and, for all $q\in\left[Q\right]$,
    \begin{equation}
        g_{\bm{X}_q}\left(\bm{\mathcal{I}}\left(q\right)\right)\geq\gamma\lambda n.
    \end{equation}
\end{lemma}
\begin{proof}
    Consider $\bm{X}\sim\mathbb{P}_Q$. Let $\pi_{\left(\bm{X}_q,\bm{X}_{q+1}\right)}\left(\ket{\bm{s}_q}\bra{\bm{s}_q},\ket{\bm{s}_{q+1}}\bra{\bm{s}_{q+1}}\right)$ be as in Proposition~\ref{prop:exp_wass_dist}. We have:
    \begin{equation}
        \mathbb{E}_{\left(\ket{\bm{s}_q}\bra{\bm{s}_q},\ket{\bm{s}_{q+1}}\bra{\bm{s}_{q+1}}\right)\sim\pi_{\left(\bm{X}_q,\bm{X}_{q+1}\right)}}\left[\left\lVert\Tr_{S_{\bm{B}}}\left(\ket{\bm{s}}\bra{\bm{s}}-\ket{\bm{t}}\bra{\bm{t}}\right)\right\rVert_{W_1}^2\right]\leq\left\lVert\mathbb{E}_{\omega\sim\mathcal{U}}\left[\Tr_{S_{\bm{B}}}\left(\bm{\mathcal{G}}\left(\bm{X}_q,\omega\right)-\bm{\mathcal{G}}\left(\bm{X}_{q+1},\omega\right)\right)\right]\right\rVert_{W_2}^2.
    \end{equation}
    We also recall from Proposition~\ref{prop:exp_wass_dist}:
    \begin{equation}\label{eq:cons_of_margs}
        \begin{aligned}
            p_{\bm{X}_q}\left(\ket{\bm{s}_q}\bra{\bm{s}_q}\right)&=\sum_{\ket{\bm{s}_{q+1}}\bra{\bm{s}_{q+1}}\in\mathcal{B}}\pi_{\left(\bm{X}_q,\bm{X}_{q+1}\right)}\left(\ket{\bm{s}_q}\bra{\bm{s}_q},\ket{\bm{s}_{q+1}}\bra{\bm{s}_{q+1}}\right)\\
            &=\sum_{\ket{\bm{s}_{q-1}}\bra{\bm{s}_{q-1}}\in\mathcal{B}}\pi_{\left(\bm{X}_{q-1},\bm{X}_q\right)}\left(\ket{\bm{s}_{q-1}}\bra{\bm{s}_{q-1}},\ket{\bm{s}_q}\bra{\bm{s}_q}\right),
        \end{aligned}
    \end{equation}
    where the final equality holds due to the compatibility of the marginals of the $\pi_{\left(\bm{X}_q,\bm{X}_{q+1}\right)}$ (Eqs.~\eqref{eq:first_coupling_marg} and~\eqref{eq:sec_coupling_marg}). Finally, we define the probability distribution $\varPi_{\bm{X}}$ over $\left(\ket{\bm{s}_q}\bra{\bm{s}_q}\right)_{q=0}^Q\in\mathcal{B}^{Q+1}$ given by:
    \begin{equation}
        \varPi_{\bm{X}}\left(\left(\ket{\bm{s}_i}\bra{\bm{s}_i}\right)_{q=0}^Q\right):=\pi_{\left(\bm{X}_0,\bm{X}_1\right)}\left(\ket{\bm{s}_0}\bra{\bm{s}_0},\ket{\bm{s}_1}\bra{\bm{s}_1}\right)\prod_{q=1}^{Q-1}\frac{\pi_{\left(\bm{X}_q,\bm{X}_{q+1}\right)}\left(\ket{\bm{s}_q}\bra{\bm{s}_q},\ket{\bm{s}_{q+1}}\bra{\bm{s}_{q+1}}\right)}{p_{\bm{X}_q}\left(\ket{\bm{s}_q}\bra{\bm{s}_q}\right)}.
    \end{equation}
    Using the consistency of the single-variable marginals (Eq.~\eqref{eq:cons_of_margs}), it is easy to see by direct calculation that the two-variable marginals of $\varPi_{\bm{X}}$ agree with the $\pi_{\left(\bm{X}_q,\bm{X}_{q+1}\right)}$:
    \begin{equation}\label{eq:bigpi_two_var_margs}
        \varPi_{\bm{X}}\left(\ket{\bm{s}_q}\bra{\bm{s}_q},\ket{\bm{s}_{q+1}}\bra{\bm{s}_{q+1}}\right)=\pi_{\left(\bm{X}_q,\bm{X}_{q+1}\right)}\left(\ket{\bm{s}_q}\bra{\bm{s}_q},\ket{\bm{s}_{q+1}}\bra{\bm{s}_{q+1}}\right).
    \end{equation}

    Now, define a sample space $\varOmega:=\mathcal{B}^{Q+1}$. We use the notation $\omega_q$ (zero-indexed) to denote the projection of $\omega\in\varOmega$ to the $q$th of the factors $\mathcal{B}$. With this notation, we define the pure quantum algorithm $\bm{\widetilde{\mathcal{I}}}:\left\{q\right\}_{q=0}^Q\times\varOmega\to\mathcal{B}$:
    \begin{equation}
        \bm{\widetilde{\mathcal{I}}}\left(q,\omega\right)=\omega_q.
    \end{equation}
    By Markov's inequality, conditioned on $\bm{X}$ being such that the event $\mathcal{C}_{\bm{X}}$ occurs, we have from Markov's inequality that for any $\beta>0$:
    \begin{equation}
        \begin{aligned}
            \mathbb{P}_{\omega\sim\varPi_{\bm{X}}\mid\mathcal{C}_{\bm{X}}}&\left[\left\lVert\Tr_{S_{\bm{B}}}\left(\bm{\widetilde{\mathcal{I}}}\left(q,\omega\right)-\bm{\widetilde{\mathcal{I}}}\left(q+1,\omega\right)\right)\right\rVert_{W_1}\geq\beta L\left\lVert\bm{v}_q-\bm{v}_{q+1}\right\rVert_1\right]\\
            &=\mathbb{P}_{\omega\sim\varPi_{\bm{X}}\mid\mathcal{C}_{\bm{X}}}\left[\left\lVert\Tr_{S_{\bm{B}}}\left(\bm{\widetilde{\mathcal{I}}}\left(q,\omega\right)-\bm{\widetilde{\mathcal{I}}}\left(q+1,\omega\right)\right)\right\rVert_{W_1}^2\geq\beta^2 L^2\left\lVert\bm{v}_q-\bm{v}_{q+1}\right\rVert_1^2\right]\\
            &\leq\frac{\mathbb{E}_{\omega\sim\varPi_{\bm{X}}\mid\mathcal{C}_{\bm{X}}}\left[\left\lVert\Tr_{S_{\bm{B}}}\left(\bm{\widetilde{\mathcal{I}}}\left(q,\omega\right)-\bm{\widetilde{\mathcal{I}}}\left(q+1,\omega\right)\right)\right\rVert_{W_1}^2\right]}{\beta^2 L^2\left\lVert\bm{v}_q-\bm{v}_{q+1}\right\rVert_1^2}\\
            &\leq\frac{1}{\beta^2}
        \end{aligned}
    \end{equation}
    for any integer $0\leq q\leq Q-1$. By the union bound,
    \begin{equation}\label{eq:stab_r_algs}
        \mathbb{P}_{\omega\sim\varPi_{\bm{X}}\mid\mathcal{C}_{\bm{X}}}\left[\bigcap_{q=0}^{Q-1}\left\lVert\Tr_{S_{\bm{B}}}\left(\bm{\widetilde{\mathcal{I}}}\left(q,\omega\right)-\bm{\widetilde{\mathcal{I}}}\left(q+1,\omega\right)\right)\right\rVert_{W_1}\leq\beta L\left\lVert\bm{v}_q-\bm{v}_{q+1}\right\rVert_1\right]\geq 1-\frac{Q}{\beta^2}.
    \end{equation}
    Furthermore, by the union bound (and conditioned on $\mathcal{C}_{\bm{X}}$),
    \begin{equation}\label{eq:no_r_algs}
        \mathbb{P}_{\omega\sim\varPi_{\bm{X}}\mid\mathcal{C}_{\bm{X}}}\left[\bigcap_{q=0}^Q f_{\bm{X}_q}\left(\bm{\widetilde{\mathcal{I}}}\left(q,\omega\right)\right)\geq\gamma\lambda n\right]\geq 1-\left(Q+1\right)\sqrt{3p_{\text{f}}}.
    \end{equation}
    
    In particular, assuming $\beta$ is sufficiently large such that:
    \begin{equation}
        \frac{Q}{\beta^2}+\left(Q+1\right)\sqrt{3p_{\text{f}}}<1,
    \end{equation}
    we have from the law of total probability that there exists some $\omega^\ast\in\varOmega=\mathcal{B}^{Q+1}$ such that the events in Eqs.~\eqref{eq:stab_r_algs} and~\eqref{eq:no_r_algs} occur. The final result follows by defining:
    \begin{equation}
        \bm{\mathcal{I}}\left(q\right):=\bm{\widetilde{\mathcal{I}}}\left(q,\omega^\ast\right).
    \end{equation}
\end{proof}

\subsection{Considering Many Replicas}

We now apply Lemmas~\ref{lem:det_rand_alg_red} through~\ref{lem:stab_loc_shad_red_along_path} in sequence to $T+1$ replicas, and choose a specific interpolation path which we construct in the following way. First, we fix natural numbers $T$ and $Q$. We sample $\bm{B}\sim\mathbb{P}_{\mathrm{con}}$ and $\left(\bm{\tilde{v}}^{\left(t\right)}\right)_{t=0}^T\sim\mathbb{P}_{\mathrm{G}}$ independently, and define for each $t\in\left[T\right]$ and integer $0\leq q\leq Q$:
\begin{equation}
    \bm{v}_q^{\left(t\right)}:=\bm{\tilde{v}}^{\left(0\right)}\oplus\bm{\varUpsilon}_{\left\lfloor\frac{qm}{Q}\right\rfloor}\left(\bm{\tilde{v}}^{\left(0\right)}\oplus\bm{\tilde{v}}^{\left(t\right)}\right),
\end{equation}
where we recall the definition of $\bm{\varUpsilon}_p$ (Eq.~\eqref{eq:upsilon_proj_def}):
\begin{equation}
    \bm{\varUpsilon}_p:=\operatorname{diag}\left(\underbrace{1,\ldots,1}_p,\underbrace{0,\ldots,0}_{m-p}\right).
\end{equation}
Finally, we define for all $t\in\left[T\right]$ and integer $0\leq q\leq Q$:
\begin{equation}
    \bm{X}_q^{\left(t\right)}:=\left(\bm{B},\bm{v}_q^{\left(t\right)}\right).
\end{equation}
By construction, any pair $\left(\bm{X}_q^{\left(t\right)},\bm{X}_{q+1}^{\left(t\right)}\right)$ is marginally distributed as $\mathbb{P}_2^{\left(\frac{m}{Q}\right)}$. Furthermore, at $q=0$ all of the $\bm{X}_q^{\left(t\right)}$ are identical; at $q=Q$, their parities $\bm{v}_q^{\left(t\right)}$ are i.i.d. We let $\mathbb{P}_{T,Q}$ denote the joint distribution over all $\left(\left(\bm{X}_q^{\left(t\right)}\right)_{q=0}^Q\right)_{t=1}^T$ sampled in this way.

Our main result here is combining all of the lemmas proven to this point, and bounding the probability that collective stability and near-optimality holds for $\bm{\mathcal{I}}$ over $\mathbb{P}_{T,Q}$.
\begin{proposition}[Considering many replicas]\label{prop:red_pure_prod_algs}
    Fix $T,Q\in\mathbb{N}$. Let $\bm{\mathcal{A}}$ be a quantum algorithm that is both $\left(f,L,p_{\text{st}}\right)$-stable and $\left(\gamma,p_{\text{f}}\right)$-optimal for $g_{\bm{X}}$ over $\bm{X}\sim\mathbb{P}_{\mathrm{G}}$. Assume that
    \begin{equation}
        p_{\text{f}}<\frac{1}{3\left(Q+1\right)^2},
    \end{equation}
    and fix any
    \begin{equation}
        \beta>\sqrt{\frac{Q}{1-\left(Q+1\right)\sqrt{3p_{\text{f}}}}}.
    \end{equation}
    Then,
    \begin{equation}
        \mathbb{P}_{\bm{X}\sim\mathbb{P}_{T,Q}}\left[\mathcal{Y}\right]\geq 1-3TQp_{\text{st}}-\left(TQ+1\right)\sqrt{3p_{\text{f}}}-TQ\exp\left(-\operatorname{\Omega}\left(n\right)\right),
    \end{equation}
    where we have defined the event:
    \begin{equation}\label{eq:suc_int}
        \begin{aligned}
            \mathcal{Y}&:=\exists\left\{\bm{\mathcal{I}}_t:\left\{q\right\}_{q=0}^Q\to\mathcal{B}\right\}_{t=1}^T,S_{\bm{B}}\in\binom{\left[n\right]}{f}:\\
            &\left\{\bigcap_{t,t'=1}^T\bm{\mathcal{I}}_t\left(0\right)=\bm{\mathcal{I}}_{t'}\left(0\right)\right\}\\
            &\cap\left\{\left\{\bigcap_{t=1}^T\bigcap_{q=0}^{Q-1}\frac{1}{n}\left\lVert\Tr_{S_{\bm{B}}}\left(\bm{\mathcal{I}}_t\left(q\right)-\bm{\mathcal{I}}_t\left(q+1\right)\right)\right\rVert_{W_1}\leq\beta L\frac{\lambda}{Q}\right\}\right.\\
            &\cap\left.\left\{\bigcap_{t=1}^T\bigcap_{q=1}^Qg_{\bm{X}_q}\left(\bm{\mathcal{I}}_t\left(q\right)\right)\geq\gamma\lambda n\right\}\right\}.
        \end{aligned}
    \end{equation}
\end{proposition}
\begin{proof}
    We begin by conditioning on the event $\mathcal{C}_{\bm{X}}$ defined as in Lemma~\ref{lem:stab_loc_shad_red_along_path}. First, note that by construction of the interpolation path:
    \begin{equation}
        \left\lVert\bm{v}_q^{\left(t\right)}-\bm{v}_{q+1}^{\left(t\right)}\right\rVert_1\leq\frac{m}{Q}=\frac{\lambda}{Q}n
    \end{equation}
    for any $t\in\left[T\right]$ and integer $0\leq q\leq Q-1$. The final two events composing $\mathcal{Y}$ then follow immediately from applying Lemma~\ref{lem:stab_loc_shad_red_along_path} to $\left(\bm{X}_q^{\left(t\right)}\right)_{q=0}^Q$ for each $t\in\left[T\right]$---each yielding the function $\bm{\mathcal{I}}_t$---and the union bound. The consistency relation at $q=0$ follows by inspecting the proof of Lemma~\ref{lem:stab_loc_shad_red_along_path} and noting that $\bm{\mathcal{I}}_t\left(0\right)$ depends only on $\bm{X}_0^{\left(t\right)}$; as the $\bm{X}_0^{\left(t\right)}$ are identical across all $t\in\left[T\right]$, so are the $\bm{\mathcal{I}}_t\left(0\right)$. The final result holds by recalling the probability of the event $\mathcal{C}_{\bm{X}}$ occurring from Lemma~\ref{lem:stab_loc_shad_red_along_path}, as well as the union bound.
\end{proof}

\subsection{Distant Clustering for Independent Instances}

We here bound the probability of the \emph{chaos property} occurring. This is the second aspect of Theorem~\ref{thm:e_ogp}, which states that independent instances have distant near-optimal states.
\begin{lemma}[Distant clustering for independent instances]\label{lem:ind_clust_prob}
    Conditioned on $S\in\binom{\left[n\right]}{f}$, fix $\left(R,\mu,\nu_2,S\right)$ such that $g_{\bm{X}}$ satisfies the chaos property (Definition~\ref{def:chaos_prop}) with these parameters. Then:
    \begin{equation}
        \mathbb{P}_{\bm{X}\sim\mathbb{P}_{T,Q}}\left[\mathcal{Z}\mid S\right]\geq 1-\binom{T}{R}\exp\left(-\operatorname{\Omega}\left(n\right)\right),
    \end{equation}
    where
    \begin{equation}
       \mathcal{Z}\mid S:=\bigcap_{\mathcal{R}\in\binom{\left[T\right]}{R}}\left\{\mathcal{S}_{\left(\bm{X}_Q^{\left(r\right)}\right)_{r\in\mathcal{R}}}^{\left(R,\mu,0,\nu_2,S\right)}=\varnothing\right\}.
    \end{equation}
\end{lemma}
\begin{proof}
    As the $\bm{v}_Q^{\left(r\right)}$ are independent by construction, the desired result follows immediately from the union bound and the definition of the chaos property.
\end{proof}

\subsection{Topologically Obstructed Configurations Conditioned on Events}

Our strategy is now to show, conditioned on all of the previously-introduced events occurring, that the algorithm must output configurations that are topologically obstructed by either the $R$-OGP or the chaos property. This will follow from a similar argument as that in \cite{10.1214/23-AAP1953,anschuetz2025efficientlearningimpliesquantum}, though we reproduce it here in full for completeness. In what follows, we recall the $\left(k,S\right)$-minimum Hamming semimetric (Definition~\ref{def:k_min_ham_semi}):
\begin{equation}
    d_{k,S}\left(\bm{z},\bm{z'}\right):=\sum_{i=1}^k\min\left(\left\lVert\bm{\varSigma}_S\bm{\varGamma}_i\left(\bm{z}\oplus\bm{z'}\right)\right\rVert_1,\left\lVert\bm{\varSigma}_S\bm{\varGamma}_i\left(\bm{z}\oplus\bm{z'}\oplus\bm{1}\right)\right\rVert_1\right).
\end{equation}

First, we sample $\bm{X}\sim\mathbb{P}_{T,Q}$ and condition on the events $\mathcal{Y}$ (from Proposition~\ref{prop:red_pure_prod_algs}) and $\mathcal{Z}$ (from Lemma~\ref{lem:ind_clust_prob}). We let $\nu_1,\nu_2$ be the associated $R$-OGP parameters; if only the chaos property is satisfied, we take $\nu_1=0$. We construct a graph $G_{T,Q}$ with colored edges in the following way:
\begin{enumerate}
    \item $G_{T,Q}$ has vertex set $\left[T\right]$;
    \item The $G_{T,Q}$ has edge $\left(t,t'\right)$ if and only if there exists a $q\in\left[Q\right]$ such that:
    \begin{equation}
        d_{k,S_{\bm{B}}}\left(\bm{\mathcal{I}}_t\left(q\right),\bm{\mathcal{I}}_{t'}\left(q\right)\right)\in\left[\nu_1 n,\nu_2 n\right].
    \end{equation}
    The color of the edge $\left(t,t'\right)$ is the minimal $q\in\left[Q\right]$ such that this holds true.
\end{enumerate}
Our main goal in this section is to demonstrate that $G_{T,Q}$ has a monochromatic $R$-clique, which we know from Theorem~\ref{thm:e_ogp} is obstructed by the multi-OGP.

First, we claim that $G_{T,Q}$ is \emph{$R$-admissible} when conditioned on $\mathcal{Y}$ and $\mathcal{Z}$, defined as follows.
\begin{definition}[$R$-admissibility]
    Let $R\in\mathbb{N}$. A graph $G=\left(V,E\right)$ is said to be \emph{$R$-admissible} if, for every $\mathcal{R}\subseteq V$ with $\left\lvert\mathcal{R}\right\rvert=R$, there exist distinct $i,j\in\mathcal{R}$ such that $\left(i,j\right)\in E$.
\end{definition}
We now show that $G_{T,Q}$ is $R$-admissible when conditioned on $\mathcal{Y}$ and $\mathcal{Z}$.
\begin{lemma}[Topologically obstructed configurations conditioned on events]\label{lem:top_obs}
    Assume that:
    \begin{equation}\label{eq:suff_stab_m_admiss}
        \frac{\beta\lambda L}{Q}<\frac{\nu_2-\nu_1}{2}
    \end{equation}
    and $\gamma\geq\mu$. Draw $\left(\left(\bm{X}_q^{\left(t\right)}\right)_{q=0}^Q\right)_{t=1}^T$ according to $\mathbb{P}_{T,Q}$, and condition on the events $\mathcal{Y}$ (from Proposition~\ref{prop:red_pure_prod_algs}) and $\mathcal{Z}$ (from Lemma~\ref{lem:ind_clust_prob}) occurring with $S_{\bm{B}}=S$. Then, $G_{T,Q}$ is $R$-admissible.
\end{lemma}
\begin{proof}
    Define the function over $t,t'\in\left[T\right]$ and integer $0\leq q\leq Q$:
    \begin{equation}
        p\left(t,t',q\right):=\frac{1}{n}d_{k,S_{\bm{B}}}\left(\bm{\mathcal{I}}_t\left(q\right),\bm{\mathcal{I}}_{t'}\left(q\right)\right).
    \end{equation}
    While the $\left(k,S\right)$-minimum Hamming semimetric does not necessarily satisfy the triangle inequality, we claim that it has the following property for any $\bm{z},\bm{z'},\bm{z''}\in\mathbb{F}_2^n$:
    \begin{equation}
        d_{k,S}\left(\bm{z},\bm{z'}\right)\leq d_{k,S}\left(\bm{z},\bm{z''}\right)+d_{\mathrm{H}}\left(\bm{\varSigma}_S\bm{z'},\bm{\varSigma}_S\bm{z''}\right),
    \end{equation}
    where $d_{\mathrm{H}}$ denotes the Hamming distance. To see this, denoting $\bm{1}\in\mathbb{F}_2^n$ as the all-ones vector, we have:
    \begin{equation}
        \begin{aligned}
            d_{k,S}\left(\bm{z},\bm{z'}\right)&=\sum_{i=1}^k\min\left(\left\lVert\bm{\varSigma}_S\bm{\varGamma}_i\left(\bm{z}\oplus\bm{z'}\right)\right\rVert_1,\left\lVert\bm{\varSigma}_S\bm{\varGamma}_i\left(\bm{z}\oplus\bm{z'}\oplus\bm{1}\right)\right\rVert_1\right)\\
            &=\sum_{i=1}^k\min\left(\left\lVert\bm{\varSigma}_S\bm{\varGamma}_i\left(\bm{z}\oplus\bm{z''}\oplus\bm{z'}\oplus\bm{z''}\right)\right\rVert_1,\left\lVert\bm{\varSigma}_S\bm{\varGamma}_i\left(\bm{z}\oplus\bm{z''}\oplus\bm{z'}\oplus\bm{z''}\oplus\bm{1}\right)\right\rVert_1\right)\\
            &\leq\sum_{i=1}^k\min\left(\left\lVert\bm{\varSigma}\bm{\varGamma}_i\left(\bm{z}\oplus\bm{z''}\right)\right\rVert+\left\lVert\bm{\varSigma}\bm{\varGamma}_i\left(\bm{z'}\oplus\bm{z''}\right)\right\rVert_1,\left\lVert\bm{\varSigma}\bm{\varGamma}_i\left(\bm{z}\oplus\bm{z''}\oplus\bm{1}\right)\right\rVert+\left\lVert\bm{\varSigma}\bm{\varGamma}_i\left(\bm{z'}\oplus\bm{z''}\right)\right\rVert_1\right)\\
            &=d_{k,S}\left(\bm{z},\bm{z''}\right)+d_{\mathrm{H}}\left(\bm{\varSigma}_S\bm{z'},\bm{\varSigma}_S\bm{z''}\right).
        \end{aligned}
    \end{equation}
    Therefore, conditioned on $\mathcal{Y}$, by Proposition~\ref{prop:red_pure_prod_algs} we have that $p$ is Lipschitz in $q$ in that:
    \begin{equation}\label{eq:lip_of_p}
        \begin{aligned}
            \left\lvert p\left(t,t',q\right)-p\left(t,t',q+1\right)\right\rvert&\leq\frac{1}{n}d_{\mathrm{H}}\left(\Tr_{S_{\bm{B}}}\left(\bm{\mathcal{I}}_t\left(q\right)\right),\Tr_{S_{\bm{B}}}\left(\bm{\mathcal{I}}_t\left(q+1\right)\right)\right)+\frac{1}{n}d_{\mathrm{H}}\left(\Tr_{S_{\bm{B}}}\left(\bm{\mathcal{I}}_{t'}\left(q\right)\right),\Tr_{S_{\bm{B}}}\left(\bm{\mathcal{I}}_{t'}\left(q+1\right)\right)\right)\\
            &\leq\frac{2\beta\lambda L}{Q}\\
            &<\nu_2-\nu_1,
        \end{aligned}
    \end{equation}
    where the final inequality is due to Eq.~\eqref{eq:suff_stab_m_admiss}. Furthermore, by construction $p\left(0\right)=0$. Finally, writing $\bm{X}_q^{\left(t\right)}=\left(\bm{B},\bm{v}_q^{\left(t\right)}\right)$, recall that the $\bm{v}_Q^{\left(t\right)}$ are independent. Therefore, conditioned on $\mathcal{Z}$ with $S=S_{\bm{B}}$, by the $\left(\gamma,0\right)$-near-optimality of $\bm{\mathcal{I}}$ (when conditioned on $\mathcal{Y}$) we have that there exist some $s^\ast\neq t^\ast\in\left[T\right]$ such that:
    \begin{equation}
        p\left(s^\ast,t^\ast,Q\right)\geq\nu_2.
    \end{equation}
    Let $q^\ast$ be the smallest $q\in\left[Q\right]$ satisfying $p\left(s^\ast,t^\ast,q^\ast\right)\geq\nu_2$. In particular, $p\left(s^\ast,t^\ast,q^\ast-1\right)<\nu_2$. Furthermore, by Eq.~\eqref{eq:lip_of_p},
    \begin{equation}
        p\left(s^\ast,t^\ast,q^\ast-1\right)>\nu_2-\left(\nu_2-\nu_1\right)=\nu_1.
    \end{equation}
    Combining these two observations,
    \begin{equation}
        p\left(s^\ast,t^\ast,q^\ast-1\right)\in\left(\nu_1,\nu_2\right).
    \end{equation}
    Recalling the definition of $p\left(t,t',q\right)$ and $G_{T,Q}$ yields the final result.
\end{proof}

We now cite a result linking $R$-admissibility to the existence of a monochromatic $R$-clique.
\begin{proposition}[$G$ contains a monochromatic $R$-clique~{\cite[Proposition~6.12]{10.1214/23-AAP1953}}]
    Assume $G$ is $R$-admissible, has $C$ edge colors, and has $\exp_2\left(C^{4mC}\right)$ vertices. Then, $G$ has a monochromatic clique of cardinality $R$.
\end{proposition}
This immediately gives the following result when applied to $G_{T,Q}$.
\begin{proposition}[$G_{T,Q}$ contains a monochromatic $R$-clique]\label{prop:m_clique}
    If $T=\exp_2\left(Q^{4RQ}\right)$, $G_{T,Q}$ conditioned on $\mathcal{Y}$ and $\mathcal{Z}$ with $S_{\bm{B}}=S$ has a monochromatic clique of cardinality $R$ if:
    \begin{equation}
        \frac{\beta\lambda L}{Q}<\frac{\nu_2-\nu_1}{2}
    \end{equation}
    and $\gamma\geq\mu$.
\end{proposition}

\subsection{Completing the Proof}

We now have all of the ingredients to complete Theorem~\ref{thm:stab_algs_fail}. First, we lower bound the probability that the event:
\begin{equation}
    \mathcal{W}:=\bigsqcup_{S\in\binom{\left[n\right]}{f}}\left\{\mathcal{Y}\mid S_{\bm{B}}=S\right\}\cap\left\{\mathcal{Z}\mid S\right\}
\end{equation}
occurs, with $\mathcal{Y}$ and $\mathcal{Z}$ defined in Proposition~\ref{prop:red_pure_prod_algs} and Lemma~\ref{lem:ind_clust_prob}, respectively.
\begin{lemma}[Probability of good events]\label{lem:prob_good_events}
    Fix $T$ and $Q$ as constants independent of $n$. Assume $p_{\text{st}}$, $p_{\text{f}}$, and $\beta\in\mathbb{R}^+$ are such that:
    \begin{align}
        3TQp_{\text{st}}+\left(TQ+1\right)\sqrt{3p_{\text{f}}}&\leq 1-\exp\left(-\operatorname{o}\left(n\right)\right);\\
        \beta&>\sqrt{\frac{Q}{1-\left(Q+1\right)\sqrt{3p_{\text{f}}}}}.
    \end{align}
    Then,
    \begin{equation}
        \mathbb{P}_{\mathbb{P}_k}\left[\mathcal{W}\right]\geq\exp\left(-\operatorname{o}\left(n\right)\right).
    \end{equation}
\end{lemma}
\begin{proof}
    By Proposition~\ref{prop:red_pure_prod_algs}, Lemma~\ref{lem:ind_clust_prob}, and the union bound, $\mathcal{W}$ occurs with probability at least:
    \begin{equation}
        \begin{aligned}
            \mathbb{P}_{\mathbb{P}_k}\left[\mathcal{W}\right]&\geq 1-3TQp_{\text{st}}+\left(TQ+1\right)\sqrt{3p_{\text{f}}}-TQ\exp\left(-\operatorname{\Omega}\left(n\right)\right)\\
            &\geq\exp\left(-\operatorname{o}\left(n\right)\right).
        \end{aligned}
    \end{equation}
\end{proof}

We now use Proposition~\ref{prop:m_clique} to show a contradiction with the statement of Theorem~\ref{thm:stab_algs_fail}. Once again we assume $\mu,\nu_1,\nu_2$ are as in the assumed $R$-OGP; if the chaos property assumption is taken instead, fix $\nu_1=0$ in what follows. Conditioned on $\mathcal{W}$ (and given the assumed bounds on $\gamma,p_{\text{f}}$ and $f,L,p_{\text{st}}$), Proposition~\ref{prop:m_clique} states that there exists some $0\leq q\leq Q-1$ and $\mathcal{R}\in\binom{\left[T\right]}{R}$ such that, for all $t\neq t'\in\mathcal{R}$,
\begin{align}
    d_{k,S_{\bm{B}}}\left(\bm{\mathcal{I}}_t\left(q\right),\bm{\mathcal{I}}_{t'}\left(q\right)\right)&\in\left[\nu_1 n,\nu_2 n\right],\\
    g_{\bm{X}_q^{\left(t\right)}}\left(\bm{\mathcal{I}}_t\left(q\right)\right)&\geq\gamma\lambda n,\\
    g_{\bm{X}_q^{\left(t'\right)}}\left(\bm{\mathcal{I}}_{t'}\left(q\right)\right)&\geq\gamma\lambda n.
\end{align}
Namely, recalling the definition of the set $\mathcal{S}_{\left(\bm{X}_q^{\left(r\right)}\right)_{r\in\mathcal{R}}}$ in the definition of the multi-OGP (Definition~\ref{def:multi_ogp}), conditioned on $\mathcal{W}$ it is the case that this set is nonempty. This implies that:
\begin{equation}
    \begin{aligned}
        \mathbb{P}_{\bm{X}\sim\mathbb{P}_{T,Q}}\left[\mathcal{S}_{\left(\bm{X}_q^{\left(r\right)}\right)_{r\in\mathcal{R}}}^{\left(R,\mu,\nu_1,\nu_2,S_{\bm{B}}\right)}\neq\varnothing\right]&\geq\mathbb{P}_{\bm{X}\sim\mathbb{P}_{T,Q}}\left[\mathcal{W}\right]\\
        &\geq\exp\left(-\operatorname{o}\left(n\right)\right),
    \end{aligned}
\end{equation}
where the final line follows from Lemma~\ref{lem:prob_good_events}.

However, this contradicts both the $R$-OGP and the chaos property. For the former, from Definition~\ref{def:multi_ogp},
\begin{equation}
    \mathbb{P}_{\bm{X}\sim\mathbb{P}_{T,Q}}\left[\mathcal{S}_{\left(\bm{X}_q^{\left(r\right)}\right)_{r\in\mathcal{R}}}^{\left(R,\mu,\nu_1,\nu_2,S_{\bm{B}}\right)}\neq\varnothing\right]\leq\exp\left(-\operatorname{\Omega}\left(n\right)\right)
\end{equation}
if $\gamma\geq\mu$, yielding a contradiction. In the case of the chaos property, when $Q=1$ the $\bm{X}_1^{\left(r\right)}$ are constrained to independent, so by Definition~\ref{def:chaos_prop}:
\begin{equation}
    \mathbb{P}_{\bm{X}\sim\mathbb{P}_{T,Q}}\left[\mathcal{S}_{\left(\bm{X}_q^{\left(r\right)}\right)_{r\in\mathcal{R}}}^{\left(R,\mu,\nu_1,\nu_2,S_{\bm{B}}\right)}\neq\varnothing\right]\leq\exp\left(-\operatorname{\Omega}\left(n\right)\right)
\end{equation}
if $\gamma\geq\mu$, once again yielding a contradiction. Taking:
\begin{align}
    \beta&=\sqrt{\frac{Q+1}{1-\left(Q+1\right)\sqrt{3p_{\text{f}}}}},\\
    \delta&=\frac{1}{6TQ}=\frac{1}{6Q\exp_2\left(Q^{4RQ}\right)},\label{eq:delta_stab_def}
\end{align}
with $Q$ fixed to be $1$ if the chaos property is satisfied, then proves Theorem~\ref{thm:stab_algs_fail}.

\section{Maximum Value of Transposed Gallager \texorpdfstring{\textsc{MAX-$k$-XOR-SAT}}{MAX-k-XOR-SAT}}\label{sec:max_value_gallager_xor_sat}

We here prove the maximum w.h.p.\ attainable value for \textsc{MAX-$k$-XOR-SAT} instances drawn from the transposed Gallager ensemble (Definition~\ref{def:trans_gall_ens}). First, we prove that \textsc{MAX-$k$-XOR-SAT} over this distribution exhibits \emph{concentration}, namely, that the maximum number of satisfied clauses exponentially concentrates around its mean. We then show that the maximum fraction of satisfied clauses is, with probability exponentially close to $1$ over this ensemble, given by:
\begin{equation}\label{eq:theta_star_def}
    \theta^\ast:=1-\operatorname{H}_2^{-1}\left(1-\frac{1}{\lambda}\right)
\end{equation}
at sufficiently large $k$. In the large $\lambda$ limit, we can expand the inverse binary entropy near $1$ to find this expression takes the form:
\begin{equation}
    \theta^\ast=\frac{1}{2}+\sqrt{\frac{\ln\left(2\right)}{2\lambda}}+\operatorname{o}_\lambda\left(\lambda^{-\frac{1}{2}}\right).
\end{equation}
This functional form was conjectured by~\cite{jordan2025optimizationdecodedquantuminterferometry} for \textsc{MAX-$k$-XOR-SAT} instances drawn from this ensemble; we prove this is indeed the correct functional form.

We now formally state the main result of this section.
\begin{theorem}[Maximum value of transposed Gallager \textsc{MAX-$k$-XOR-SAT}]\label{thm:max_value}
    Let $\theta^\ast$ be defined as in Eq.~\eqref{eq:theta_star_def}. For any $\epsilon>0$,
    \begin{equation}\label{eq:lim_max_ub}
        \mathbb{P}_{\left(\bm{B},\bm{v}\right)\sim\mathbb{P}_{\mathrm{G}}}\left[\frac{1}{m}\max_{\bm{z}\in\mathbb{F}_2^n}\left(m-\left\lVert\bm{B}\bm{z}\oplus\bm{v}\right\rVert_1\right)-\theta^\ast\geq\epsilon\right]\leq\exp\left(-\operatorname{\Omega}\left(n\right)\right).
    \end{equation}
    Furthermore, for any $\epsilon>0$, there exists a $k_\epsilon=\operatorname{O}_{\epsilon^{-1}}\left(\log\left(\epsilon^{-1}\right)\right)$ such that, for all $k\geq k_\epsilon$,
    \begin{equation}\label{eq:lim_max}
        \mathbb{P}_{\left(\bm{B},\bm{v}\right)\sim\mathbb{P}_{\mathrm{G}}}\left[\left\lvert\frac{1}{m}\max_{\bm{z}\in\mathbb{F}_2^n}\left(m-\left\lVert\bm{B}\bm{z}\oplus\bm{v}\right\rVert_1\right)-\theta^\ast\right\rvert\geq\epsilon\right]\leq\exp\left(-\operatorname{\Omega}\left(n\right)\right).
    \end{equation}
\end{theorem}

\subsection{Concentration}

We first prove a concentration bound that we will use later.
\begin{proposition}[Concentration]\label{prop:concentration}
    Let
    \begin{equation}
        g^\ast\left(\bm{B},\bm{v}\right):=\max_{\bm{z}\in\mathbb{F}_2^n}\left(m-\left\lVert\bm{B}\bm{z}\oplus\bm{v}\right\rVert_1\right).
    \end{equation}
    There exists a universal constant $C>0$ such that:
    \begin{equation}
        \mathbb{P}_{\left(\bm{B},\bm{v}\right)\sim\mathbb{P}_{\mathrm{G}}}\left[\left\lvert g^\ast\left(\bm{B},\bm{v}\right)-\mathbb{E}_{\bm{v}\sim\mathbb{P}_{\mathrm{par}}}\left[g^\ast\left(\bm{B},\bm{v}\right)\right]\right\rvert\geq tm\right]\leq 2\exp\left(-Cmt^2\right).
    \end{equation}
\end{proposition}
\begin{proof}
    This concentration property follows from the fact that Lispchitz functions of a uniformly random point on the hypercube concentrate. In particular, there exists a universal constant $C>0$ such that, for any function $f:\left\{0,1\right\}^m\to\mathbb{R}$ satisfying:
    \begin{equation}\label{eq:l_lipschitz}
        \left\lvert f\left(\bm{v}\right)-f\left(\bm{v'}\right)\right\rvert\leq\frac{L}{m}\left\lVert\bm{v}\oplus\bm{v'}\right\rVert_1
    \end{equation}
    for all $\bm{v}\in\left\{0,1\right\}^m$, we have for $\bm{v}$ drawn uniformly at random that~\cite[Theorem~5.2.5]{vershynin_vectors}:
    \begin{equation}\label{eq:lipschitz_conc}
        \mathbb{P}\left[\left\lvert f\left(\bm{v}\right)-\mathbb{E}\left[f\left(\bm{v}\right)\right]\right\rvert\geq t\right]\leq 2\exp\left(-\frac{Cmt^2}{L^2}\right).
    \end{equation}
    We define:
    \begin{equation}
        g\left(\bm{v};\bm{B},\bm{z}\right):=m-\left\lVert\bm{B}\bm{z}\oplus\bm{v}\right\rVert_1.
    \end{equation}
    This function is obviously $L$-Lipschitz in $\bm{v}$ according to Eq.~\eqref{eq:l_lipschitz} with $L=m$. As the supremum of a set of $L$-Lipschitz functions is $L$-Lipschitz, we have from Eq.~\eqref{eq:lipschitz_conc} that:
    \begin{equation}
        \mathbb{P}_{\left(\bm{B},\bm{v}\right)\sim\mathbb{P}_{\mathrm{G}}}\left[\left\lvert g^\ast\left(\bm{B},\bm{v}\right)-\mathbb{E}_{\bm{v}\sim\mathbb{P}_{\mathrm{par}}}\left[g^\ast\left(\bm{B},\bm{v}\right)\right]\right\rvert\geq t\right]\leq 2\exp\left(-\frac{Ct^2}{m}\right).
    \end{equation}
    Taking $t\mapsto tm$ yields the final result.
\end{proof}

\subsection{Upper Bound}

We first upper bound the (w.h.p.) maximum number of satisfiable clauses, i.e., we first prove Eq.~\eqref{eq:lim_max_ub}. To achieve this, we define the random counting variable for any $0<\theta<1$:
\begin{equation}\label{eq:counting_var_def}
    N_\theta:=\left\lvert\left\{\bm{z}\in\mathbb{F}_2^n:\left\lVert\bm{B}\bm{z}\oplus\bm{v}\right\rVert_1\leq\left(1-\theta\right)m\right\}\right\rvert.
\end{equation}
Our desired upper bound then can be framed as an upper bound on:
\begin{equation}
    \mathbb{P}_{\left(\bm{B},\bm{v}\right)\sim\mathbb{P}_{\mathrm{G}}}\left[N_\theta\geq 1\right]\leq\mathbb{E}_{\bm{v}\sim\mathbb{P}_{\mathrm{par}}}\left[N_\theta\right]=\sum_{\bm{z}\in\mathbb{F}_2^n}\mathbb{P}_{\left(\bm{B},\bm{v}\right)\sim\mathbb{P}_{\mathrm{G}}}\left[\left\lVert\bm{B}\bm{z}\oplus\bm{v}\right\rVert_1\leq\left(1-\theta\right)m\right].
\end{equation}
In what follows we assume $\theta\geq\frac{1}{2}$ for brevity, as $\theta=\frac{1}{2}$ is trivially asymptotically achieved by random guessing.

Note that $\bm{v}$ has i.i.d.\ Bernoulli entries. This probability is therefore then just the probability a binomial random variable with $m$ trials succeeds at least $\theta m$ times. We therefore have from the CDF of the binomial distribution that:
\begin{equation}\label{eq:first_moment}
    \begin{aligned}
        \mathbb{E}_{\left(\bm{B},\bm{v}\right)\sim\mathbb{P}_{\mathrm{G}}}\left[N_\theta\right]&=2^n\operatorname{I}_{1/2}\left(\theta m+1,m-\theta m-1\right)\\
        &=2^n\operatorname{I}_{1/2}\left(\theta\lambda n+1,\left(1-\theta\right)\lambda n-1\right)\\
        &=\exp_2\left(-\left(\lambda-1\right)n+\operatorname{H}_2\left(\theta\right)\lambda n+\operatorname{O}\left(\log\left(n\right)\right)\right),
    \end{aligned}
\end{equation}
where $\operatorname{I}_p$ is the regularized incomplete beta function with parameter $p$ and $\operatorname{H}_2$ is the binary entropy function (in bits). This gives in the $n\to\infty$ limit:
\begin{equation}\label{eq:upper_bound}
    \mathbb{P}_{\left(\bm{B},\bm{v}\right)\sim\mathbb{P}_{\mathrm{G}}}\left[N_\theta\geq 1\right]\leq\exp\left(-\operatorname{\Omega}\left(n\right)\right)
\end{equation}
for any:
\begin{equation}
    \theta>\theta^\ast:=1-\operatorname{H}_2^{-1}\left(1-\frac{1}{\lambda}\right),
\end{equation}
where the inverse of the binary entropy is here defined with the codomain $\left[0,\frac{1}{2}\right]$.

\subsection{Lower Bound}

We now lower bound the (w.h.p.) maximum number of satisfiable clauses at sufficiently large $k$, which will complete the proof of Eq.~\eqref{eq:lim_max} (and therefore Theorem~\ref{thm:max_value}) given Eq.~\eqref{eq:lim_max_ub}. Recalling the definition of $N_\theta$ from Eq.~\eqref{eq:counting_var_def}, we will achieve this by using the Paley--Zygmund inequality:
\begin{equation}\label{eq:gen_pz}
    \mathbb{P}\left[N_\theta\geq 1\right]\geq\frac{\mathbb{E}\left[N_\theta\right]^2}{\mathbb{E}\left[N_\theta^2\right]},
\end{equation}
as well as a boosting trick due to Frieze~\cite{frieze1990171}. Once again, in what follows we assume $\theta\geq\frac{1}{2}$ for brevity, as asymptotically uniformly random $\bm{z}$ achieve $\theta=\frac{1}{2}$ w.h.p. 

We have already calculated the numerator of the right-hand side of Eq.~\eqref{eq:gen_pz} in Eq.~\eqref{eq:first_moment}, leaving only the second moment of $N_\theta$ to bound. We have by Markov's inequality and the law of total probability:
\begin{equation}
    \begin{aligned}
        \mathbb{E}_{\left(\bm{B},\bm{v}\right)\sim\mathbb{P}_{\mathrm{G}}}\left[N_\theta^2\right]=\sum_{\bm{z}\in\mathbb{F}_2^n}\sum_{\bm{z'}\in\mathbb{F}_2^n}&\mathbb{P}_{\left(\bm{B},\bm{v}\right)\sim\mathbb{P}_{\mathrm{G}}}\left[\left\lVert\bm{B}\bm{z}\oplus\bm{v}\right\rVert_1\leq\left(1-\theta\right)m\wedge\left\lVert\bm{B}\bm{z'}\oplus\bm{v}\right\rVert_1\leq\left(1-\theta\right)m\right]\\
        =\sum_{\bm{z}\in\mathbb{F}_2^n}\sum_{\bm{z'}\in\mathbb{F}_2^n}\sum_{\substack{\bm{\tilde{v}}\in\mathbb{F}_2^m\\\text{s.t. }\left\lVert\bm{\tilde{v}}\right\rVert_1\leq\left(1-\theta\right)m}}&\mathbb{P}_{\left(\bm{B},\bm{v}\right)\sim\mathbb{P}_{\mathrm{G}}}\left[\bm{B}\bm{z}\oplus\bm{v}=\bm{\tilde{v}}\right]\\
        &\times\mathbb{P}_{\left(\bm{B},\bm{v}\right)\sim\mathbb{P}_{\mathrm{G}}}\left[\left\lVert\bm{B}\bm{z'}\oplus\bm{v}\right\rVert_1\leq\left(1-\theta\right)m\mid\bm{B}\bm{z}\oplus\bm{v}=\bm{\tilde{v}}\right]\\
        \leq\sum_{\bm{z}\in\mathbb{F}_2^n}\sum_{\bm{z'}\in\mathbb{F}_2^n}\sum_{\substack{\bm{\tilde{v}}\in\mathbb{F}_2^m\\\text{s.t. }\left\lVert\bm{\tilde{v}}\right\rVert_1\leq\left(1-\theta\right)m}}&\mathbb{P}_{\left(\bm{B},\bm{v}\right)\sim\mathbb{P}_{\mathrm{G}}}\left[\bm{B}\bm{z}\oplus\bm{v}=\bm{\tilde{v}}\right]\\
        &\times\mathbb{P}_{\bm{B}\sim\mathbb{P}_{\mathrm{con}}}\left[\left\lVert\bm{B}\left(\bm{z}\oplus\bm{z'}\right)\oplus\bm{\tilde{v}}\right\rVert_1\leq\left(1-\theta\right)m\right]\\
        =2^{-m}\sum_{\bm{z}\in\mathbb{F}_2^n}\sum_{\bm{z'}\in\mathbb{F}_2^n}\sum_{\substack{\bm{\tilde{v}}\in\mathbb{F}_2^m\\\text{s.t. }\left\lVert\bm{\tilde{v}}\right\rVert_1\leq\left(1-\theta\right)m}}&\mathbb{P}_{\bm{B}\sim\mathbb{P}_{\mathrm{con}}}\left[\left\lVert\bm{B}\left(\bm{z}\oplus\bm{z'}\right)\oplus\bm{\tilde{v}}\right\rVert_1\leq \left(1-\theta\right)m\right].
    \end{aligned}
\end{equation}
As the summand depends only on $\bm{z}\oplus\bm{z'}$, we can write this more simply as:
\begin{equation}
    \mathbb{E}_{\left(\bm{B},\bm{v}\right)\sim\mathbb{P}_{\mathrm{G}}}\left[N_\theta^2\right]\leq 2^{-\left(\lambda-1\right)n}\sum_{\bm{y}\in\mathbb{F}_2^n}\sum_{\substack{\bm{\tilde{v}}\in\mathbb{F}_2^m\\\text{s.t. }\left\lVert\bm{\tilde{v}}\right\rVert_1\leq\left(1-\theta\right)m}}\mathbb{P}_{\bm{B}\sim\mathbb{P}_{\mathrm{con}}}\left[\left\lVert\bm{B}\bm{y}\oplus\bm{\tilde{v}}\right\rVert_1\leq \left(1-\theta\right)m\right].
\end{equation}
Finally, we can more conveniently write the sum as an expectation over $\bm{\tilde{v}}$ being drawn from a Bernoulli distribution conditioned on $tm$ of the entries being nonzero, i.e.,
\begin{equation}\label{eq:second_moment_simple}
    \begin{aligned}
        \mathbb{E}_{\left(\bm{B},\bm{v}\right)\sim\mathbb{P}_{\mathrm{G}}}\left[N_\theta^2\right]&\leq 2^{-\left(\lambda-1\right)n+\operatorname{H}_2\left(\theta\right)\lambda n+\operatorname{O}\left(\log\left(n\right)\right)}\sup_{\substack{t\in\left[0,\theta\right]\\\text{s.t. }tm\in\mathbb{Z}}}\sum_{\bm{y}\in\mathbb{F}_2^n}\mathbb{P}_{\left(\bm{B},\bm{\tilde{v}}\right)\sim\mathbb{P}_{\mathrm{con}}\otimes\mathcal{U}_{m,\left(1-t\right)m}}\left[\left\lVert\bm{B}\bm{y}\oplus\bm{\tilde{v}}\right\rVert_1\leq \left(1-\theta\right)m\right]\\
        &=:2^{-\left(\lambda-1\right)n+\operatorname{H}_2\left(\theta\right)\lambda n+\operatorname{O}\left(\log\left(n\right)\right)}\sup_{\substack{t\in\left[0,\theta\right]\\\text{s.t. }tm\in\mathbb{Z}}}M\left(t;\theta\right),
    \end{aligned}
\end{equation}
where $\mathcal{U}_{m,w}$ is the uniform distribution over $\mathbb{F}_2^m$ conditioned on the bit string having Hamming weight $w$.

We now specialize to $\bm{B}^\intercal$ being drawn from the Gallager ensemble. By the definition of the Gallager ensemble, we can write $\bm{B}\bm{y}$ in terms of independent components:
\begin{equation}
    \begin{aligned}
        M\left(t;\theta\right)&=\sum_{\bm{Y}\in\mathbb{F}_2^{k\times\frac{n}{k}}}\mathbb{P}_{\left(\bm{B},\bm{\tilde{v}}\right)\sim\mathbb{P}_{\mathrm{con}}\otimes\mathcal{U}_{m,\left(1-t\right)m}}\left[\left\lVert\bm{\tilde{v}}\oplus\left(\bm{B}\cdot\bigoplus_{i=1}^k\bm{Y}_i^\intercal\right)\right\rVert_1\leq\left(1-\theta\right)m\right]\\
        &=2^n\mathbb{E}_{\bm{Y}\sim\mathcal{U}^{k\otimes\frac{n}{k}}}\left[\mathbb{P}_{\left(\bm{B},\bm{\tilde{v}},\left(\bm{y}_i\right)_i\right)\sim\mathbb{P}_{\mathrm{con}}\otimes\mathcal{U}_{m,\left(1-t\right)m}}\left[\left\lVert\bm{\tilde{v}}\oplus\left(\bm{B}\cdot\bigoplus_{i=1}^k\bm{Y}_i^\intercal\right)\right\rVert_1\leq\left(1-\theta\right)m\right]\right].
    \end{aligned}
\end{equation}
Here, $\bigoplus$ denotes summation modulo $2$, the $\bm{Y}\sim\mathcal{U}^{k\otimes\frac{n}{k}}$ is drawn uniformly at random from $\mathbb{F}_2^{k\times n/k}$, and we use $\bm{Y}_i$ to denote the $i$th row of $\bm{Y}$. By the definition of the Gallager ensemble, each $\bm{B}\bm{\varGamma}_i\cdot\bigoplus_{i=1}^k\bm{Y}_i^\intercal$ is manifestly independent, where $\bm{\varGamma}_i$ are projectors onto $n/k$-dimensional subspaces as in Eq.~\eqref{eq:gamma_proj_def}; if $\bm{Y}_i$ has Hamming weight $\varDelta_i$, $\bm{B}\bm{\varGamma}_i\cdot\bigoplus_{i=1}^k\bm{Y}_i^\intercal$ is distributed uniformly randomly over vectors in $\mathbb{F}_2^m$ with Hamming weight $\lambda k\varDelta_i$. Motivated by this, defining the product distribution:
\begin{equation}\label{eq:tilde_u_set_def}
    \mathcal{\tilde{U}}_{\bm{\varDelta},t}:=\left(\bigotimes_{i=1}^k\mathcal{U}_{m,\lambda k\varDelta_i}\right)\otimes\mathcal{U}_{m,\left(1-t\right)m},
\end{equation}
we can rewrite $M\left(t;\theta\right)$ succinctly as:
\begin{equation}\label{eq:gallager_block_decomp}
    M\left(t;\theta\right)=2^n\mathbb{E}_{\bm{\varDelta}\sim\operatorname{Bin}\left(\frac{n}{k},\frac{1}{2}\right)^{\otimes k}}\left[\mathbb{P}_{\bm{u}\sim\mathcal{\tilde{U}}_{\bm{\varDelta},t}}\left[\left\lVert\bigoplus_{i=1}^{k+1}\bm{u}_i\right\rVert_1\leq\left(1-\theta\right)m\right]\right].
\end{equation}
We now fix $0<\epsilon<1$, which we will set later. We now define the following set of $\bm{\varDelta}$ excluding exceptionally small $\varDelta_i$:
\begin{equation}\label{eq:g_set_def}
    \mathcal{G}_\epsilon:=\left\{\bm{0}\preceq\bm{\varDelta}\preceq\frac{n}{k}\bm{1}:\frac{1}{k}\sum_{i=1}^k\min\left(\frac{k\varDelta_i}{n},1-\frac{k\varDelta_i}{n}\right)\geq\frac{\epsilon}{2}\right\}.
\end{equation}
We have for any $0<\epsilon<1$:
\begin{equation}\label{eq:m_bound}
    M\left(t;\theta\right)\leq 2^n\sup_{\bm{\varDelta}\in\mathcal{G}_\epsilon}\mathbb{P}_{\bm{u}\sim\mathcal{\tilde{U}}_{\bm{\varDelta},t}}\left[\left\lVert\bigoplus_{i=1}^{k+1}\bm{u}_i\right\rVert_1\leq\left(1-\theta\right)m\right]+2^n\mathbb{P}_{\bm{\varDelta}\sim\operatorname{Bin}\left(\frac{n}{k},\frac{1}{2}\right)^{\otimes k}}\left[\bm{\varDelta}\not\in\mathcal{G}_\epsilon\right].
\end{equation}
We can simply bound the second term using the density of the binomial distribution. First, we define an equivalency relation $\sim$ on $\bm{\varDelta}\in\mathcal{G}_\epsilon^\complement$:
\begin{equation}
    \bm{\varDelta}\sim\bm{\varDelta'}\iff\bigwedge_{i=1}^k\left\{\min\left(\frac{k\varDelta_i}{n},1-\frac{k\varDelta_i}{n}\right)=\min\left(\frac{k\varDelta_i'}{n},1-\frac{k\varDelta_i'}{n}\right)\right\}.
\end{equation}
Note the size of each equivalency class is at most $2^k$, and each has unique representative with each $\frac{k\varDelta_i}{n}\leq\frac{1}{2}$. We therefore have that:
\begin{equation}
    \left\lvert\mathcal{G}_\epsilon^\complement\right\rvert\leq 2^k\left\lvert\mathcal{G}_\epsilon^\complement/\sim\right\rvert,
\end{equation}
where $\mathcal{G}_\epsilon/\sim$ is the quotient set:
\begin{equation}
    \mathcal{G}_\epsilon^\complement/\sim=\left\{\bm{0}\preceq\bm{\varDelta}\preceq\frac{n}{2k}\bm{1}:\frac{1}{k}\sum_{i=1}^k\frac{k\varDelta_i}{n}<\frac{\epsilon}{2}\right\}.
\end{equation}
As $\sum_{i=1}^k\varDelta_i$ is binomial distributed with $n$ trials, we have:
\begin{equation}
    \begin{aligned}
        \mathbb{P}_{\bm{\varDelta}\sim\operatorname{Bin}\left(\frac{n}{k},\frac{1}{2}\right)^{\otimes\frac{k}{\lambda}}}\left[\bm{\varDelta}\not\in\mathcal{G}_\epsilon\right]&\leq 2^k\mathbb{P}_{\varDelta\sim\operatorname{Bin}\left(n,\frac{1}{2}\right)}\left[\varDelta<\frac{\epsilon}{2}n\right]\\
        &=2^{-n+\operatorname{H}_2\left(\frac{\epsilon}{2}\right)n+\operatorname{O}\left(\log\left(n\right)\right)}.
    \end{aligned}
\end{equation}

Before continuing with the remaining probability term in Eq.~\eqref{eq:m_bound}, we claim the following property of $\mathcal{G}_\epsilon$ which we will find useful later.
\begin{lemma}[Existence of $\mathcal{I}$]\label{lem:good_set}
    For every
    \begin{equation}
        \bm{\varDelta}\in\mathcal{G}_\epsilon:=\left\{\bm{0}\preceq\bm{\varDelta}\preceq\frac{n}{k}\bm{1}:\frac{1}{k}\sum_{i=1}^k\min\left(\frac{k\varDelta_i}{n},1-\frac{k\varDelta_i}{n}\right)\geq\frac{\epsilon}{2}\right\}
    \end{equation}
    and $c\in\left(0,1\right)$, there exists an index set $\mathcal{I}\subseteq\left[n\right]$ of cardinality $\left\lfloor ck\right\rfloor$ such that:
    \begin{equation}\label{eq:i_index_set_def}
        \forall i\in\mathcal{I},\,\bigwedge_{i\in\mathcal{I}}\left\lvert 1-2\min\left(\frac{k\varDelta_i}{n},1-\frac{k\varDelta_i}{n}\right)\right\rvert\leq\frac{1-\epsilon}{1-c}.
    \end{equation}
\end{lemma}
\begin{proof}
    Note we can rewrite $\mathcal{G}_\epsilon$ as:
    \begin{equation}
        \mathcal{G}_\epsilon=\left\{\bm{0}\preceq\bm{\varDelta}\preceq\frac{n}{k}\bm{1}:\frac{1}{k}\sum_{i=1}^k\left\lvert 1-\frac{2k\varDelta_i}{n}\right\rvert\leq 1-\epsilon\right\}.
    \end{equation}
    For a given $\bm{\varDelta}\in\mathcal{G}_\epsilon$, let $\mathcal{I}$ be the index set associated with the smallest $\left\lfloor ck\right\rfloor$ of the
    \begin{equation}
        x_i:=\left\lvert 1-\frac{2k\varDelta_i}{n}\right\rvert.
    \end{equation}
    Let $x_{i^\ast}$ be the largest $x_i$ of the $i\in\mathcal{I}$. We have:
    \begin{equation}
        \begin{aligned}
            \frac{1}{k}\sum_{i=1}^k x_i&\leq 1-\epsilon\\
            \implies x_{i^\ast}&\leq\left(1-\epsilon\right)k-\sum_{i\neq i^\ast\in\mathcal{I}}x_i-\sum_{i\not\in\mathcal{I}}x_i\\
            &\leq\left(1-\epsilon\right)k-\left\lvert\mathcal{I}^\complement\right\rvert x_{i^\ast}\\
            &\leq\left(1-\epsilon\right)k-\left(1-c\right)kx_{i^\ast}\\
            \implies x_{i^\ast}&\leq\frac{k}{\left(1-c\right)k+1}\left(1-\epsilon\right)\\
            &\leq\frac{1-\epsilon}{1-c}.
        \end{aligned}
    \end{equation}
\end{proof}

We now consider the remaining probability term in Eq.~\eqref{eq:m_bound}. We can equivalently write this term as a probability of the sum (over $\mathbb{F}_2$) of independent vectors being zero (up to a multiplicative factor). To see this, we sample another independent vector of i.i.d.\ Bernoulli random variables $\bm{u}_{k+2}$ conditioned on having Hamming weight at most $\left(1-t'\right)m$ for some $t'\in\left[0,\theta\right]$. Writing:
\begin{equation}
    \mathcal{U}_{\bm{\varDelta},t,t'}:=\mathcal{\tilde{U}}_{\bm{\varDelta},t}\otimes\mathcal{U}_{m,\left(1-t'\right)m},
\end{equation}
we have:
\begin{equation}\label{eq:prob_bound}
    \sup_{\bm{\varDelta}\in\mathcal{G}_\epsilon}\mathbb{P}_{\bm{u}\sim\mathcal{\tilde{U}}_{\bm{\varDelta},t}}\left[\left\lVert\bigoplus_{i=1}^{k+1}\bm{u}_i\right\rVert_1\leq\left(1-\theta\right)m\right]\leq 2^{\operatorname{H}_2\left(\theta\right)\lambda n}\sup_{\bm{\varDelta}\in\mathcal{G}_\epsilon}\sup_{\substack{t'\in\left[0,\theta\right]\\\text{s.t. }tm\in\mathbb{N}}}\mathbb{P}_{\bm{u}\sim\mathcal{U}_{\bm{\varDelta},t,t'}}\left[\bigoplus_{i=1}^{k+2}\bm{u}_i=\bm{0}\right].
\end{equation}
We can now use techniques from \cite{calkin1997} to bound this probability. In what follows, we use the shorthand:
\begin{equation}
    w_i=\left\{\begin{array}{ll}
        \lambda k\varDelta_i, & \text{if }i\leq k;\\
        \left(1-t\right)m, & \text{if }i=k+1;\\
        \left(1-t'\right)m, & \text{if }i=k+2
    \end{array}\right.
\end{equation}
for the Hamming weight of $\bm{u}_i$. Using \cite[Lemma~2.2]{calkin1997}, the probability can be exactly written as:
\begin{equation}
    \mathbb{P}_{\bm{u}\sim\mathcal{U}_{\bm{\varDelta},t,t'}}\left[\bigoplus_{i=1}^{k+2}\bm{u}_i=\bm{0}\right]=2^{-m}\sum_{x=0}^m\binom{m}{x}\prod_{i=1}^{k+2}\lambda_x^{\left(w_i\right)},
\end{equation}
where
\begin{equation}
    \lambda_x^{\left(w_i\right)}=\sum_{\beta=0}^{w_i}\left(-1\right)^\beta\frac{\binom{x}{\beta}\binom{m-x}{w_i-\beta}}{\binom{m}{w_i}}=\binom{m}{w_i}^{-1}K_{w_i}\left(x\right).
\end{equation}
Here, $K_w$ denotes the $w$th binary Kravchuk polynomial. We now fix any $c\in\left(0,1\right)$ (such that $ck$ is an integer) and $\mathcal{I}$ as in Lemma~\ref{lem:good_set}. We can bound via H\"{o}lder's inequality and the fact that all $\left\lvert\lambda_x^{\left(w_i\right)}\right\rvert\leq 1$:
\begin{equation}\label{eq:prob_eq_zero_ind}
    \begin{aligned}
        \mathbb{P}_{\bm{u}\sim\mathcal{U}_{\bm{\varDelta},t,t'}}\left[\bigoplus_{i=1}^{k+2}\bm{u}_i=\bm{0}\right]&=\frac{1}{\prod_{i=1}^{k+2}\binom{m}{w_i}}\mathbb{E}_{x\sim\operatorname{Bin}\left(m,\frac{1}{2}\right)}\left[\prod_{i=1}^{k+2}K_{w_i}\left(x\right)\right]\\
        &\leq\frac{1}{\prod_{i=1}^k\binom{m}{w_i}}\mathbb{E}_{x\sim\operatorname{Bin}\left(m,\frac{1}{2}\right)}\left[\left\lvert\prod_{i=1}^k K_{w_i}\left(x\right)\right\rvert\right]\\
        &\leq\frac{1}{\prod_{i\in\mathcal{I}}\binom{m}{w_i}}\prod_{i\in\mathcal{I}}\left(\mathbb{E}_{x\sim\operatorname{Bin}\left(m,\frac{1}{2}\right)}\left[\left\lvert K_{w_i}\left(x\right)\right\rvert^{ck}\right]\right)^{\frac{1}{ck}}.
    \end{aligned}
\end{equation}
Fortunately, the moments of the $\left\lvert K_{w_i}\left(x\right)\right\rvert$ over $x\sim\operatorname{Bin}\left(m,\frac{1}{2}\right)$ are known~\cite{9398654}. We have the bound~\cite[Corollary~4]{9398654}:
\begin{equation}\label{eq:kravchuk_bound}
    \prod_{i\in\mathcal{I}}\left(\mathbb{E}_{x\sim\operatorname{Bin}\left(m,\frac{1}{2}\right)}\left[\left\lvert K_{w_i}\left(x\right)\right\rvert^{ck}\right]\right)^{\frac{1}{ck}}\leq\sqrt{\prod_{i\in\mathcal{I}}\binom{m}{w_i}}\exp_2\left(\frac{m}{ck}\sum_{i=1}^{ck}\psi\left(ck,\min\left(1-\frac{w_i}{m},\frac{w_i}{m}\right)\right)\right),
\end{equation}
where:
\begin{equation}\label{eq:psi_def_mt}
    \psi\left(p,x\right):=p-1+\log_2\left(\left(1-\delta\right)^p+\delta^p\right)-\frac{p}{2}\operatorname{H}_2\left(x\right)-px\log_2\left(1-2\delta\right),
\end{equation}
with $\delta$ implicitly defined via:
\begin{equation}\label{eq:delta_def}
    x=\left(\frac{1}{2}-\delta\right)\frac{\left(1-\delta\right)^{p-1}-\delta^{p-1}}{\left(1-\delta\right)^p+\delta^p}.
\end{equation}
In Appendix~\ref{sec:bound_on_psi}, Proposition~\ref{prop:kravchuk_psi_bound}, we prove the following general upper bound on $\psi$:
\begin{equation}
    \psi\left(p,x\right)\leq\frac{4}{\ln\left(2\right)}\left(1-2x\right)^{\frac{p-1}{2}}+\frac{p}{2}\operatorname{H}_2\left(x\right)-1.
\end{equation}
This gives:
\begin{equation}
    \begin{aligned}
        \frac{1}{ck}\sum_{i\in\mathcal{I}}\psi\left(k,\min\left(1-\frac{w_i}{m},\frac{w_i}{m}\right)\right)&\leq -1+\frac{1}{2}\sum_{i\in\mathcal{I}}\operatorname{H}_2\left(\frac{w_i}{m}\right)+\frac{4}{\ln\left(2\right)ck}\sum_{i\in\mathcal{I}}\left(1-2\min\left(\frac{w_i}{m},1-\frac{w_i}{m}\right)\right)^{\frac{ck-1}{2}}\\
        &\leq\frac{4}{\ln\left(2\right)}\left(\frac{1-\epsilon}{1-c}\right)^{\frac{ck-1}{2}}-1+\frac{1}{2}\sum_{i\in\mathcal{I}}\operatorname{H}_2\left(\frac{w_i}{m}\right),
    \end{aligned}
\end{equation}
where the final line follows from the definition of $\mathcal{I}$. Using Stirling's approximation then gives:
\begin{equation}\label{eq:mod_u_prob_bound}
    \begin{aligned}
        \mathbb{P}_{\bm{u}\sim\mathcal{U}_{\bm{\varDelta},t,t'}}\left[\bigoplus_{i=1}^{k+2}\bm{u}_i=\bm{0}\right]&\leq\exp_2\left(-\left(1-\frac{4}{\ln\left(2\right)}\left(\frac{1-\epsilon}{1-c}\right)^{\frac{ck-1}{2}}\right)m+\frac{1}{2}\sum_{i\in\mathcal{I}}\operatorname{H}_2\left(\frac{w_i}{m}\right)-\frac{1}{2}\sum_{i\in\mathcal{I}}\operatorname{H}_2\left(\frac{w_i}{m}\right)+\operatorname{O}\left(\log\left(n\right)\right)\right)\\
        &=\exp_2\left(-\left(1-\frac{4}{\ln\left(2\right)}\left(\frac{1-\epsilon}{1-c}\right)^{\frac{ck-1}{2}}\right)m+\operatorname{O}\left(\log\left(n\right)\right)\right).
    \end{aligned}
\end{equation}
Taking $c=\epsilon/2$ then gives:\footnote{Assuming $\epsilon k$ is an even integer, such that the condition that $ck$ is an integer is satisfied; we will neglect repeating this for the rest of the proof for brevity.}
\begin{equation}
    \mathbb{P}_{\bm{u}\sim\mathcal{U}_{\bm{\varDelta},t,t'}}\left[\bigoplus_{i=1}^{k+2}\bm{u}_i=\bm{0}\right]\leq\exp_2\left(-\left(1-\ce^{-\operatorname{\Omega}_k\left(k\right)}\right)m+\operatorname{O}\left(\log\left(n\right)\right)\right)
\end{equation}
for any fixed $\epsilon$.

Putting everything together, for any choice of $0<\epsilon<1$,
\begin{equation}
    M\left(t;\theta\right)\leq 2^{-\left(\left(1-\ce^{-\operatorname{\Omega}_k\left(k\right)}\right)\lambda-1\right)n+\operatorname{H}_2\left(\theta\right)\lambda n+\operatorname{O}\left(\log\left(n\right)\right)}+2^{\operatorname{H}_2\left(\frac{\epsilon}{2}\right)n+\operatorname{O}\left(\log\left(n\right)\right)},
\end{equation}
and therefore:
\begin{equation}
    \mathbb{E}_{\left(\bm{B},\bm{v}\right)\sim\mathbb{P}_{\mathrm{G}}}\left[N_\theta^2\right]\leq 2^{-\left(\lambda-1\right)n+\operatorname{H}_2\left(\theta\right)\lambda n+\operatorname{O}\left(\log\left(n\right)\right)}\left(2^{-\left(\left(1-\ce^{-\operatorname{\Omega}_k\left(k\right)}\right)\lambda-1\right)n+\operatorname{H}_2\left(\theta\right)\lambda n+\operatorname{O}\left(\log\left(n\right)\right)}+2^{\operatorname{H}_2\left(\frac{\epsilon}{2}\right)n+\operatorname{O}\left(\log\left(n\right)\right)}\right).
\end{equation}

We now return to the Paley--Zygmund inequality. Reusing the first moment calculation from Eq.~\eqref{eq:first_moment}, we have as $n\to\infty$:
\begin{equation}
    \begin{aligned}
        \mathbb{P}_{\left(\bm{B},\bm{v}\right)\sim\mathbb{P}_{\mathrm{G}}}\left[N_\theta\geq 1\right]&\geq\frac{\mathbb{E}_{\left(\bm{B},\bm{v}\right)\sim\mathbb{P}_{\mathrm{G}}}\left[N_\theta\right]^2}{\mathbb{E}_{\left(\bm{B},\bm{v}\right)\sim\mathbb{P}_{\mathrm{G}}}\left[N_\theta^2\right]}\\
        &\geq\frac{2^{-\left(\lambda-1\right)n+\operatorname{H}_2\left(\theta\right)\lambda n+\operatorname{O}\left(\log\left(n\right)\right)}}{2^{-\left(\left(1-\ce^{-\operatorname{\Omega}_k\left(k\right)}\right)\lambda-1\right)n+\operatorname{H}_2\left(\theta\right)\lambda n+\operatorname{O}\left(\log\left(n\right)\right)}+2^{\operatorname{H}_2\left(\frac{\epsilon}{2}\right)n+\operatorname{O}\left(\log\left(n\right)\right)}}.
    \end{aligned}
\end{equation}
We now set $\theta=\theta^\ast:=1-\operatorname{H}_2^{-1}\left(1-\lambda^{-1}\right)$. For any constant $0<\epsilon<1$, for sufficiently large $k$ (logarithmic in $1/\epsilon$), this gives:
\begin{equation}\label{eq:prob_n_nonempty_lb}
    \begin{aligned}
        \mathbb{P}_{\left(\bm{B},\bm{v}\right)\sim\mathbb{P}_{\mathrm{G}}}\left[N_\theta\geq 1\right]\geq 2^{-\operatorname{H}_2\left(\frac{\epsilon}{2}\right)n+\operatorname{O}\left(\log\left(n\right)\right)}.
    \end{aligned}
\end{equation}

We now boost this lower bound using a trick by Frieze~\cite{frieze1990171}. Define:
\begin{equation}
    \mu\left(\bm{B},\bm{v}\right):=\max_{\bm{z}\in\mathbb{F}_2^n}\left(m-\left\lVert\bm{B}\bm{z}\oplus\bm{v}\right\rVert_1\right)
\end{equation}
for simplicity in what follows. Fix any $0<\delta<1$. First, we claim that there exists some $k_\delta$ such that, for sufficiently large $n$ and any $k\geq k_\delta$, for any $\bm{B}$,
\begin{equation}\label{eq:exp_lb}
    \mathbb{E}_{\bm{v}\sim\mathbb{P}_{\mathrm{par}}}\left[\mu\left(\bm{B},\bm{v}\right)\right]\geq\left(1-\frac{\delta}{2}\right)\theta^\ast\lambda n.
\end{equation}
To see this, assume (looking toward contradiction) that it is not. We then have by the concentration of the maximal function value (Proposition~\ref{prop:concentration}) that there exists some universal constant $C>0$ such that:
\begin{equation}
    \begin{aligned}
        \mathbb{P}_{\left(\bm{B},\bm{v}\right)\sim\mathbb{P}_{\mathrm{G}}}\left[\mu\left(\bm{B},\bm{v}\right)\geq\theta^\ast\lambda n\right]&\leq\mathbb{P}_{\left(\bm{B},\bm{v}\right)\sim\mathbb{P}_{\mathrm{G}}}\left[\mu\left(\bm{B},\bm{v}\right)-\mathbb{E}_{\bm{v}\sim\mathbb{P}_{\mathrm{par}}}\left[\mu\left(\bm{B},\bm{v}\right)\right]\geq\frac{\delta}{2}\theta^\ast\lambda n\right]\\
        &\leq 2\exp\left(-4C\delta^2\theta^{\ast 2}\lambda n\right)\\
        &\leq 2\exp\left(-C\delta^2\lambda n\right).
    \end{aligned}
\end{equation}
However, we showed in Eq.~\eqref{eq:prob_n_nonempty_lb} that there exists some $k_\delta$ depending only on $C$ and $\delta$ (and logarithmic in $1/\delta)$ such that, for all $k\geq k_\delta$,
\begin{equation}
    \mathbb{P}_{\left(\bm{B},\bm{v}\right)\sim\mathbb{P}_{\mathrm{G}}}\left[\mu\left(\bm{B},\bm{v}\right)\geq\theta^\ast\lambda n\right]\geq\exp\left(-0.1C\delta^2 n\right)\geq\exp\left(-0.1C\delta^2\lambda n\right);
\end{equation}
this follows by taking $\epsilon$ sufficiently small in Eq.~\eqref{eq:prob_n_nonempty_lb}. In particular, for sufficiently large $n$, we have a contradiction. Therefore, Eq.~\eqref{eq:exp_lb} holds true for any choice of $0<\delta<1$ and $\bm{B}$ as $n\to\infty$ for sufficiently large $k$. Once again using the concentration of the maximal function value (Proposition~\ref{prop:concentration}), we therefore have that for any $0<\delta<1$ there exists a $k_\delta=\operatorname{O}_{\delta^{-1}}\left(\log\left(\delta^{-1}\right)\right)$ such that for any $k\geq k_\delta$ and sufficiently large $n$:
\begin{equation}
    \mathbb{P}_{\left(\bm{B},\bm{v}\right)\sim\mathbb{P}_{\mathrm{G}}}\left[\mu\left(\bm{B},\bm{v}\right)\geq\left(1-\delta\right)\theta^\ast\lambda n\right]\leq\exp\left(-\operatorname{\Omega}\left(n\right)\right).
\end{equation}
This taken in combination with Eq.~\eqref{eq:upper_bound} proves Theorem~\ref{thm:max_value}.

\section{Proof of Theorems~\ref{thm:chaos_prop} and a \texorpdfstring{$2$}{2}-OGP}\label{sec:proof_of_chaos_prop}

In this Appendix we prove that transposed Gallager \texorpdfstring{\textsc{MAX-$k$-XOR-SAT}}{MAX-k-XOR-SAT} exhibits both the chaos property (Definition~\ref{def:chaos_prop}) and the $R$-OGP (Definition~\ref{def:multi_ogp}), computing the threshold at all $R$ for the former and the threshold for the latter in the special case $R=2$. We conjecture that the true $R$-OGP threshold (i.e., the threshold when $R$ is chosen optimally) is much smaller; this is motivated by the fact that both of the stable algorithms beyond DQI considered in Sec.~\ref{sec:numerics} achieve only a satisfied fraction $1/2$ in the $k\to\infty$ limit, whereas what we compute for the $2$-OGP does not satisfy this property. We leave the computation of the optimal $R$-OGP threshold for future work.

The former result is presented as Theorem~\ref{thm:chaos_prop} of the main text, which we restate here for convenience.
\tgxorsatchaosprop*
Theorem~\ref{thm:chaos_prop} will be proven in the course of proving the existence of an $R$-OGP for the model.
\begin{theorem}[Transposed Gallager \textsc{MAX-$k$-XOR-SAT} exhibits a $2$-OGP]\label{thm:e_ogp}
    Fix any $0\leq\nu_1<\nu_2<1/2$ and $0\leq\sigma\leq 1$. For any $\kappa\in\left[0,1\right]$, fix $\mathbb{P}_2^{\left(\kappa\right)}$ to be a $\kappa$-correlated ensemble (Definition~\ref{def:kappa_corr_ens}). For any function $S:\mathbb{F}_2^{m\times n}\to\binom{\left[n\right]}{\left\lfloor\sigma n\right\rfloor}$ and
    \begin{equation}\label{eq:mu_bound_ogp}
        \mu>\mu^\ast:=1-\operatorname{H}_2^{-1}\left(1-\frac{1}{2\lambda}-\frac{1}{2\lambda}\left(\operatorname{H}_2\left(\nu_2\right)+\sigma\right)-\frac{2\ce}{\ln\left(2\right)}\exp\left(-\nu_1 k\right)\right),
    \end{equation}
    it is the case that:
    \begin{equation}\label{eq:ogp_prob_bound}
        \mathbb{P}_{\left(\bm{B},\bm{v}^{\left(r\right)}\right)_{r=1}^2\sim\mathbb{P}_2^{\left(\kappa\right)}}\left[\mathcal{S}_{\left(\bm{B},\bm{v}^{\left(r\right)}\right)_{r=1}^2}^{\left(\mu,\nu_1,\nu_2,S_{\bm{B}}\right)}\neq\varnothing\right]\leq\exp\left(-\operatorname{\Omega}\left(n\right)\right),
    \end{equation}
    where $\mathcal{S}_{\left(\bm{B},\bm{v}^{\left(r\right)}\right)_{r=1}^2}^{\left(\mu,\nu_1,\nu_2,S_{\bm{B}}\right)}$ is the random set of tuples $\left(\bm{z}^{\left(r\right)}\right)_{r=1}^2\in\left\{-1,1\right\}^{2\times n}$ satisfying:
    \begin{enumerate}
        \item \textbf{$\mu$-SAT fraction:} For all $r\in\left[2\right]$,
        \begin{equation}
            g_{\bm{X}^{\left(r\right)}}\left(\bm{z}^{\left(r\right)}\right)\geq\mu\lambda n.
        \end{equation}
        \item \textbf{$\left(k,S\right)$-minimum Hamming semimetric bound:} Recalling the $\left(k,S\right)$-minimum Hamming semimetric (Definition~\ref{def:k_min_ham_semi}),
        \begin{equation}\label{eq:dist_const}
            d_{k,S_{\bm{B}}}\left(\bm{z}^{\left(1\right)},\bm{z}^{\left(2\right)}\right)\in\left[\nu_1 n,\nu_2 n\right].
        \end{equation}
    \end{enumerate}
    Finally, when the $\bm{v}^{\left(r\right)}$ are drawn independently ($\kappa=1$),
    \begin{equation}\label{eq:chaos_property}
        \mathbb{P}_{\left(\bm{B},\bm{v}^{\left(r\right)}\right)_{r=1}^2\sim\mathbb{P}_2^{\left(1\right)}}\left[\mathcal{S}_{\left(\bm{B},\bm{v}^{\left(r\right)}\right)_{r=1}^2}^{\left(\mu,0,\nu_2,S_{\bm{B}}\right)}\neq\varnothing\right]\leq\exp\left(-\operatorname{\Omega}\left(n\right)\right)
    \end{equation}
    with the same choice of $\mu,\nu_2,S_{\bm{B}}$.
\end{theorem}
\begin{proof}
    Consider the random set as defined in the theorem statement:
    \begin{equation}
        \mathcal{S}_{\left(\bm{B},\bm{v}^{\left(r\right)}\right)_{r=1}^2}^{\left(\mu,\nu_1,\nu_2,S_{\bm{B}}\right)}:=\left\{\left(\bm{z}^{\left(r\right)}\right)_{r=1}^2\in\mathbb{F}_2^{2\times n}:\bigwedge_{r=1}^2\left\lVert\bm{B}\bm{z}^{\left(r\right)}\oplus\bm{v}^{\left(r\right)}\right\rVert_1\leq\left(1-\mu\right)m\wedge\frac{1}{n}d_{k,S_{\bm{B}}}\left(\bm{z}^{\left(1\right)},\bm{z}^{\left(2\right)}\right)\in\left[\nu_1,\nu_2\right]\right\}.
    \end{equation}
    We define the counting variable that is the cardinality of this set:
    \begin{equation}
        N_{\mu,\nu_1,\nu_2,S_{\bm{B}}}:=\left\lvert\mathcal{S}_{\left(\bm{B},\bm{v}^{\left(r\right)}\right)_{r=1}^2}^{\left(\mu,\nu_1,\nu_2,S_{\bm{B}}\right)}\right\rvert.
    \end{equation}
    For simplicity of notation in what follows, we let $\mathcal{Z}\subset\mathbb{F}_2^{2\times n}$ denote the set of $\bm{Z}=\left(\bm{z}^{\left(r\right)}\right)_{r=1}^2$ such that they satisfy Eq.~\eqref{eq:dist_const}.
    
    We consider the two discussed cases separately: when the $\bm{v}^{\left(r\right)}$ are drawn independently and when they are not. We begin with the latter case. The former case will also prove Theorem~\ref{thm:chaos_prop}.
    
    \subsection{Dependent \texorpdfstring{$\bm{v}^{\left(r\right)}$}{v\^r}}
    
    We first define the set of low-weight vectors:
    \begin{equation}
        \mathcal{X}=\left\{\left(\bm{x}^{\left(r\right)}\right)_{r=1}^2\in\mathbb{F}_2^{2\times n}:\bigwedge_{r=1}^2\left\lVert\bm{x}^{\left(r\right)}\right\rVert_1\leq\left(1-\mu\right)m\right\}.
    \end{equation}
    We have by Markov's inequality and the law of total probability that:
    \begin{equation}\label{eq:sec_mom}
        \begin{aligned}
            \mathbb{P}\left[N_{\mu,\nu_1,\nu_2,S_{\bm{B}}}\geq 1\right]&\leq\sum_{\bm{Z}\in\mathcal{Z}}\mathbb{P}\left[\bigwedge_{r=1}^2\left\lVert\bm{B}\bm{z}^{\left(r\right)}\oplus\bm{v}^{\left(r\right)}\right\rVert_1\leq\left(1-\mu\right)m\right]\\
            &=\sum_{\bm{Z}\in\mathcal{Z}}\sum_{\bm{X}\in\mathcal{X}}\mathbb{P}\left[\bigwedge_{r=1}^2\bm{B}\bm{z}^{\left(r\right)}\oplus\bm{v}^{\left(r\right)}=\bm{x}^{\left(r\right)}\right]\\
            &=\sum_{\bm{Z}\in\mathcal{Z}}\sum_{\bm{X}\in\mathcal{X}}\sum_{\bm{y}\in\mathbb{F}_2^m}\mathbb{P}\left[\left.\bigwedge_{r=1}^2\bm{B}\bm{z}^{\left(r\right)}\oplus\bm{v}^{\left(r\right)}=\bm{x}^{\left(r\right)}\right|\bm{v}^{\left(1\right)}=\bm{y}\right]\mathbb{P}\left[\bm{v}^{\left(1\right)}=\bm{y}\right]\\
            &=2^{-m}\sum_{\bm{Z}\in\mathcal{Z}}\sum_{\bm{X}\in\mathcal{X}}\sum_{\bm{y}\in\mathbb{F}_2^m}\mathbb{P}\left[\left.\bigwedge_{r=1}^2\bm{B}\bm{z}^{\left(r\right)}\oplus\bm{v}^{\left(r\right)}=\bm{x}^{\left(r\right)}\right|\bm{v}^{\left(1\right)}=\bm{y}\right].
        \end{aligned}
    \end{equation}
    By the assumed dependence structure, there exists some projector $\bm{\varUpsilon}$ such that:
    \begin{equation}
        \bm{v}^{\left(2\right)}=\bm{v}^{\left(1\right)}\oplus\bm{\varUpsilon}\left(\bm{v}^{\left(1\right)}\oplus\bm{\tilde{v}}^{\left(2\right)}\right),
    \end{equation}
    where $\bm{\tilde{v}}^{\left(2\right)}$ is independent from $\bm{v}^{\left(1\right)}$. We therefore have:
    \begin{equation}
        \begin{aligned}
            \mathbb{P}\left[N_{\mu,\nu_1,\nu_2,S_{\bm{B}}}\geq 1\right]&\leq 2^{-m}\sum_{\bm{Z}\in\mathcal{Z}}\sum_{\bm{X}\in\mathcal{X}}\sum_{\bm{y}\in\mathbb{F}_2^m}\mathbb{P}\left[\left.\bigwedge_{r=1}^2\bm{B}\bm{z}^{\left(r\right)}\oplus\bm{v}^{\left(r\right)}=\bm{x}^{\left(r\right)}\right|\bm{v}^{\left(1\right)}=\bm{y}\right]\\
            &=2^{-m}\sum_{\bm{Z}\in\mathcal{Z}}\sum_{\bm{X}\in\mathcal{X}}\sum_{\bm{y}\in\mathbb{F}_2^m}\mathbb{P}\left[\left\{\bm{B}\bm{z}^{\left(1\right)}\oplus\bm{x}^{\left(1\right)}=\bm{y}\right\}\wedge\bm{B}\bm{z}^{\left(2\right)}\oplus\bm{y}\oplus\bm{\varUpsilon}\left(\bm{y}\oplus\bm{\tilde{v}}^{\left(2\right)}\right)=\bm{x}^{\left(2\right)}\right]\\
            &=2^{-m}\sum_{\bm{Z}\in\mathcal{Z}}\sum_{\bm{X}\in\mathcal{X}}\sum_{\bm{y}\in\mathbb{F}_2^m}\mathbb{P}\left[\left\{\bm{B}\bm{z}^{\left(1\right)}\oplus\bm{x}^{\left(1\right)}=\bm{y}\right\}\wedge\bm{B}\left(\bm{z}^{\left(1\right)}\oplus\bm{z}^{\left(2\right)}\right)\oplus\bm{\varUpsilon}\left(\bm{y}\oplus\bm{\tilde{v}}^{\left(2\right)}\right)=\bm{x}^{\left(1\right)}\oplus\bm{x}^{\left(2\right)}\right]\\
            &=2^{-m}\sum_{\bm{Z}\in\mathcal{Z}}\sum_{\bm{X}\in\mathcal{X}}\sum_{\bm{y}\in\mathbb{F}_2^m}\mathbb{P}\left[\left\{\bm{B}\bm{z}^{\left(1\right)}\oplus\bm{x}^{\left(1\right)}=\bm{y}\right\}\wedge\bm{B}\left(\bm{z}^{\left(1\right)}\oplus\bm{z}^{\left(2\right)}\right)\oplus\bm{\varUpsilon}\bm{\tilde{v}}^{\left(2\right)}=\bm{x}^{\left(1\right)}\oplus\bm{x}^{\left(2\right)}\right]\\
            &=2^{-m}\sum_{\bm{Z}\in\mathcal{Z}}\sum_{\bm{X}\in\mathcal{X}}\mathbb{P}\left[\bm{B}\left(\bm{z}^{\left(1\right)}\oplus\bm{z}^{\left(2\right)}\right)=\bm{x}^{\left(1\right)}\oplus\bm{x}^{\left(2\right)}\oplus\bm{\varUpsilon}\bm{\tilde{v}}^{\left(2\right)}\right],
        \end{aligned}
    \end{equation}
    with the penultimate line following as $\bm{\tilde{v}}^{\left(2\right)}$ and $\bm{y}\oplus\bm{\tilde{v}}^{\left(2\right)}$ are distributed identically, and the final line following from the law of total probability.
    
    We now focus on the probability term. Recall from the definition of the Gallager ensemble that:
    \begin{equation}
        \bm{B}\left(\bm{z}^{\left(1\right)}\oplus\bm{z}^{\left(2\right)}\right)=\bigoplus_{i=1}^k\bm{B}\bm{\varGamma}_i\left(\bm{z}^{\left(1\right)}\oplus\bm{z}^{\left(2\right)}\right),
    \end{equation}
    where the $\bm{B}\bm{\varGamma}_i\left(\bm{z}^{\left(1\right)}\oplus\bm{z}^{\left(2\right)}\right)$ are distributed as independent, uniformly random vectors with Hamming weight constrained to be $k\lambda\left\lVert\bm{z}^{\left(1\right)}\oplus\bm{z}^{\left(2\right)}\right\rVert_1$. We have already calculated a bound on this probability in Appendix~\ref{sec:max_value_gallager_xor_sat}, in Eq.~\eqref{eq:prob_eq_zero_ind}. Writing:
    \begin{equation}
        \bm{u}_i:=\bm{B}\bm{\varGamma}_i\left(\bm{z}^{\left(1\right)}\oplus\bm{z}^{\left(2\right)}\right)
    \end{equation}
    for $1\leq i\leq k$ and:
    \begin{equation}
        \bm{u}_{k+1}:=\bm{\pi}\left(\bm{x}^{\left(1\right)}\oplus\bm{x}^{\left(2\right)}\oplus\bm{\varUpsilon}\bm{\tilde{v}}^{\left(2\right)}\right)
    \end{equation}
    for $\bm{\pi}$ a uniformly random permutation, and defining the shorthand notation:
    \begin{equation}
        w_i=k\lambda\left\lVert\bm{\varGamma}_i\left(\bm{z}^{\left(1\right)}\oplus\bm{z}^{\left(2\right)}\right)\right\rVert_1,
    \end{equation}
    we have (Eq.~\eqref{eq:prob_eq_zero_ind}):
    \begin{equation}
        \mathbb{P}_{\bm{u}\sim\mathcal{U}_{\bm{\varDelta},t,t'}}\left[\bigoplus_{i=1}^{k+1}\bm{u}_i=\bm{0}\right]\leq\frac{1}{\prod_{i=1}^k\binom{m}{w_i}}\mathbb{E}_{x\sim\operatorname{Bin}\left(m,\frac{1}{2}\right)}\left[\left\lvert\prod_{i=1}^k K_{w_i}\left(x\right)\right\rvert\right],
    \end{equation}
    where $K_w$ denotes the $w$th binary Kravchuk polynomial. From here, we will proceed slightly differently from the proof of Theorem~\ref{thm:max_value}. Defining:
    \begin{align}
        x_i&:=\min\left(\frac{w_i}{n},1-\frac{w_i}{n}\right),\\
        p_i&:=\frac{\sum_{j=1}^k\min\left(\frac{w_j}{n},1-\frac{w_j}{n}\right)}{\min\left(\frac{w_i}{n},1-\frac{w_i}{n}\right)},
    \end{align}
    we have via H\"{o}lder's inequality and the Kravchuk moment bound~\cite[Corollary~4]{9398654}:
    \begin{equation}
        \begin{aligned}
            \mathbb{P}_{\bm{u}\sim\mathcal{U}_{\bm{\varDelta},t,t'}}\left[\bigoplus_{i=1}^{k+1}\bm{u}_i=\bm{0}\right]&\leq\frac{1}{\prod_{i=1}^k\binom{m}{w_i}}\mathbb{E}_{x\sim\operatorname{Bin}\left(m,\frac{1}{2}\right)}\prod_{i=1}^k\left[\left\lvert K_{w_i}\left(x\right)\right\rvert^{p_i}\right]^{\frac{1}{p_i}}\\
            &\leq\frac{1}{\sqrt{\prod_{i=1}^k\binom{m}{w_i}}}\exp_2\left(m\sum_{i=1}^k p_i^{-1}\psi\left(p_i,x_i\right)\right),
        \end{aligned}
    \end{equation}
    where $\psi$ is as in Eq.~\eqref{eq:psi_def_mt}. Now, recall that $\psi$ has upper bound (Proposition~\ref{prop:kravchuk_psi_bound}):
    \begin{equation}
        \psi\left(p,x\right)\leq\frac{4}{\ln\left(2\right)}\left(1-2x\right)^{\frac{p-1}{2}}+\frac{p}{2}\operatorname{H}_2\left(x\right)-1,
    \end{equation}
    giving:
    \begin{equation}
        \begin{aligned}
            \mathbb{P}_{\bm{u}\sim\mathcal{U}_{\bm{\varDelta},t,t'}}\left[\bigoplus_{i=1}^{k+1}\bm{u}_i=\bm{0}\right]&\leq\frac{1}{\sqrt{\prod_{i=1}^k\binom{m}{w_i}}}\exp_2\left(m\sum_{i=1}^k p_i^{-1}\left(\frac{4}{\ln\left(2\right)}\left(1-2x_i\right)^{\frac{p_i-1}{2}}+\frac{p_i}{2}\operatorname{H}_2\left(\frac{w_i}{m}\right)-1\right)\right)\\
            &\leq\exp_2\left(-\left(1-\frac{4}{\ln\left(2\right)}\sum_{i=1}^k p_i^{-1}\left(1-2x_i\right)^{\frac{p_i-1}{2}}\right)m+\operatorname{O}\left(\log\left(n\right)\right)\right).
        \end{aligned}
    \end{equation}
    Using the standard upper bound $1-x\leq\exp\left(-x\right)$ for $x\geq 0$ gives:
    \begin{equation}
        \begin{aligned}
            \mathbb{P}_{\bm{u}\sim\mathcal{U}_{\bm{\varDelta},t,t'}}\left[\bigoplus_{i=1}^{k+1}\bm{u}_i=\bm{0}\right]&\leq\exp_2\left(-\left(1-\frac{4}{\ln\left(2\right)}\sum_{i=1}^k p_i^{-1}\exp\left(-\left(p_i-1\right)x_i\right)\right)m+\operatorname{O}\left(\log\left(n\right)\right)\right)\\
            &\leq\exp_2\left(-\left(1-\frac{4\ce}{\ln\left(2\right)}\sum_{i=1}^k p_i^{-1}\exp\left(-p_i x_i\right)\right)m+\operatorname{O}\left(\log\left(n\right)\right)\right).
        \end{aligned}
    \end{equation}
    Finally, substituting the definition of $p_i$ yields:
    \begin{equation}
        \mathbb{P}_{\bm{u}\sim\mathcal{U}_{\bm{\varDelta},t,t'}}\left[\bigoplus_{i=1}^{k+1}\bm{u}_i=\bm{0}\right]\leq\exp_2\left(-\left(1-\frac{4\ce}{\ln\left(2\right)\sum_{j=1}^k x_j}\sum_{i=1}^k x_i\exp\left(-\sum_{j=1}^k x_j\right)\right)m+\operatorname{O}\left(\log\left(n\right)\right)\right).
    \end{equation}
    We now use the $\nu_1$ bound in the definition of $\mathcal{Z}$, which implies:
    \begin{equation}
        \begin{aligned}
            \sum_{j=1}^k x_j&=\frac{k\lambda}{n}\sum_{j=1}^k\min\left(\left\lVert\bm{\varGamma}_i\left(\bm{z}^{\left(1\right)}\oplus\bm{z}^{\left(2\right)}\right)\right\rVert_1,\left\lVert\bm{\varGamma}_i\left(\bm{z}^{\left(1\right)}\oplus\bm{z}^{\left(2\right)}\oplus\bm{1}\right)\right\rVert_1\right)\\
            &\geq\frac{k\lambda}{n}\sum_{j=1}^k\min\left(\left\lVert\bm{\varSigma}_{S_{\bm{B}}}\bm{\varGamma}_i\left(\bm{z}^{\left(1\right)}\oplus\bm{z}^{\left(2\right)}\right)\right\rVert_1,\left\lVert\bm{\varSigma}_{S_{\bm{B}}}\bm{\varGamma}_i\left(\bm{z}^{\left(1\right)}\oplus\bm{z}^{\left(2\right)}\oplus\bm{1}\right)\right\rVert_1\right)\\
            &\geq\nu_1 k.
        \end{aligned}
    \end{equation}
    This gives:
    \begin{equation}
        \begin{aligned}
            \mathbb{P}_{\bm{u}\sim\mathcal{U}_{\bm{\varDelta},t,t'}}\left[\bigoplus_{i=1}^{k+1}\bm{u}_i=\bm{0}\right]&\leq\exp_2\left(-\left(1-\frac{4\ce}{\ln\left(2\right)\sum_{j=1}^k x_j}\sum_{i=1}^k x_i\exp\left(-\nu_1 k\right)\right)m+\operatorname{O}\left(\log\left(n\right)\right)\right)\\
            &=\exp_2\left(-\left(1-\frac{4\ce}{\ln\left(2\right)}\exp\left(-\nu_1 k\right)\right)m+\operatorname{O}\left(\log\left(n\right)\right)\right).
        \end{aligned}
    \end{equation}
    
    Putting everything together,
    \begin{equation}
        \mathbb{P}\left[N_{\mu,\nu_1,\nu_2,S_{\bm{B}}}\geq 1\right]\leq 2^{-m}\sum_{\bm{Z}\in\mathcal{Z}}\sum_{\bm{X}\in\mathcal{X}}\exp_2\left(-\left(1-\frac{4\ce}{\ln\left(2\right)}\exp\left(-\nu_1 k\right)\right)m+\operatorname{O}\left(\log\left(n\right)\right)\right).
    \end{equation}
    Counting:
    \begin{align}
        \left\lvert\mathcal{Z}\right\rvert&\leq 2^n\times 2^{\operatorname{O}\left(\log\left(n\right)\right)}\binom{n}{\nu_2 n}2^{\left\lvert S_{\bm{B}}\right\rvert}=\exp_2\left(n+\operatorname{H}_2\left(\nu_2\right)n+\sigma n+\operatorname{O}\left(\log\left(n\right)\right)\right),\\
        \left\lvert\mathcal{X}\right\rvert&\leq \left(2^{\operatorname{O}\left(\log\left(n\right)\right)}\binom{m}{\left(1-\mu\right)m}\right)^2=\exp_2\left(2\operatorname{H}_2\left(\mu\right)m+\operatorname{O}\left(\log\left(n\right)\right)\right),
    \end{align}
    we finally have:
    \begin{equation}
        \begin{aligned}
            \mathbb{P}\left[N_{\mu,\nu_1,\nu_2,S_{\bm{B}}}\geq 1\right]&\leq\exp_2\left(n+\operatorname{H}_2\left(\nu_2\right)n+\sigma n+2\operatorname{H}_2\left(\mu\right)m-m-\left(1-\frac{4\ce}{\ln\left(2\right)}\exp\left(-\nu_1 k\right)\right)m+\operatorname{O}\left(\log\left(n\right)\right)\right)\\
            &=:\exp_2\left(\varPsi\left(\mu,\nu_1,\nu_2,\sigma\right)n+\operatorname{O}\left(\log\left(n\right)\right)\right),
        \end{aligned}
    \end{equation}
    where in the final line we defined the function:
    \begin{equation}\label{eq:varpsi_def}
        \varPsi\left(\mu,\nu_1,\nu_2,\sigma\right):=1+\operatorname{H}_2\left(\nu_2\right)+\sigma+2\lambda\operatorname{H}_2\left(\mu\right)-2\lambda+\frac{4\ce\lambda}{\ln\left(2\right)}\exp\left(-\nu_1 k\right).
    \end{equation}
    Demonstrating Eq.~\eqref{eq:ogp_prob_bound} for any given $\mu>\mu^\ast$ then boils down to choosing $\nu_1$, $\nu_2$, $\sigma$, and $R$ such that:
    \begin{equation}
        \varPsi\left(\mu,\nu_1,\nu_2,\sigma\right)<0,
    \end{equation}
    that is,
    \begin{equation}
        \operatorname{H}_2\left(\mu\right)<1-\frac{1}{2\lambda}-\frac{1}{2\lambda}\left(\operatorname{H}_2\left(\nu_2\right)+\sigma\right)-\frac{2\ce}{\ln\left(2\right)}\exp\left(-\nu_1 k\right).
    \end{equation}
    This proves Eq.~\eqref{eq:mu_bound_ogp}.
    
    \subsection{Independent \texorpdfstring{$\bm{v}^{\left(r\right)}$}{v\^r}}
    
    We now briefly discuss the case when the $\bm{v}^{\left(r\right)}$ are drawn independently. We generalize to the general case $R\geq 2$ to simultaneously prove Theorem~\ref{thm:chaos_prop}. By Markov's inequality and the properties of the Bernoulli distribution:
    \begin{equation}
        \begin{aligned}
            \mathbb{P}\left[N_{\mu,\nu_1,\nu_2,S_{\bm{B}}}\geq 1\right]&\leq\sum_{\bm{Z}\in\mathcal{Z}}\mathbb{P}\left[\bigwedge_{r=1}^R\left\lVert\bm{B}\bm{z}^{\left(r\right)}\oplus\bm{v}^{\left(r\right)}\right\rVert_1\leq\left(1-\mu\right)m\right]\\
            &=\sum_{\bm{Z}\in\mathcal{Z}}2^{-R\left(1-\operatorname{H}_2\left(\mu\right)\right)m+\operatorname{O}\left(\log\left(n\right)\right)}\\
            &\leq\exp_2\left(n+\left(R-1\right)\left(\operatorname{H}_2\left(\nu_2\right)+\sigma\right)n+R\operatorname{H}_2\left(\mu\right)m-Rm+\operatorname{O}\left(\log\left(n\right)\right)\right)\\
            &=:\exp_2\left(\tilde{\varPsi}\left(\mu,\nu_2,\sigma\right)n+\operatorname{O}\left(\log\left(n\right)\right)\right),
        \end{aligned}
    \end{equation}
    where in the final line we defined the function:
    \begin{equation}
        \tilde{\varPsi}\left(\mu,\nu_2,\sigma;R\right):=1+\left(R-1\right)\left(\operatorname{H}_2\left(\nu_2\right)+\sigma\right)+R\lambda\operatorname{H}_2\left(\mu\right)-R\lambda.
    \end{equation}
    Taking the condition:
    \begin{equation}
        \tilde{\varPsi}\left(\mu,\nu_2,\sigma;R\right)<0
    \end{equation}
    proves Theorem~\ref{thm:chaos_prop}. Furthermore, $\tilde{\varPsi}\left(\mu,\nu_2,\sigma;2\right)\leq\varPsi\left(\mu,\nu_1,\nu_2,\sigma\right)$, where $\varPsi$ is defined in Eq.~\eqref{eq:varpsi_def}. In particular, for any choice of $\mu,\nu_1,\nu_2,\sigma$ where $\varPsi\left(\mu,\nu_1,\nu_2,\sigma\right)<0$, it is the case that:
    \begin{equation}
        \tilde{\varPsi}\left(\mu,\nu_2,\sigma;2\right)<0.
    \end{equation}
    This completes the proof of Theorem~\ref{thm:e_ogp}.
\end{proof}

\section{QAOA for \texorpdfstring{\textsc{MAX-$k$-XOR-SAT}}{MAX-k-XOR-SAT}}\label{sec:qaoa_app_rat}

We here prove the claim made in the main text that depth-$1$ QAOA~\cite{farhi2014quantumapproximateoptimizationalgorithm} achieves an expected satisfied fraction of:
\begin{equation}
    \frac{\mathbb{E}_{\bm{v}\sim\mathbb{P}_{\mathrm{par}}}\left[\left\langle g\right\rangle_{\mathrm{QAOA}}\right]}{m}\geq\frac{1}{2}+\left(1+\operatorname{o}_k\left(1\right)\right)\sqrt{\frac{\ln\left(k\right)}{4\ce k\lambda}}.
\end{equation}
We begin with~\cite[Theorem~1]{marwaha2022boundsapproximating}, which shows that depth-$1$ QAOA achieves an expected satisfied fraction of:
\begin{equation}
    \frac{\mathbb{E}_{\bm{v}\sim\mathbb{P}_{\mathrm{par}}}\left[\left\langle g\right\rangle_{\mathrm{QAOA}}\right]}{m}=\frac{1}{2}-\frac{\ci s}{4}\left(\left(p+\ci q c^{k\lambda}\right)^k-\left(p-\ci q c^{k\lambda}\right)^k\right)
\end{equation}
for \textsc{MAX-$k$-XOR-SAT} with $\bm{v}\sim\mathbb{P}_{\mathrm{par}}$, where:
\begin{align}
    s&=\sin\left(\gamma\right),\\
    c&=\cos\left(\gamma\right),\\
    p&=\cos\left(2\beta\right),\\
    q&=\sin\left(2\beta\right),
\end{align}
for any $\gamma,\beta\in\left[0,2\cpi\right]$ which can be optimized over. Choosing:
\begin{align}
    \gamma&=\sqrt{\frac{\ln\left(\frac{4k}{\cpi^2}\right)}{k\lambda}},\\
    \beta&=\frac{1}{2\sqrt{k}},
\end{align}
we have in the large-$k$ limit that:
\begin{equation}
    \begin{aligned}
        \left(p\pm\ci qc^{k\lambda}\right)^k&=\left(1-\frac{1}{2k}+\operatorname{O}\left(k^{-2}\right)\pm\left(\frac{\ci}{\sqrt{k}}+\operatorname{O}_k\left(k^{-\frac{3}{2}}\right)\right)\left(1-\frac{\ln\left(\frac{4k}{\cpi^2}\right)}{2k\lambda}+\operatorname{\tilde{O}}_k\left(k^{-2}\right)\right)^{k\lambda}\right)^k\\
        &=\left(1-\frac{1}{2k}+\operatorname{O}_k\left(k^{-2}\right)\pm\left(1+\operatorname{\tilde{O}}_k\left(k^{-1}\right)\right)\frac{\ci}{\sqrt{k}}\exp\left(-\frac{1}{2}\ln\left(\frac{4k}{\cpi^2}\right)\right)\right)^k\\
        &=\left(1-\frac{1}{2k}\pm\frac{\cpi\ci}{2k}+\operatorname{\tilde{O}}_k\left(k^{-2}\right)\right)^k\\
        &=\exp\left(-\frac{1}{2}\pm\frac{\cpi\ci}{2}\right)+\operatorname{\tilde{O}}_k\left(k^{-1}\right).
    \end{aligned}
\end{equation}
This gives:
\begin{equation}
    \frac{\mathbb{E}_{\bm{v}\sim\mathbb{P}_{\mathrm{par}}}\left[\left\langle g\right\rangle_{\mathrm{QAOA}}\right]}{m}\geq\frac{1}{2}+\sqrt{\frac{\ln\left(k\right)}{4\ce k\lambda}}+\operatorname{O}_k\left(k^{-\frac{1}{2}}\right).
\end{equation}

\section{Bound on \texorpdfstring{$\psi$}{Psi}}\label{sec:bound_on_psi}

Recall from Appendix~\ref{sec:max_value_gallager_xor_sat} the function $\psi\left(p,x\right)$ governing the asymptotic behavior of moments of Kravchuk polynomials~\cite{9398654}:
\begin{equation}\label{eq:kravchuk_psi_def}
    \psi\left(p,x\right):=p-1+\log_2\left(\left(1-\delta\right)^p+\delta^p\right)-\frac{p}{2}\operatorname{H}_2\left(x\right)-px\log_2\left(1-2\delta\right),
\end{equation}
with $\delta$ implicitly defined via:
\begin{equation}\label{eq:x_delta_def_app}
    x=\left(\frac{1}{2}-\delta\right)\frac{\left(1-\delta\right)^{p-1}-\delta^{p-1}}{\left(1-\delta\right)^p+\delta^p}.
\end{equation}
We here prove an upper bound on this function.
\begin{proposition}\label{prop:kravchuk_psi_bound}
    Let $\psi\left(p,x\right)$ be as in Eq.~\eqref{eq:kravchuk_psi_def}. Then for any $p\geq 3$:
    \begin{equation}
        \psi\left(p,x\right)\leq\frac{4}{\ln\left(2\right)}\left(1-2x\right)^{\frac{p-1}{2}}+\frac{p}{2}\operatorname{H}_2\left(x\right)-1.
    \end{equation}
\end{proposition}
\begin{proof}
    First, immediately from Eq.~\eqref{eq:x_delta_def_app} we have the bound:
    \begin{equation}
        x\leq\frac{\frac{1}{2}-\delta}{1-\delta},
    \end{equation}
    and in particular:
    \begin{equation}
        \delta\leq\frac{1-2x}{2\left(1-x\right)}
    \end{equation}
    and
    \begin{equation}\label{eq:delta_ratio_bound}
        \frac{\delta}{1-\delta}\leq 1-2x.
    \end{equation}
    Similarly, using the equality:
    \begin{equation}
        \frac{1-2\delta}{\left(1-\delta\right)^p+\delta^p}=\frac{2x}{\left(1-\delta\right)^{p-1}-\delta^{p-1}}
    \end{equation}
    and the inequality Eq.~\eqref{eq:delta_ratio_bound} we have:
    \begin{equation}
        \begin{aligned}
            1-x&=\frac{\left(1-\delta\right)^p+\delta^p-\left(\frac{1}{2}-\delta\right)\left(1-\delta\right)^{p-1}+\left(\frac{1}{2}-\delta\right)\delta^{p-1}}{\left(1-\delta\right)^p+\delta^p}\\
            &=\frac{\left(1-\delta\right)^{p-1}+\delta^{p-1}}{2\left(\left(1-\delta\right)^p+\delta^p\right)}\\
            &=\frac{1}{2\left(1-\delta\right)}\left(\frac{\left(1-\delta\right)^p+\left(1-\delta\right)\delta^{p-1}}{\left(1-\delta\right)^p+\delta^p}\right)\\
            &=\frac{1}{2\left(1-\delta\right)}\left(1+\frac{\delta^{p-1}\left(1-2\delta\right)}{\left(1-\delta\right)^p+\delta^p}\right)\\
            &=\frac{1}{2\left(1-\delta\right)}\left(1+\frac{2x}{\frac{\left(1-\delta\right)^{p-1}}{\delta^{p-1}}-1}\right)\\
            &\leq\frac{1}{2\left(1-\delta\right)}\left(1+\frac{2x}{\left(1-2x\right)^{-\left(p-1\right)}-1}\right).
        \end{aligned}
    \end{equation}
    This gives a general bound on $\psi$:
    \begin{equation}
        \begin{aligned}
            \psi\left(p,x\right)&\leq p-1+\log_2\left(\left(1-\delta\right)^p+\delta^p\right)-\frac{p}{2}\operatorname{H}_2\left(x\right)-px\log_2\left(\frac{x}{1-x}\right)\\
            &=\log_2\left(\left(1-\delta\right)^p+\delta^p\right)+p\log_2\left(2\left(1-x\right)\right)+\frac{p}{2}\operatorname{H}_2\left(x\right)-1\\
            &=p\log_2\left(1-\delta\right)+\log_2\left(1+\frac{\delta^p}{\left(1-\delta\right)^p}\right)+p\log_2\left(2\left(1-x\right)\right)+\frac{p}{2}\operatorname{H}_2\left(x\right)-1\\
            &\leq\log_2\left(1+\left(1-2x\right)^p\right)+p\log_2\left(1+\frac{2x}{\left(1-2x\right)^{-\left(p-1\right)}-1}\right)+\frac{p}{2}\operatorname{H}_2\left(x\right)-1.
        \end{aligned}
    \end{equation}

    We now bound the first two terms to be take more convenient forms. We use in what follows the upper bound for $x\geq 0$~\cite[Theorem~2.4]{sandor2017two}:\footnote{We use the usual identification of $\lim_{x\to 0}\frac{\sinh\left(x\right)}{x}=1$.}
    \begin{equation}
        \frac{\sinh\left(x\right)}{x}\geq\cosh\left(\frac{x}{\sqrt{3}}\right),
    \end{equation}
    along with the standard upper bounds (for $r\geq 0$):
    \begin{align}
        \ln\left(1+x\right)&\leq x,\\
        \left(1-x\right)^r&\leq\exp\left(-rx\right).
    \end{align}
    For the first term, we have the upper bound:
    \begin{equation}
        \log_2\left(1+\left(1-2x\right)^p\right)\leq\frac{\left(1-2x\right)^p}{\ln\left(2\right)}.
    \end{equation}
    For the second term, we calculate:
    \begin{equation}
        \begin{aligned}
            \log_2\left(1+\frac{2x}{\left(1-2x\right)^{-\left(p-1\right)}-1}\right)&\leq\frac{2x\left(1-2x\right)^{p-1}}{\ln\left(2\right)\left(1-\left(1-2x\right)^{p-1}\right)}\\
            &\leq\frac{2x\left(1-2x\right)^{p-1}}{\ln\left(2\right)\left(1-\exp\left(-2\left(p-1\right)x\right)\right)}\\
            &=\frac{x\exp\left(\left(p-1\right)x\right)\left(1-2x\right)^{p-1}}{\ln\left(2\right)\sinh\left(\left(p-1\right)x\right)}\\
            &\leq\frac{\exp\left(\left(p-1\right)x\right)\left(1-2x\right)^{p-1}}{\ln\left(2\right)\left(p-1\right)\cosh\left(\frac{x}{\sqrt{3}}\right)}\\
            &=\frac{\exp\left(\left(p-1\right)x\right)\left(1-2x\right)^{\frac{1}{2}\left(p-1\right)-\frac{1}{2\sqrt{3}}}\left(1-2x\right)^{\frac{1}{2}\left(p-1\right)+\frac{1}{2\sqrt{3}}}}{\ln\left(2\right)\left(p-1\right)\cosh\left(\frac{x}{\sqrt{3}}\right)}\\
            &\leq\frac{\exp\left(\frac{x}{\sqrt{3}}\right)\left(1-2x\right)^{\frac{1}{2}\left(p-1\right)+\frac{1}{2\sqrt{3}}}}{\ln\left(2\right)\left(p-1\right)\cosh\left(\frac{x}{\sqrt{3}}\right)}\\
            &\leq\frac{2}{\ln\left(2\right)\left(p-1\right)}\left(1-2x\right)^{\frac{1}{2}\left(p-1\right)+\frac{1}{2\sqrt{3}}}\\
            &\leq\frac{2}{\ln\left(2\right)\left(p-1\right)}\left(1-2x\right)^{\frac{p-1}{2}}.
        \end{aligned}
    \end{equation}
    Noting $\left(1-2x\right)^p\leq\left(1-2x\right)^{\frac{p-1}{2}}$ as $p\geq 3$ allows us to combine the two terms into:
    \begin{equation}
        \frac{\left(1-2x\right)^p}{\ln\left(2\right)}+\frac{2}{\ln\left(2\right)\left(p-1\right)}\left(1-2x\right)^{\frac{p-1}{2}}\leq\frac{\left(3p-1\right)}{\ln\left(2\right)\left(p-1\right)}\left(1-2x\right)^{\frac{p-1}{2}}.
    \end{equation}
    Furthermore, as $p\geq 3$ implies that:
    \begin{equation}
        \frac{3p-1}{p-1}\leq 4,
    \end{equation}
    we can further simplify this as:
    \begin{equation}
        \frac{\left(3p-1\right)}{\ln\left(2\right)\left(p-1\right)}\left(1-2x\right)^{\frac{p-1}{2}}\leq\frac{4}{\ln\left(2\right)}\left(1-2x\right)^{\frac{p-1}{2}}.
    \end{equation}
    This completes the proof.
\end{proof}

\section{Local Restrictability of Common Code Ensembles}\label{sec:loc_rest_common_code_ens}

The local restrictability property is not unique to the Gallager ensemble; rather, we argue that it is a relatively generic property of LDPC code ensembles.
We here provide some additional examples of restrictable LDPC ensembles.
In each case, we fix a constant $\lambda^{-1} \in (0, 1)$ and implicitly assume that $\lambda^{-1} m \in \mathbb{Z}$.  
The constant $r = 1 - \lambda^{-1} \in (0, 1)$ is known as the design rate of an ensemble.
When we refer to an ensemble of codes, we mean a collection of distributions $\mathcal{D}_{n, \lambda}$, i.e., one distribution for each $n$ and $\lambda$. 
\begin{definition}[Bernoulli ensemble] \label{def:Bernoulli_code}
    Fix $\lambda^{-1} \in (0, 1)$ a constant.
    The Bernoulli ensemble with probability $p$ is the distribution defined by the process in which each entry of the parity check matrix $\boldsymbol{H} \in \mathbb{F}_2^{\lambda^{-1} m \times m}$ is i.i.d.\ $\operatorname{Bern}(p)$. 
\end{definition}

The Bernoulli ensemble does not give uniform check weights. That is, the weight of two distinct checks may differ slightly. If we insist that all checks must have the same weight, we must instead use the right-regular ensemble, defined as follows.
\begin{definition}[Right-regular ensemble] \label{def:right_regular_code}
    Fix $\lambda^{-1} \in (0, 1)$ a constant. 
    The right-regular ensemble with check sparsity $d$ is defined by a parity check matrix $\boldsymbol{H} \in \mathbb{F}_2^{\lambda^{-1} m \times m}$ whose rows are independently generated. 
    For each row, we randomly choose $d$ indices (without replacement) on whose corresponding bit to place a $1$. 
    The bits on the remaining $m-d$ indices are set to $0$.
\end{definition}

Both of these ensembles are manifestly self-similar: since each check is generated independently from all others, restriction to the first $t$ checks is equivalent to stopping the generation process after the first $t$ steps.
Hence, to show restrictability, it suffices to show goodness for a generic parameter regime.
In the Bernoulli ensemble case, the probability parameter $p$ cannot be $O(1/n)$, as otherwise any column of the check has probability $(1 - p)^m = \Omega(1)$ of being 0, resulting in a distance-1 code.
Nevertheless, when $p = \Omega(\frac{\log m}{m})$, the Bernoulli ensemble is good.

\begin{lemma}[Binomial parity] \label{lemma:binomial_parity}
    Let $X \sim \operatorname{Bin}(n, p)$. Then $\mathbb{P}[X \equiv 0 \pmod{2}] = \frac{1}{2} + \frac{1}{2} (1 - 2p)^n$.
\end{lemma}
\begin{proof}
By direct computation with the binomial theorem, one can check that 
\begin{align}
    \mathbb{P}[X \equiv 0 \text{ (mod 2)}] - \mathbb{P}[X \equiv 1 \text{ (mod 2)}] = ((1-p) - p)^n .
\end{align}
At the same time, $\mathbb{P}[X \equiv 0 \text{ (mod 2)}] + \mathbb{P}[X \equiv 1\text{ (mod 2)}] = 1$. Adding the two expressions and simplified proves the claim.
\end{proof}

\begin{theorem}[Bernoulli ensemble is good] \label{thm:Bernoulli_ensemble_good}
Let $\lambda^{-1} = \frac{n}{m}$.
If $\operatorname{H}_2(\delta) < \lambda^{-1} < C$, then Bernoulli ensemble with probability parameter $C \frac{\ln m}{m}$ is $(\delta m, 1/n^{\Omega(1)})$-good.
\end{theorem}

\begin{proof}
Fix an error $\boldsymbol{e}$ of weight $w$.
Since each parity check is drawn independently and each check is a sum of $w$ independent Bernoulli bits with probability $p$, the probability that $\boldsymbol{e}$ has 0 syndrome is, by Lemma~\ref{lemma:binomial_parity}, \begin{align}
    \Pr[\boldsymbol{He} = 0] = \left( \frac{1 + (1-2p)^w}{2} \right)^{n} =: q_w^n .
\end{align}
We will bound the probability that $\boldsymbol{e}$ is a codeword by a union bound, i.e. \begin{align}
    \Pr[\exists \boldsymbol{e} \,:\, |\boldsymbol{e}| \leq \delta m \,\cap\, \boldsymbol{He} = 0] \leq \sum_{w=1}^{\delta m} \binom{m}{w} q_w^n .
\end{align}
We bound this sum in three regimes.
First, we consider the case when $w$ is sufficiently small.
Then \begin{align}
    q_w & \leq \frac{1 + e^{-2pw}}{2} = e^{-pw + O((pw)^2)} .
\end{align}
Choose $\epsilon$ such that $C(1 - \epsilon) > \lambda^{-1}$.
Then by the above there exists a constant $\eta$ such that for $pw \leq \eta$, \begin{align}
    \frac{1 + e^{2pw}}{2} \leq e^{-(1-\epsilon)pw}
\end{align}
and thus $q_w^n \leq \exp(-(1 - \epsilon) \lambda^{-1} C \ln(m) w) = m^{-cw}$ where $c > 1$ is a constant.
Thus, \begin{align}
    \sum_{w=1}^{\eta/p} \binom{m}{w} q_w^n & \leq \sum_{w=1}^{\eta/p} \left( \frac{em}{w} \right)^{w} m^{-cw} \leq \sum_{w=1}^{\eta/p} \left( e m^{(1-c)} \right)^w \\
    & \leq e m^{(1-c)} \sum_{w=0}^{\infty} (e m^{(1-c)})^w \leq (2 + \operatorname{o}_m(1)) e m^{(1-c)} = m^{-\Omega(1)} ,
\end{align}
where we have used the fact that $r := e m^{(1-c)} < 1/2$ for sufficiently large $m$ and thus the geometric series $\frac{1}{1-r} < 2$.
Next, we consider the regime in which $\eta \leq pw \leq K$ for some constant $K$ (independent of $m$).
In this case, $q_w \leq \frac{1 - e^{-2\eta}}{2} =: \beta_\eta < 1$ uniformly.
At the same time, \begin{align}
    \log \sum_{w \leq K/p} \binom{m}{w} & \leq \log \left( \frac{K}{p} \binom{m}{M/p} \right) \leq \frac{K}{p} \log \frac{emp}{K} + \log \frac{K}{p} = \frac{K}{p} \log \frac{e C \ln m}{K}  + \log \frac{K}{p} \\
    & = \frac{K}{p} (1 + \operatorname{o}_m(1)) \log \log m = \operatorname{O}_m \left( \frac{\log \log m}{\log m} m \right) = \operatorname{o}_m(m) .
\end{align}
Thus, this intermediate regime decays exponentially: \begin{align}
    \sum_{w : pw \in [\eta, K]} \binom{m}{w} q_w^n & = \exp_2\left(-m \left(\lambda^{-1} \log\frac{1}{\beta_\eta} - \operatorname{o}_m(1) \right) \right) = \exp(-\Theta_m(m)) ,
\end{align}
and, importantly, this holds regardless of the choice of $K$ so long as it is a constant independent of $m$.
Finally, in the large regime, we use the bound $\sum_{w \leq \delta m} \binom{m}{w} \leq \exp_2(-m \operatorname{H}_2(\delta))$~\cite{richardson2008modern}.
Thus, \begin{align}
    \sum_{w = Kp}^{\delta m} \binom{m}{w} q_w^n & \leq  \exp_2 \left( -m \left( \operatorname{H}_2(\delta) - \lambda^{-1} \log \frac{1}{\beta_K} \right) \right) .
\end{align}
This is $\exp(-\Theta_m(m))$ so long as $\operatorname{H}_2(\delta) < \lambda^{-1} \log 1/\beta_K$.
Note that as $K \to \infty$, $\log 1/\beta_K \to 1$, and we assume that $\operatorname{H}_2(\delta) < \lambda^{-1}$.
Hence, we can always choose a constant $K$ sufficiently large such that $\operatorname{H}_2(\delta) < \lambda^{-1} \log 1/\beta_K$, and in doing so, we have \begin{align}
    \Pr[\exists \boldsymbol{e} \,:\, |\boldsymbol{e}| \leq \delta m \,\cap\, \boldsymbol{He} = 0] \leq m^{-\Omega_m(1)} + \exp(-\Theta_m(m)) + \exp(-\Theta_m(m)) = m^{-\Omega_m(1)} .
\end{align}
\end{proof}

Since the sparsities are not fixed in the Bernoulli ensemble, we also bound them probabilistically. 
The log-local parameter $p = C \frac{\log m}{m}$ yields check sparsities and bit sparsities of size $\sim \log m$ with high probability.

\begin{lemma}[Bernoulli ensemble weight bounds] \label{lemma:Bernoulli_code_weight_bound}
    Let $\lambda^{-1} \in (0, 1)$ be a constant. Let $\boldsymbol{H} \in \mathbb{F}_2^{\lambda^{-1}m \times m}$ be sampled $H_{ij} \stackrel{\text{iid}}{\sim} \operatorname{Bern}(p)$ with $p = C \frac{\ln m}{m}$ and $\lambda^{-1}C \geq 4$. Then the check sparsity $W_r := \max_{i} \sum_{j} H_{ij}$ and bit sparsity $W_c = \max_j \sum_{i} H_{ij}$ (sums are as integers, \textit{not} mod 2) are at most $(1 + 2/\sqrt{C}) C \ln m$ and $(1 + 2/\sqrt{\lambda^{-1} C}) C \lambda^{-1} \ln m$, respectively, with probability $1 - 1/n^{\Omega(1)}$, where $n = \lambda^{-1} m$.
\end{lemma}
\begin{proof}
We apply a Chernoff bound for Bernoulli random variables, which states that for $X \sim \operatorname{Bin}(m, p)$, and $\delta \in [0, 1]$, $\mathbb{P}[X - mp \geq \delta  m p] \leq \exp\left(- \frac{\delta^2 mp}{3} \right)$.
Let $w^{(r)}_i = \sum_{j} H_{ij}$ be the weight of the $i$th row and $w^{(c)}_j = \sum_{i} H_{ij}$ be the weight of the $j$th column. Then $W_r = \max_i w^{(r)}_i$ and $W_c = \max_j w^{(c)}_j$. By Chernoff, \begin{align}
    \mathbb{P}[w^{(r)}_i - C \ln m \geq \delta C \ln m] \leq \exp\left(- \frac{\delta^2 C \ln m}{3} \right) . 
\end{align}
The event $\{W_r - C \ln m \geq \delta C \ln m\}$ is the same event as $\{ \exists i \,:\, w^{(r)}_i - C \ln m \geq \delta C \ln m \}$. By a union bound, \begin{align}
    \mathbb{P}[\exists i \,:\, w^{(r)}_i - C \ln m \geq \delta C \ln m] & \leq m \exp\left(- \frac{\delta^2 C \ln m}{3} \right) \leq \exp \left(- \left( \frac{\delta^2 C \ln m}{3} - \log m \right) \right) \\
    & = \exp \left(- \left( \frac{\delta^2 C}{3} - 1 \right) \ln m \right) .
\end{align}
This will be $\exp{-\Omega(\ln m)}$ if $\frac{\delta^2 C}{3} > 1$, i.e. if $\delta > \sqrt{3/C}$. So, $\delta = 2/\sqrt{C}$ suffices, and when $C \geq 4$, $\delta \leq 1$ as needed.
The proof proceeds analogously for the column weight, except that we now sum of $\lambda^{-1}m$ bits instead, so $mp \mapsto \lambda^{-1}m p = \lambda^{-1}C \ln m$.
\end{proof}

We next demonstrate that the right-regular ensemble is similarly good.
Our strategy is to utilize anti-correlation and the result of Theorem~\ref{thm:Bernoulli_ensemble_good} rather than to prove the result from scratch.

\begin{theorem}[Right-regular ensemble is good] \label{thm:right_regular_ensemble_good}
    Let $\lambda^{-1} = \frac{n}{m}$ and let $d = C \ln m$.
    If $\operatorname{H}_2(\delta) < \lambda^{-1} < C$, then the right-regular ensemble with check sparsity $d$ is $(\delta m, 1/n^{\Omega(1)})$-good.
\end{theorem}
\begin{proof}
We again split our union bounded probability into various regimes.
Let $p := d/m$ and and choose $\epsilon$ sufficiently small so that $(1 - \epsilon) C > \lambda^{-1}$.
Let $S = \operatorname{supp}(\bm{e})$ and $T$ be the support of the first check.
Then $\bm{e}$ is not detected by the first check if and only if $S \cap T \equiv 0 \pmod{2}$; denote this probability by $q_w$.
Then $q_w = 1 - \Pr[S \cap T \equiv 1 \pmod{2]} \leq 1 - \Pr[S \cap T = 1]$, and \begin{align}
    \Pr[S \cap T = 1] & = \frac{\binom{w}{1} \binom{m-w}{d-1}}{\binom{m}{w}} = \frac{wd}{m} \frac{\binom{m-w}{d-1}}{\binom{m-1}{d-1}} .
\end{align}
The binomial coefficient ratio is the probability that when we sample $d-1$ balls from a set consisting of $w - 1$ white balls and $m-w$ black balls, that all the balls sampled are black.
By a union bound, this probability is at least $1 - \frac{w-1}{m-1}(d-1) \geq 1 - 2pw$.
Suppose now that $pw \leq \eta$, where $\eta$ is a constant.
For $\eta$ sufficiently small, $pw (1 - 2pw) \geq (1 - \e) \eta$.
Thus, \begin{align}
    q_w \leq 1 - pw (1 - 2pw) \leq 1 - (1 - \epsilon) pw \leq \exp_2(- (1 - \e)pw) .
\end{align}
This is the same bound we obtained on the corresponding $q_w$ in the Bernoulli case, and thus we can apply the analysis from Theorem~\ref{thm:Bernoulli_ensemble_good} for this regime.
For the remaining intermediate and large regimes, note that if we had instead sampled $T$ with replacement rather than without replacement, we would obtain precisely the Bernoulli model.
This model differs from our present model only when there is at least one collision, which occurs with probability at most $\binom{d}{2} \frac{1}{m} = \operatorname{o}_m(1)$.
Hence, $q_w = \frac{1 - (1-2p)^w}{2} + \operatorname{o}_m(1)$.
We can then directly apply the intermediate and large regimes from Theorem~\ref{thm:Bernoulli_ensemble_good} to complete the proof.
\end{proof}

\section{DQI With a BP Decoder is Obstructed at the Branching OGP Threshold}\label{sec:dqi_bp_b_ogp}

We here show that the branching OGP~\cite{10.1002/cpa.22222}---a generalization of the multi-OGP---obstructs DQI with an inverse Lipschitz decoder with respect to a slightly different ensemble than the Gallager ensemble we consider in the main text. Here, we consider the \emph{Poisson ensemble} of \textsc{MAX-$k$-XOR-SAT}, where each constraint has i.i.d.\ multiplicity $\operatorname{Pois}\left(m/\binom{n}{k}\right)$. This ensemble is notable as it is known it exhibits a branching OGP~\cite{jones_et_al_full}.

Our strategy will be to show that stable quantum algorithms with respect to this ensemble of problem instances are \emph{overlap concentrated}, as it is known that the branching OGP obstructs overlap concentrated algorithms~\cite{10.1002/cpa.22222,jones_et_al_full}. We restate these known results here for convenience; the bound is expressed in terms of the \emph{$k$th algorithmic Parisi constant} $\operatorname{\normalfont\textsc{P}}_k^{\mathrm{ALG}}$, sometimes also called the \emph{algorithmic threshold of the $k$-spin model}. We refer the reader to \cite{el2021optimization} and \cite[Eq.~(2.7)]{alaoui2020algorithmicthresholdsmeanfield} for a definition.
\begin{definition}[Overlap concentrated algorithms~{\cite[Definition~2.1]{10.1002/cpa.22222}},~{\cite[Definition~6.2]{jones_et_al_full}}]
    A randomized algorithm $\mathcal{A}:\mathcal{X}\times\varOmega\to\mathbb{F}_2^n$ with associated probability space $\left(\varOmega,\mathbb{P}_\varOmega\right)$ is said to be $\left(\delta,\nu\right)$-overlap concentrated with respect to a joint distribution $\left(\bm{X},\bm{Y}\right)\sim\mathbb{P}_2^{\left(\kappa\right)}$ of correlated problem instances parameterized by $\kappa$ if, for all $\kappa$,
    \begin{equation}
        \mathbb{P}_{\omega\sim\mathbb{P}_\varOmega}\left[\mathbb{P}_{\left(\bm{X},\bm{Y}\right)\sim\mathbb{P}_2^{\left(\kappa\right)}}\left[\left\lvert d_{\mathrm{H}}\left(\mathcal{A}\left(\bm{X}\right),\mathcal{A}\left(\bm{Y}\right)\right)-\mathbb{E}_{\left(\bm{X},\bm{Y}\right)\sim\mathbb{P}_2^\kappa}\left[d_{\mathrm{H}}\left(\mathcal{A}\left(\bm{X}\right),\mathcal{A}\left(\bm{Y}\right)\right)\right\rvert\geq\frac{\delta}{2}n\right]\right]>\nu\right]\leq\operatorname{o}\left(1\right).
    \end{equation}
\end{definition}
\begin{theorem}[Overlap concentrated algorithms are obstructed in Poisson \textsc{MAX-$k$-XOR-SAT}~{\cite[Corollary~6.10]{jones_et_al_full}}]\label{thm:ov_conc_ob_b_ogp}
    Let $k$ be even and fix $\epsilon>0$. All $\left(\delta,\nu\right)$-overlap concentrated algorithms $\mathcal{A}$ with $\nu\leq\operatorname{o}\left(1\right)$ are such that, for sufficiently large $\lambda$, with probability at least $1-\operatorname{o}\left(1\right)$,
    \begin{equation}
        \frac{g_{\bm{X}}\left(\mathcal{A}\left(\bm{X},\omega\right)\right)}{m}\leq\mu_{\mathrm{b-OGP}}+\frac{\epsilon}{\sqrt{\lambda}}:=\frac{1}{2}+\frac{1}{2\sqrt{\lambda}}\operatorname{\normalfont\textsc{P}}_k^{\mathrm{ALG}}+\frac{\epsilon}{\sqrt{\lambda}}.
    \end{equation}
\end{theorem}
Notably, $\mu_{\mathrm{b-OGP}}$ is precisely the value achieved by the classical algorithm approximate message passing for the i.i.d.\ ensemble of \textsc{MAX-$k$-XOR-SAT}~\cite{cheairi2024algorithmicuniversalitylowdegreepolynomials}. We conjecture approximate message passing achieves this same value for the Poisson ensemble of instances as well.

We now prove that stable quantum algorithms are overlap concentrated. First, we define a version of stable quantum algorithms relative to this new ensemble. Note that instances of this ensemble are labeled by the multiplicities of the constraints $\bm{m}$ as well as the parities of the constraints $\bm{v}$.
\begin{definition}[Quantum stability relative to the Poisson ensemble]\label{def:gen_stable_qas}
    Let $\bm{\mathcal{A}}:\mathbb{F}_2^D\times\varOmega\to\mathcal{S}_n^{\mathrm{m}}$ be a quantum algorithm with associated probability space $\left(\varOmega,\mathbb{P}_\varOmega\right)$.

    $\bm{\mathcal{A}}$ is said to be \emph{$\left(f,L,p_{\mathrm{st}}\right)$-stable} relative to the Poisson ensemble of \textsc{MAX-$k$-XOR-SAT} if for all pairs of inputs:
    \begin{equation}
        \mathbb{P}_{\omega\sim\mathbb{P}_\varOmega}\left[\left\lVert\bm{\mathcal{A}}\left(\bm{m},\bm{v},\omega\right)-\bm{\mathcal{A}}\left(\bm{m'},\bm{v'},\omega\right)\right\rVert_{W_2}\leq f+L\left(\left\lVert\bm{m}-\bm{m'}\right\rVert_1+\left\lVert\bm{v}-\bm{v'}\right\rVert_1\right)\right]\geq 1-p_{\mathrm{st}}.
    \end{equation}
\end{definition}
It is easy to adapt the proof of Theorem~\ref{thm:dqi_is_lipschitz} from the main text to also prove DQI with an inverse Lipschitz decoder is stable relative to the Poisson ensemble of \textsc{MAX-$k$-XOR-SAT} instances. As DQI is ill-defined for constraint satisfaction problems where each constraint is included more than once, we assume DQI is implemented for this ensemble by ignoring any repeated constraints.
\begin{theorem}[DQI with a Lipschitz decoder is Lipschitz]\label{thm:dqi_is_poisson_lipschitz}
    Assume $\operatorname{DQI}_\ell$ is implemented using an $L$-inverse Lipschitz decoder. Then $\operatorname{DQI}_\ell$ is $\left(0,L+k,0\right)$-stable relative to the Poisson ensemble of \textsc{MAX-$k$-XOR-SAT}.
\end{theorem}
\begin{proof}
    Fix $\bm{v}$ and $\bm{v'}$, and let $\bm{d}:=\bm{v}\oplus\bm{v'}$. We then have:
    \begin{equation}
        \bm{\rho}_2\left(\bm{B},\bm{v'}\right)=\bigotimes_{i=1}^m\bm{Z}_i^{d_i}\bm{\rho}_2\left(\bm{B},\bm{v}\right)\bigotimes_{i=1}^m\bm{Z}_i^{d_i}.
    \end{equation}
    Just as in Proposition~\ref{prop:commuting_steps}, we hope now to commute $\bigotimes_{i=1}^m\bm{Z}_i^{d_i}$ through the decoding step. Let $S_{\bm{d}}\subseteq\left[m\right]$ be the support of $\bm{d}$, and let $R_{\bm{d}}\subseteq\left[n\right]$ be the support of $\mathcal{D}^{-1}\left(\bm{d}\right)$, which by the $L$-inverse Lipschitz assumption of the decoder is of cardinality at most $\left\lvert R_{\bm{d}}\right\rvert\leq L\left\lVert\bm{d}\right\rVert_1$. Commuting $\bigotimes_{i=1}^m\bm{Z}_i^{d_i}$ through the decoding operator $\bm{O}_{\bm{B}}$ implementing $\mathcal{D}$, we therefore have:
    \begin{equation}
        \bm{O}_{\bm{B}}\bigotimes_{i=1}^m\bm{Z}_i^{d_i}=\bm{V}_{R_{\bm{d}}}\bm{O}_{\bm{B}},
    \end{equation}
    for some operator $\bm{V}_{R_{\bm{d}}}$ with support only on $R_{\bm{d}}$. By construction,
    \begin{align}
        \operatorname{DQI}_\ell\left(\bm{B},\bm{v}\right)&=\bm{H}^{\otimes n}\bm{\rho}_3\left(\bm{B},\bm{v}\right)\bm{H}^{\otimes n},\\
        \operatorname{DQI}_\ell\left(\bm{B},\bm{v'}\right)&=\bm{H}^{\otimes n}\bm{V}_{R_{\bm{d}}}\bm{\rho}_3\left(\bm{B},\bm{v}\right)\bm{V}_{R_{\bm{d}}}^\dagger\bm{H}^{\otimes n}.
    \end{align}
    Now, by the locality of $\bm{V}_{R_{\bm{d}}}$ and the invariance of the Wasserstein metric under $1$-local unitary transformations,
    \begin{equation}
        \begin{aligned}
            \left\lVert\operatorname{DQI}_\ell\left(\bm{B},\bm{v}\right)-\operatorname{DQI}_\ell\left(\bm{B},\bm{v'}\right)\right\rVert_{W_2}&=\left\lVert\bm{\rho}_3\left(\bm{B},\bm{v}\right)-\bm{\rho}_3\left(\bm{B'},\bm{v'}\right)\right\rVert_{W_2}\\
            &\leq\left\lVert\bm{\rho}_3\left(\bm{B},\bm{v}\right)-\bm{\rho}_3\left(\bm{B},\bm{v'}\right)\right\rVert_{W_2}+k\left\lVert\bm{m}-\bm{m'}\right\rVert_1\\
            &=\left\lVert\bm{\rho}_3\left(\bm{B},\bm{v}\right)-\bm{V}_{R_{\bm{d}}}\bm{\rho}_3\left(\bm{B},\bm{v}\right)\bm{V}_{R_{\bm{d}}}^\dagger\right\rVert_{W_2}+k\left\lVert\bm{m}-\bm{m'}\right\rVert_1\\
            &\leq L\left\lVert\bm{d}\right\rVert_1+k\left\lVert\bm{m}-\bm{m'}\right\rVert_1.
        \end{aligned}
    \end{equation}
    Recalling $\bm{d}=\bm{v}\oplus\bm{v'}$ yields the final result.
\end{proof}

We now prove that quantum algorithms satisfying this definition of stability are \emph{overlap concentrated}.
\begin{proposition}[Stable quantum algorithms are overlap concentrated]\label{prop:stable_quant_algs_are_overlap_conc}
    Let $\bm{\mathcal{A}}\left(\bm{X},\omega\right)$ be an $\left(0,L,p_{\mathrm{st}}\right)$-stable quantum algorithm (Definition~\ref{def:gen_stable_qas}), and let $\mathcal{C}\left(\bm{X},\omega,\nu\right)$ the associated randomized classical algorithm induced by measuring $\bm{\mathcal{A}}\left(\bm{X},\omega\right)$ in the computational basis, where $\left(\varUpsilon,\mathbb{P}_\varUpsilon\right)$ with $\upsilon\sim\mathbb{P}_\varUpsilon$ is the probability space governing the randomness over the computational basis measurements. For any choice of $\beta>0$ and $\delta>0$, $\mathcal{C}$ is $\left(\delta,\nu\right)$-overlap concentrated, where:
    \begin{equation}
        \nu=p_{\mathrm{st}}+\frac{1}{\beta^2}+2\exp\left(-\frac{\delta^2 n}{4\left(4\lambda\beta^2 L^2+\beta L\delta/3\right)}\right).
    \end{equation}
\end{proposition}
\begin{proof}
    In what follows, we use $\bm{X}=\left(\bm{m},\bm{v}\right)$ and $\bm{Y}=\left(\bm{m'},\bm{v'}\right)$ as shorthand for two problem instances. As the quantum Wasserstein distance upper bounds the classical Wasserstein distance on classical states (Proposition~\ref{prop:quant_wass_prod_states}), by the contractivity of the Wasserstein metric under local channels (Proposition~\ref{prop:contractivity_w2}),
    \begin{equation}
        \sqrt{\mathbb{E}_{\upsilon\sim\mathbb{P}_\varUpsilon}\left[d_{\mathrm{H}}\left(\mathcal{C}\left(\bm{X},\omega,\nu\right),\mathcal{C}\left(\bm{Y},\omega,\nu\right)\right)^2\right]}\leq\left\lVert\mathcal{C}\left(\bm{X},\omega,\nu\right)-\mathcal{C}\left(\bm{Y},\omega,\nu\right)\right\rVert_{W_2}\leq\left\lVert\bm{\mathcal{A}}\left(\bm{X},\omega\right)-\bm{\mathcal{A}}\left(\bm{Y},\omega\right)\right\rVert_{W_2}.
    \end{equation}
    Condition on $\omega$ such that the probability-$1-p_{\mathrm{st}}$ event:
    \begin{equation}
        \left\lVert\bm{\mathcal{A}}\left(\bm{X},\omega\right)-\bm{\mathcal{A}}\left(\bm{Y},\omega\right)\right\rVert_{W_2}\leq L\left\lVert\bm{X}-\bm{Y}\right\rVert_1.
    \end{equation}
    occurs. We therefore have from Markov's inequality:
    \begin{equation}
        \begin{aligned}
            \mathbb{P}_{\upsilon\sim\mathbb{P}_\varUpsilon}\left[d_{\mathrm{H}}\left(\mathcal{C}\left(\bm{X},\omega,\nu\right),\mathcal{C}\left(\bm{Y},\omega,\nu\right)\right)>\beta L\left\lVert\bm{X}-\bm{Y}\right\rVert_1\right]&\leq\frac{\mathbb{E}_{\upsilon\sim\mathbb{P}_\varUpsilon}\left[d_{\mathrm{H}}\left(\mathcal{C}\left(\bm{X},\omega,\nu\right),\mathcal{C}\left(\bm{Y},\omega,\nu\right)\right)^2\right]}{\beta^2 L^2\left\lVert\bm{X}-\bm{Y}\right\rVert_1^2}\\
            &\leq\frac{1}{\beta^2}.
        \end{aligned}
    \end{equation}
    Now, condition on $\upsilon$ such that this event occurs. Conditioned on these two events, then, $\mathcal{C}\left(\bm{X},\omega,\nu\right)$ is a Lipschitz function of $\bm{X}$. The result then follows from McDiarmid's inequality applied to Lipschitz functions of i.i.d.\ Poisson random variables (see, e.g.,~\cite[Corollary~A.2]{jones_et_al_full}).
\end{proof}

Therefore, DQI with an inverse Lipschitz decoder is obstructed by the branching OGP for the Poisson ensemble of \textsc{MAX-$k$-XOR-SAT} instances.
\dqibppoisson*
\begin{proof}
    The result immediately follows from taking $\beta=n^{1/4}$ in Proposition~\ref{prop:stable_quant_algs_are_overlap_conc} and combining it with the stability of DQI with an inverse Lipschitz decoder (Theorem~\ref{thm:dqi_is_poisson_lipschitz}).
\end{proof}

\end{document}